\documentclass[10pt,final,letterpaper,twocolumn,twoside]{IEEEtran}
\usepackage{Steven}
\usepackage{amsmath}
\usepackage{amssymb}
\usepackage{graphicx} 
\usepackage{float}
\usepackage{algpseudocode}
\usepackage{algorithm}
\usepackage{xcolor}

\newcommand{\sizecorr}[1]{\makebox[0cm]{\phantom{$\displaystyle #1$}}}
\DeclareMathOperator*{\argmin}{arg\,min}
\DeclareMathOperator*{\argmax}{arg\,max}
\newcommand{\vol}{\mathrm{vol}}

\newcommand{\Lsrs}{\Lambda_{\rm srs}}
\newcommand{\Lpo}{\Lambda_{\rm po}}
\newcommand{\Lso}{\Lambda_{\rm so}}
\newcommand{\Lposo}{\Lambda_{\rm po \& so}}
\newcommand{\Leo}{\Lambda_{\rm eo}}
\newcommand{\Lei}{\Lambda_{\rm ei}}
\newcommand{\Lsi}{\Lambda_{\rm si}}
\newcommand{\Lpi}{\Lambda_{\rm pi}}

\newtheorem{remark}{Remark}

\usepackage{hyperref}
\hypersetup{colorlinks=false}
\hypersetup{colorlinks,citecolor=black,filecolor=black,linkcolor=black,urlcolor=black}

\title{Properties of an Aloha-like stability region}
\author{
Nan Xie,~\IEEEmembership{Member,~IEEE,}
John MacLaren Walsh,~\IEEEmembership{Member,~IEEE,} and \\
Steven Weber,~\IEEEmembership{Senior~Member,~IEEE}
\thanks{
Manuscript received August 15, 2014; revised September 21, 2015 and August 31, 2016; accepted November 06, 2016. Date of current version November 21, 2016. This work was partially funded by a grant CCF-1016588 from the National Science Foundation and a grant from the Defense Advanced Research Projects Agency (DARPA) Information Theory for Mobile Ad Hoc Networks (IT-MANET) project \#W911NF-07-1-0028. Preliminary results were presented at ISIT 2010 \cite{XieWeb2010}, ITA 2010, and ITA 2011.}%
\thanks{N.~Xie, J.~Walsh, and S.~Weber are with the Department of Electrical and Computer Engineering, Drexel University, Philadelphia, PA, 19104 USA (e-mail: nx23@drexel.edu; jwalsh@coe.drexel.edu; sweber@coe.drexel.edu). }
\thanks{Communicated by C.\ Emre Koksal, Associate Editor for Communication Networks.}
\thanks{Color versions of one or more of the figures in this paper are available online at http://ieeexplore.ieee.org.}
\thanks{This paper has been accepted for publication by IEEE Transactions on Information Theory.  The DOI is 10.1109/TIT.2016.2640302, and copyright has been transferred to IEEE.  An (early access) version of this article is available from IEEE at http://ieeexplore.ieee.org/document/7784826/.}
}

\begin{document}

\maketitle

\begin{abstract}
A well-known inner bound on the stability region of the finite-user slotted Aloha protocol is the set of all arrival rates for which there exists some choice of the contention probabilities such that the associated worst-case service rate for each user exceeds the user's arrival rate, denoted $\Lambda$.  Although testing membership in $\Lambda$ of a given arrival rate can be posed as a convex program, it is nonetheless of interest to understand the properties of this set.  In this paper we develop new results of this nature, including $i)$ an equivalence between membership in $\Lambda$ and the existence of a positive root of a given polynomial, $ii)$ a method to construct a vector of contention probabilities to stabilize any stabilizable arrival rate vector, $iii)$ the volume of $\Lambda$, $iv)$ explicit polyhedral, spherical, and ellipsoid inner and outer bounds on $\Lambda$, and $v)$ characterization of the generalized convexity properties of a natural ``excess rate'' function associated with $\Lambda$, including the convexity of the set of contention probabilities that stabilize a given arrival rate vector.
\end{abstract}


\section{Introduction} 
\label{sec:intro} 

This paper addresses membership testing and structural properties of a natural inner bound on the stability region of the finite-user slotted-time Aloha medium access control (MAC) protocol under the collision channel model, hereafter {\em the Aloha protocol} \cite{Abr1970}.  The Aloha protocol is specified by a tuple $(n,\xbf, \pbf)$ where $n \in \Nbb$ is the number of users wishing to communicate with a common base station, $\xbf \in \Rbb^n_+$ denotes the arrival rate of new packets at each user's queue (one queue per user, each queue assumed capable of holding an unlimited number of packets awaiting transmission), and $\pbf \in [0,1]^n$ denotes each user's chosen contention probability, i.e., the probability with which any user with a non-empty queue will contend for the channel.  User contention decisions are synchronized at the beginning of each time slot, and, conditioned on the queue lengths, the user contention decisions are independent across users and across time slots.  Each packet transmission requires exactly one time slot.  Under the assumed collision channel model, an attempted transmission succeeds in a given time slot if and only if it is the only attempt in that time slot.  Ternary channel feedback (success, collision, idle) from the base station to each user at the end of each time slot is assumed to be both instantaneous and error-free.  

The stability region (of a MAC protocol) is defined as the set of arrival rate vectors (with elements corresponding to exogenous arrival rates at each user's queue) such that by some appropriate choice of the parameter(s) no user will, as time tends to infinity, accumulate an infinite backlog of packets waiting to be transmitted. The stability region asks for \textit{necessary and sufficient} conditions in order for every user's queue to remain bounded.  Let $q_{i}(t)$ denote user $i$'s queue length at time $t$; queue $i$ is stable if $\lim_{L \to \infty} \lim_{t \to \infty} \Pbb (q_{i}(t) < L) = 1$, and the system is considered stable if every queue is stable. Since all the states of the underlying discrete time Markov chain (DTMC) of queue length vectors (defined on $\Zbb_{+}^{n}$) communicate, the stability of the system, or equivalently, the positive recurrence of the DTMC, amounts to the property that each queue has a non-zero probability of being empty, i.e., $\lim_{t \to \infty} \Pbb(q_{i}(t) = 0) > 0$ for all $i \in \left\{1, \ldots, n \right\}$.

There is a significant body of work that derives bounds on the stability region of the Aloha protocol (denoted $\Lambda_{\rm A}$) from a queueing-theoretic perspective, including \cite{RaoEph1988, LuoEph1999, KomMaz2013}. In contrast to this approach, in this work we develop bounds and properties for an important and natural inner bound on the Aloha stability region, namely, the set of arrival rates for which there exists a vector of contention probabilities with associated worst-case service rates component-wise exceeding each arrival rate, denoted below by $\Lambda$.  One motivation to study this inner bound $\Lambda$ is that testing membership of a candidate arrival rate vector $\xbf$ in $\Lambda$ is easier than (but nonetheless has certain challenges similar to those encountered in) testing membership in the Aloha stability region $\Lambda_{\rm A}$.  In either case one must, either implicitly or explicitly, identify $\pbf$, a vector of stabilizing contention probabilities, for which $\xbf$ can be shown to be in $\Lambda$ or $\Lambda_{\rm A}$.  The difficulty is that the set of potential controls $\pbf$ is uncountably infinite ($\pbf \in [0,1]^n$), and as such, given $\xbf$, it is not obvious whether or not such a $\pbf$ exists, i.e., whether or not $\xbf$ is stabilizable.  

Our results address this challenge in several ways.  First, we give a novel characterization of membership in $\Lambda$ in terms of whether or not a certain order-$n$ polynomial equation has a positive root.  Second, we give several equivalent formulations for $\Lambda$, each with its own advantages/interpretations.  Third, we give a means of constructing a suitable control $\pbf$ for any stabilizable rate vector $\xbf$.  Fourth, we give polyhedral, spherical, and ellipsoid inner and outer non-parametric (explicit) bounds on $\Lambda$, which constitute, variously, necessary or sufficient conditions on membership in $\Lambda$.  These explicit inner and outer bounds partially illuminate the shape and structure of $\Lambda$ as a function of $n$.  Finally, we present certain structural properties of certain functions and sets naturally associated with $\Lambda$, including the excess rate function and (an inner bound of) the set of contention probabilities that stabilize a given arrival rate vector.

The inner bound $\Lambda \subseteq \Lambda_{\rm A} \subseteq \Rbb_+^n$ studied in this paper is:
\begin{equation} 
\label{eq:Lambda}
\Lambda \equiv \left\{ \xbf \in \Rbb_{+}^{n} : \exists \pbf \in [0,1]^n :  x_i \leq p_i \prod_{j \neq i} (1-p_j), \forall i \in [n] \right\}.
\end{equation}
Here $[n] = \left\{1, \ldots, n\right\}$. The expression $p_i \prod_{j \neq i} (1-p_j)$ is the worst-case service rate for user $i$'s queue, namely the service rate assuming all users have non-empty queues and thus all users are eligible for channel contention.  In particular, user $i$'s transmission is successful in such a time slot if user $i$ elects to contend (with probability $p_i$) and each other user $j \neq i$ does not contend (each with independent probability $1-p_j$ for a non-empty queue).  Clearly $\Lambda \subseteq \Lambda_{\rm A}$, since an arrival rate that is stabilizable under the worst-case service rate is certainly stabilizable under a better service rate.  Our aim in this paper is to establish properties of and non-parametric bounds on $\Lambda$. We emphasize that we call sets such as $\Lambda$ ``parametric'' due to the observation that asserting membership in them requires explicitly or implicitly identifying another parameter, which may be viewed as auxiliary from the perspective of membership testing.

\subsection{Motivation}
\label{ssec:motivationAndproblemstatement}
To provide additional motivation for this investigation, we attempt to establish below that $i)$ in spite of its age and simplicity, Aloha is nonetheless still relevant in both the design and analysis of modern communication systems, and $ii)$ knowledge about the Aloha stability region, including in particular the stability region properties established in this paper, is important to both understanding how such systems perform, and how they should be operated.  

\noindent {\bf Relevance of Aloha to modern communication systems.}  Although by today's standards the idea of the Aloha protocol is very simple, at its inception the idea of allowing for random transmission attempts, and thereby random transmission collisions, was a revolutionary idea relative to the existing paradigm of avoiding collisions completely through scheduled resource allocation.  Many currently dominant wireless technologies do not use plain Aloha; e.g., WiFi's DCF sublayer uses carrier sense multiple access/collision avoidance (CSMA/CA), and this might lead one to believe Aloha is not relevant to modern communication systems.  A rebuttal to this view was asserted in a 2009 article \cite{Abr2009} by Norman Abramson, the inventor of Aloha, where he wrote ``Today Aloha channels are utilized in all major mobile networks and in almost all two-way satellite data networks.''  Examples include $i)$ GSM systems for sending control signals from mobile nodes to the base station using a random access channel (RACH), and $ii)$ very small aperture terminal (VSAT) satellite networks for sending channel reservation messages.  Regarding cellular, Abramson further opined in a 2012 editorial \cite{AbrSac2012} that with the increasing demand of high data rate and IP-based web traffic in developing 4G networks, a greater use of Aloha random access channels is expected, for both user packet data as well as signaling and control purposes. Regarding WiFi, Abramson wrote in that same article ``Ironically, recent chatter on the web dealing with full duplex WiFi hints at further development of WiFi in the direction of the original Aloha architecture.''

An unfortunate drawback of slotted Aloha is its low throughput (e.g., it is simple to establish that slotted Aloha with a large number of symmetric users on the collision channel can achieve a maximum throughput of $1/\erm \approx 36.8\%$).  Intuitively, this low maximum throughput appears to be due to the protocol's simplicity, specifically, the failure of users under Aloha to ``listen before they speak,'' i.e., carrier sensing.  Indeed, carrier sense multiple access (CSMA) with collision detection (CD) is the basis of the successful ethernet protocol.  The performance of CSMA in the wireless domain, however, is hampered by two key differences from the wired domain: $i)$ the half-duplex constraint, and $ii)$ the hidden and exposed terminal problems.  The former prevents each transmitter from sensing collisions while transmitting, and the latter prevents each transmitter from sensing collisions at its intended receiver or nearby nodes.  In summary, although CSMA offers certain advantages over Aloha, it faces its own performance challenges and limitations. 

In fact, underwater acoustic sensor networks (UW-ASN) \cite{AkyPom2004} are a noteworthy scenario where the Aloha protocol may outperform more sophisticated CSMA-based protocols.  A UW-ASN consists of unmanned or autonomous underwater vehicles/sensors, deployed to perform collaborative monitoring tasks over a given area, connected with acoustic links. Compared with terrestrial counterparts, underwater acoustic communications are mainly influenced by long, and highly variable, propagation delay \cite{AkyPom2004}, and moreover underwater acoustic channels are temporally and spatially variable. The peculiar characteristics of underwater acoustic channels (in particular limited bandwidth and high and variable delay) pose additional challenges to the design of suitable medium access control protocols.  

Resource sharing in a UW-ASN can be achieved by contention-free methods (static channelization) or by contention-based protocols. Contention-free methods include frequency, time, and code division multiple access (FDMA, TDMA, and CDMA, respectively).  FDMA-based approaches are vulnerable to fading \cite{SozSto2000}, not flexible (e.g., to accommodate varying transmission rates \cite{HeiSto2012}), and can be inefficient in the presence of bursty traffic.  In fact, underwater acoustic channels are doubly selective meaning their multipath profiles are both temporally long (substantial delay) and rapidly time-varying (Doppler spreads): the former entails prohibitive overhead while the latter impairs the orthogonality of frequency carriers. TDMA-based mechanisms require strict time synchronization that is ill-matched to the highly variable delay characteristics of underwater acoustic channels.  CDMA-based methods are more robust to multi-path fading than FDMA, and do not require the time synchronization of TDMA, but the hardware and computation power required are in conflict with the desire for UW-ASN nodes to be small in size, low in cost, and energy efficient.  In summary, contention-free methods are not ideal for UW-ASN networks.  Neither physical sensing (e.g., CSMA) nor virtual sensing (e.g., RTS/CTS) protocols, however, will perform well for UW-ASN networks, on account of the difficulty in carrier sensing caused by the long and variable propagation latencies.  It seems possible that the Aloha protocol may well be a suitable MAC protocol for UW-ASN, as its inherent design simplicity offers a natural performance robustness in the face of channel uncertainties.  

\noindent {\bf Importance of Aloha stability region properties.} The first reason for understanding the stability region of the Aloha protocol is the long-observed tantalizing contrast between the simplicity of the protocol itself and the (apparent) difficulty in obtaining its stability region.  Aloha is arguably the most basic of medium access protocols, and yet, in spite of extensive effort for over thirty five years by numerous researchers around the world, the stability region remains elusive.  This discrepancy makes investigation of this problem an important open question in the theory of communication systems.

The second reason is that (queue) stability is perhaps the most important property a random medium access protocol can have.  Unstable queues lead to unbounded delays, and, by extension, to system collapse.  Given a set of nodes vying for access, the first question to assess is {\em whether} or not the collection of arrival rates is stabilizable under the protocol.  If the answer is no, then the base station must intervene to reduce the arrival rates until they are stabilizable.  If the answer is yes, the second question to assess is {\em how} the system may be stabilized.  In the context of Aloha, this question is how to select the contention probabilities so as to stabilize the target arrival rates.  Finally, a third question to ask about a stabilizable (and stabilized) arrival rate vector is whether or not it is {\em throughput efficient}.  In the context of Aloha, this corresponds to selecting the arrival rates to be on or near the Pareto frontier of the stability region, i.e., the stable points not throughput dominated by any other stable point.  All three of these natural questions (stabilizability, how to stabilize, and how to find throughput efficient operating points) require knowledge of the stability region.    

In fact, many of the results we obtain have a natural application in the operation of an Aloha protocol.  If the operator wishes to know whether or not the given arrival rate vector is in $\Lambda$, then our polynomial root property may be leveraged to this end.  In addition, our inner and outer bounds may be applied.  If the rate vector is found to lie inside (outside) any of our inner (outer) bounds then the rate vector is known to be (not) stabilizable; these inner and outer bounds have the benefit of being extremely simple for checking membership.  If the operator wishes to find a suitable control (contention probability vector) for stabilizing a given arrival rate vector, then we offer two sets of relevant results.  First, our polynomial root property provides the ``critical stabilizing control,'' i.e., the contention probability vector with a corresponding service rate vector matching the given arrival rate vector.  Second, our excess rate function can be optimized to identify a control that maximizes some measure of distance (e.g., some norm) between the arrival rate and service rate vectors, over the set of controls that stabilize the given rate vector.  

Besides facilitating easy membership testing, another important value of establishing non-parametric inner and outer bounds on $\Lambda$ lies in improved geometric intuition.  Looking at the definition of $\Lambda$, it is difficult, in our opinion, to intuit its geometric properties, especially in high dimensions. Thus to gain geometric intuition of the set $\Lambda$ constitutes a key motivation of this paper: this influences both our choice of the families of bounds and, for each type of bounds, the construction of it. We choose polyhedra, spheres, and ellipsoids since they are some of the simplest geometric objects. Other geometric objects could yield better bounds, yet those objects may not be as intuitive (and/or analytically tractable), especially when it comes to higher dimensions. The construction of a bound is also attempted to be made simple and intuitive, such as the semi-symmetric displacement (on the coordinate axes) of the intersecting hyperplanes used in the polyhedral outer bound, and enforcing tangency/incidence in the construction of the ellipsoid bounds. The quality of the bounds, as measured by the volumes of them, gives insight into the extent to which $\Lambda$ can be understood to ``look like'' these various (simple, non-parametric) sets.  By analogy, there are many capacity regions in information theory characterized by the use of auxiliary random variables with an unspecified distribution; these auxiliary random variables often cloud one's ability to gain geometric intuition about these regions.

\subsection{Related work}
\label{ssec:relatedwork}
The throughput analysis of the Aloha packet system with and without slots can be found in Roberts \cite{Rob1975} and Abramson \cite{Abr1977}.  The Aloha stability region problem was posed in 1979 by Tsybakov and Mikhailov \cite{TsyMik1979} who also solved the $n=2$ and the homogeneous $n$-user case, for both of which they showed $\Lambda = \Lambda_{\rm A}$. Szpankowski \cite{Szp1994} studied this problem when $n > 2$, with result expressed in terms of the joint statistics of the queue lengths. The use of the so-called ``dominant system'' in Rao and Ephremides \cite{RaoEph1988}, as well as Luo and Ephremides \cite{LuoEph1999}, established some important bounds on the stability region. Anantharam \cite{Ana1991} showed $\Lambda=\Lambda_{\rm A}$ for a certain correlated arrival process by applying the Harris correlation inequality. Using mean field analysis, assuming each queue's evolution is independent, Bordenave et al.\ \cite{BorMcD2008} were able to show $\Lambda = \Lambda_{\rm A}$ holds asymptotically in $n$. Recently Kompalli and Mazumdar \cite{KomMaz2013} obtained bounds that are linear with respect to the users' arrival rates, based on a Foster-Lyapunov approach. To date, characterization of $\Lambda_{\rm A}$ remains open for the general $n$-user case with general arrival processes, although it's been conjectured (\cite[\S V]{RaoEph1988}, \cite[\S V Thm.\ 2]{LuoEph2006}) that $\Lambda$ coincides with the Aloha stability region $\Lambda_{\rm A}$. More recently, Subramanian and Leith \cite{SubLei2013} showed structural properties such as boundary and convexity properties of the rate region of CSMA/CA wireless local-area networks which includes Aloha and IEEE 802.11 as special cases. In a similar vein, Leith, Subramanian, and Duffy \cite{LeiSub2010} established the log-convexity of the rate region in 802.11 WLANs, which yields immediate implications for utility optimization based results to be applied to fair resource allocations. Gupta and Stolyar \cite{GupSto2012} considered a generalized model of slotted Aloha by allowing \textit{asymmetric} interference between concurrent transmission attempts on a collision channel/link and derived properties of the throughput region and its Pareto boundary (frontier) such as compactness, non-convexity, and the smoothness of the Pareto frontier.

Besides its intimate connection with the Aloha stability region $\Lambda_{\rm A}$, the set $\Lambda$ has also been featured in an information theoretic context.  Namely, in 1985 Massey and Mathys \cite{MasMat1985} proved $\Lambda$ is the capacity region of the  collision channel without feedback. In the same issue, Post \cite{Pos1985} established the convexity of the complement of $\Lambda$ in the non-negative orthant $\Rbb_{+}^{n}$.

The discussion here would be incomplete unless we mention that, instead of just analyzing the existing protocols, there exists a large body of work addressing the design of random access algorithms. The results can be grouped based on criteria such as whether it is distributed/decentralized, asynchronous, needs control messages, collision-free, etc.  For example, a recent paper by Ouyang and Teneketzis \cite{OuyTen2015} presented a common information based multiple access (CIMA) protocol that only needs local information, does not have the overhead for channel sensing, is collision free and achieves the full throughput region of the collision channel, which, compared to the polynomial back-off protocols proposed by H\aa stad, Leighton, and Rogoff \cite{HasLei1996}, has lower delay. For further pointers of this literature, we refer to the reader to the references in \cite{OuyTen2015} and \cite{GupSto2012}, which include, most notably, Jiang and Walrand \cite{JiaWal2011} and Jiang et al.\ \cite{JiaSha2010}.

\subsection{Summary of bounds on $\Lambda$}
In this paper we present a variety of inner and outer bounds on $\Lambda$, including the ``square-root-sum'' inner bound $\Lsrs$ (\S\ref{sec:volLambdaAndLsrs}, Prop.\ \ref{prop:srsib}), polyhedral inner $\Lpi$ and outer $\Lpo$ bounds (\S\ref{sec:piAndpo}, Props.\ \ref{prop:Lambdapi} and \ref{prop:Lambdapo}), spherical inner $\Lsi$ and outer $\Lso$ bounds (\S\ref{sec:siAndso}, Props.\ \ref{prop:Lambdasi} and \ref{prop:Lambdaso}), and ellipsoid inner $\Lei$ and outer $\Leo$ bounds (\S\ref{sec:eiAndeo}, Props.\ \ref{prop:Ei} and \ref{prop:Eo}).  The volumes of the aforementioned bounds as a function of the number of users, $n$, are collected in Fig.\ \ref{fig:vol} and Tables \ref{tab:voltable} and \ref{tab:voltableFractionOfSimplex}, where $\mathrm{i}$ ($\mathrm{o}$) refers to inner (outer) bound respectively, and $\mathrm{p}$, $\mathrm{s}$, $\mathrm{e}$ refers to polyhedral, spherical, ellipsoid, respectively.  We give the volumes themselves, as well as the volumes normalized by the volume of the (trivial) simplex outer bound of $1/n!$.  The volumes of $\Lsrs$, $\Lpi$, $\Lambda$ are computed exactly from closed-form expressions we derive in the paper. The (exact) volume of $\Lpo$ is obtained using the {\sf lrs} \cite{Avi2000} software. All other volumes are estimated using standard Monte-Carlo simulation.\footnote{As an aside, a very recent paper by Cousins and Vempala \cite{CouVem2016} provides results (and code) for computing the volume of a convex body defined as the intersection of an explicit set of linear inequalities and a set of ellipsoids, which allow the user to tradeoff the accuracy of the volume estimates and the computational overhead (e.g., speed).} $\Lpi^{*}$ is the optimal polyhedral inner bound among its family. For the spherical bounds $\Lsi$, $\Lso$ the center of spheres are chosen such that the induced bounds are optimal within their families (hence the notation $\Lsi^{*}$ and $\Lso^{*}$). For the ellipsoid bounds $\Lei$, $\Leo$ the center of ellipsoids are chosen by setting $c = 2$. $\Lposo$ is constructed by using $\Lpo$ and $\Lso^{*}$ in conjunction namely $\Lposo \equiv \Lpo \cap \Lso^{*}$.  It is clear from Fig.\ \ref{fig:vol} that the three inner bounds (polyhedral, spherical, ellipsoid) are tighter than are the four outer bounds.  

For the Monte-Carlo volume estimates we generate independent points over $[0,1]^{n}$ uniformly at random, and use the fraction of points that fall into the region defined by the bound as our volume estimate. As the volume of the unit box $[0,1]^{n}$ is $1$, the volume of any subset of the unit box can be viewed as the probability that a point uniformly distributed over the unit box falls into this subset, which equals the mean of a Bernoulli random variable, say $Z \sim \mathrm{Ber}(v)$, for $v$ the volume of the subset. This justifies the use of the sample mean, $\hat{v}_k = (Z_1+\cdots+Z_k)/k$, as our volume estimate. We also include confidence interval estimates in both tables, indicating the relative half-width, denoted $\delta$, in order for the probability that the true mean $v$ deviates from the sample mean $\hat{v}_k$ by a fraction of no more than $\delta$ is at least $1 - \alpha$. More precisely, let $n$ and the bound (with unknown volume $v$) be given, and let $k$ be the total number of trials for generating instances of i.i.d.\ random variables $Z \sim \mathrm{Ber}(v)$. We want to find $\delta$ such that $\Pbb\left( \hat{v}_{k} \left(1 - \delta \right) \leq v \leq \hat{v}_{k} \left(1 + \delta \right) \right) \geq 1 - \alpha$. Under a normal approximation we can derive $\delta \approx \sqrt{\frac{1-\hat{v}_k}{(k-1)\hat{v}_k}} \Phi^{-1}( 1 - \alpha/2)$, applying results from \cite[\S 9.1]{BerTsi2008}.  In our simulations we use $k = 10^{8}$ and $\alpha = 5\%$.  For those volumes estimated using Monte-Carlo, the corresponding entries in Table \ref{tab:voltable} are $\hat{v}_k$ (top) and $\delta$ (bottom), and in Table \ref{tab:voltableFractionOfSimplex} are $n! \hat{v}_k$ (top) and $\delta$ (bottom).

As will be shown (Remarks \ref{remark:LpiStarVersusLsrs}, \ref{remark:LsiStarBetterThanLpiStar} and \ref{remark:LeiLeoRecoverLsiLso}), the inner bounds are ordered by volume for all $n \geq 3$\footnote{Provided the inner bounding ellipsoid is such that its center $\cbf = c \mathbf{1}$ with $c \geq (1-n m^2)/(2(1-nm))$ where $m = m(n) \equiv \frac{1}{n} \left( 1 - \frac{1}{n} \right)^{n-1}$.}, i.e.,
\begin{equation}
\mathrm{vol}(\Lsrs) \leq \mathrm{vol}(\Lpi^{*}) \leq \mathrm{vol}(\Lsi^{*}) \leq \mathrm{vol}(\Lei) \leq \mathrm{vol}(\Lambda), ~ n \geq 3.
\end{equation}
Among the outer bounds ($\Leo$, $\Lposo$, $\Lpo$ and $\Lso^{*}$) there is no such complete ordering valid for all $n$, although $\Lposo$ outperforms both $\Lpo$ and $\Lso^{*}$ by construction, and the ellipsoid outer bound $\Leo$ outperforms the optimal spherical outer bound $\Lso^{*}$ provided the outer bounding ellipsoid is such that its center $\cbf = c \mathbf{1}$ with $c \geq 1$.

\begin{figure}[ht]
\centering
\includegraphics[width=\columnwidth]{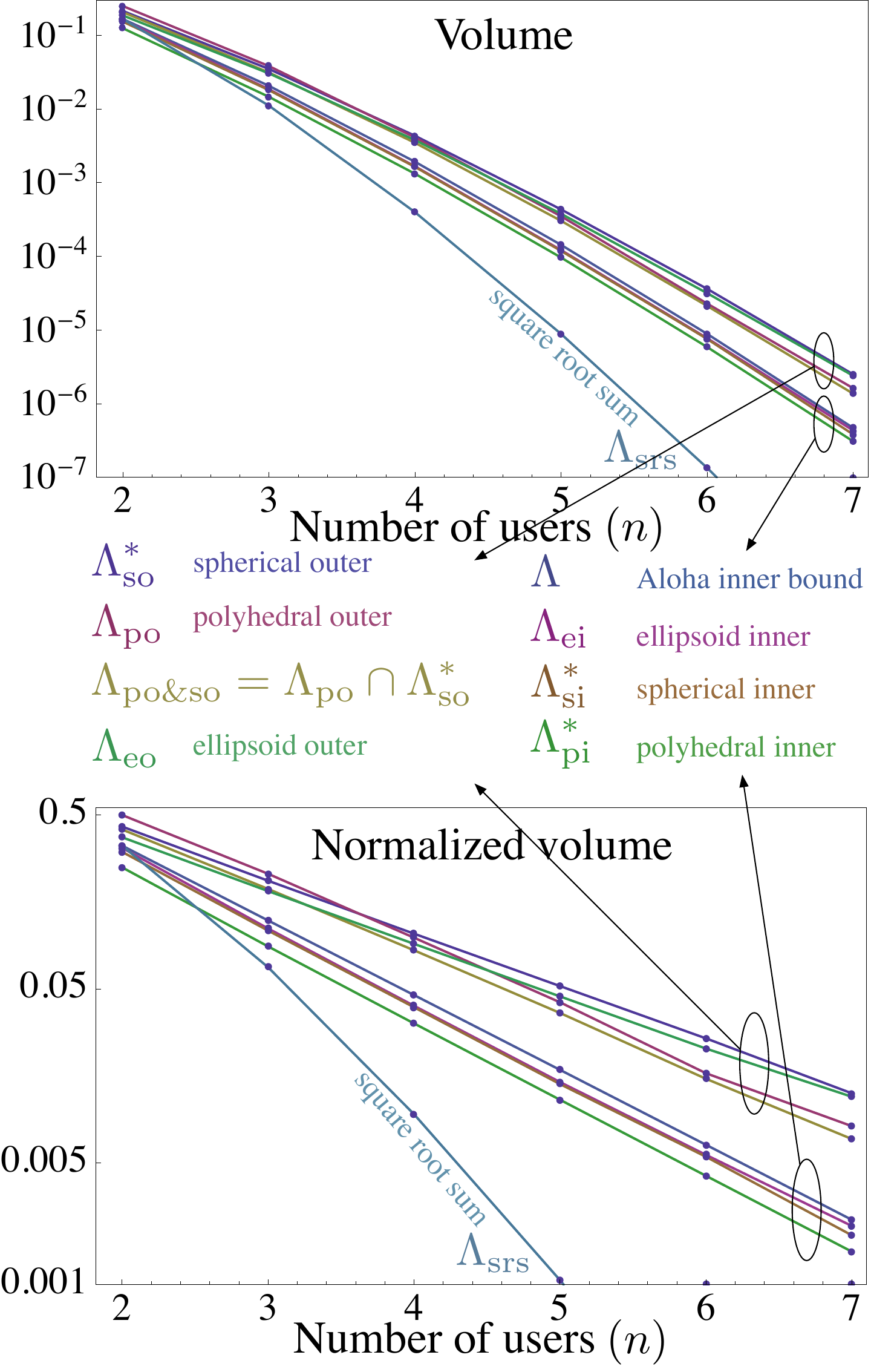}
\caption{The volumes (computed, or estimated from Monte-Carlo simulation) of the various inner and outer bounds $\Lsrs, \Lpi^{*}, \Lsi^{*}, \Lei, \Leo, \Lposo, \Lpo, \Lso^{*}$ on the Aloha stability region inner bound $\Lambda$ versus the number of users, $n$. The top figure shows the volumes and the bottom figure shows the volumes normalized by the volume of the (trivial) simplex outer bound ($1/n!$).  Each of the two ovals on each plot groups four curves, with the top oval indicating the four left labels and the bottom oval indicating the four right labels. Ellipsoids have parameter $c=2$.}
\label{fig:vol}
\end{figure}

\begin{table}
\centering
\caption{Volumes (computed or estimated) of the various bounds (estimates include normalized 95\% CI, $\delta$)}
\begin{tabular}{l|rrrrrrr}
\hline
Bounds (\S) & $n=2$ & $n=3$ & $n=4$ & $n=5$  & $n=6$ & $n=7$ \\

$\times$ & $10^{-1}$ & $10^{-2}$ & $10^{-3}$ & $10^{-4}$  & $10^{-6}$ & $10^{-7}$ \\
\hline\hline

$\Lso^{*}$  (\S \ref{sec:siAndso}) & $2.15$ & $3.49$ & $4.33$ & $4.33$ & $36.09$ & $24.90$  \\ 
& $0.00$ & $0.00$ & $0.00$ & $0.01$ & $0.03$ & $0.12$ \\ 
\cline{1-7}

$\Lpo$ (\S \ref{sec:piAndpo}) & $2.50$ & $3.82$ & $4.14$ & $3.47$ & $23.69$ & $13.65$  \\ 
\cline{1-7}

$\Lposo$ & $2.06$ & $3.13$ & $3.48$ & $3.03$ & $21.14$ & $13.60$  \\ 
& $0.00$ & $0.00$ & $0.00$ & $0.01$ & $0.04$ & $0.17$ \\ 
\cline{1-7}

$\Leo$ (\S \ref{sec:eiAndeo})  & $1.86$ & $3.04$ & $3.79$ & $3.77$ & $31.53$ & $23.80$  \\ 
& $0.00$ & $0.00$ & $0.00$ & $0.01$ & $0.04$ & $0.13$ \\ 
\cline{1-7}

$\mathbf{\Lambda}$ ~  (\S \ref{sec:volLambdaAndLsrs}) & $\mathbf{1.67}$ & $\mathbf{2.06}$ & $\mathbf{1.92}$ & $\mathbf{1.43}$ & $\mathbf{8.82}$ & $\mathbf{4.67}$  \\ 
\cline{1-7}

$\Lei$ (\S \ref{sec:eiAndeo}) & $1.62$ & $1.86$ & $1.67$ & $1.21$ & $7.72$ & $4.30$  \\ 
& $0.00$ & $0.00$ & $0.01$ & $0.02$ & $0.07$ & $0.30$ \\ 
\cline{1-7}

$\Lsi^{*}$  (\S \ref{sec:siAndso}) & $1.54$ & $1.81$ & $1.64$ & $1.18$ & $7.53$ & $3.80$  \\ 
& $0.00$ & $0.00$ & $0.01$ & $0.02$ & $0.07$ & $0.32$ \\ 
\cline{1-7}

$\Lpi^{*}$ (\S \ref{sec:piAndpo}) & $1.25$ & $1.46$ & $1.32$ & $0.96$ & $5.85$ & $3.06$  \\ 
\cline{1-7}

$\Lsrs$ (\S \ref{sec:volLambdaAndLsrs}) & $1.67$ & $1.11$ & $0.40$ & $0.09$ & $0.13$ & $0.02$  \\ 
\hline\hline
\end{tabular}
\label{tab:voltable}
\end{table}

\begin{table}
\centering
\caption{Normalized volumes (by $1/n!$) of the various bounds (estimates include normalized 95\% CI, $\delta$)}
\begin{tabular}{l|rrrrrrr}
\hline
Bounds (\S) & $n=2$ & $n=3$ & $n=4$ & $n=5$  & $n=6$ & $n=7$ \\
\hline\hline

$\Lso^{*}$  (\S \ref{sec:siAndso}) & $.429$ & $.210$ & $.104$ & $.052$ & $.026$ & $.013$  \\ 
 & $.000$ & $.001$ & $.003$ & $.009$ & $.033$ & $.124$ \\ 
\cline{1-7}

$\Lpo$ (\S \ref{sec:piAndpo}) & $.500$ & $.229$ & $.099$ & $.042$ & $.017$ & $.007$  \\ 
\cline{1-7}

$\Lposo$ & $.413$ & $.188$ & $.083$ & $.036$ & $.015$ & $.007$  \\ 
& $.004$ & $.001$ & $.003$ & $.011$ & $.043$ & $.168$ \\ 
\cline{1-7}

$\Leo$ (\S \ref{sec:eiAndeo})  & $.372$ & $.183$ & $.091$ & $.045$ & $.023$ & $.012$  \\ 
& $.000$ & $.001$ & $.003$ & $.010$ & $.035$ & $.127$ \\ 
\cline{1-7}

$\mathbf{\Lambda}$ ~  (\S \ref{sec:volLambdaAndLsrs}) & $\mathbf{.333}$ & $\mathbf{.124}$ & $\mathbf{.046}$ & $\mathbf{.017}$ & $\mathbf{.006}$ & $\mathbf{.002}$  \\ 
\cline{1-7}

$\Lei$ (\S \ref{sec:eiAndeo}) & $.324$ & $.111$ & $.040$ & $.015$ & $.006$ & $.002$  \\ 
& $.000$ & $.001$ & $.005$ & $.018$ & $.071$ & $.299$ \\ 
\cline{1-7}

$\Lsi^{*}$  (\S \ref{sec:siAndso}) & $.307$ & $.108$ & $.039$ & $.014$ & $.005$ & $.002$  \\ 
& $.001$ & $.001$ & $.005$ & $.018$ & $.071$ & $.318$ \\ 
\cline{1-7}

$\Lpi^{*}$ (\S \ref{sec:piAndpo}) & $.250$ & $.088$ & $.032$ & $.012$ & $.004$ & $.002$  \\ 
\cline{1-7}

$\Lsrs$ (\S \ref{sec:volLambdaAndLsrs}) & $.333$ & $.067$ & $.010$ & $.001$ & $.000$ & $.000$  \\ 
\hline\hline
\end{tabular}
\label{tab:voltableFractionOfSimplex}
\end{table}

\subsection{Organization and contributions}
We now describe the major sections of the paper, highlighting our main results in each section. In \S\ref{sec:rootTesting} we present a polynomial root condition for testing membership in $\Lambda$, and use this result to establish some equivalent forms of $\Lambda$.  Furthermore, the root testing can be augmented so that it allows us to exclusively find the critical stabilizing control(s).  In \S \ref{sec:volLambdaAndLsrs} we compute the volume of $\Lambda$ in closed-form, meaning it is expressed as a finite (albeit complicated) sum.  We then give a simple inner bound on $\Lambda$, exact for $n=2$, but quite weak for $n > 2$. The next three sections give explicit (non-parametric) inner and outer bounds on $\Lambda$.  Specifically, \S\ref{sec:piAndpo} gives the optimal polyhedral inner bound induced by a single hyperplane as well as a polyhedral outer bound in $\Rbb_{+}^{n}$ induced by $n+1$ hyperplanes, \S\ref{sec:siAndso} presents the optimal spherical inner and outer bounds each induced by a single sphere, and \S \ref{sec:eiAndeo} establishes ellipsoid inner and outer bounds each induced by an ellipsoid.  Our last technical section, \S \ref{sec:genconv}, shifts the focus to the generalized convexity properties of an ``excess rate'' function associated with $\Lambda$, and establishes the convexity of the set of stabilizing controls for a given rate vector assuming worst-case service rate.  A brief conclusion is given in \S\ref{sec:conclusion}, and a proof of Prop.\ \ref{prop:augumentedRootTesting} is placed in an appendix following the references.

In this paper all the vectors are column vectors and inequalities between two vectors are understood to hold component-wise. A list of general notation is given in Table \ref{tab:notationgeneral}. 
\begin{table}
\centering
\caption{General notation}
\begin{tabular}{ll}
Symbol & Meaning \\ \hline
$n$ & number of users, also the default length of a vector \\
$[n]$ & set of positive integers up to $n$ \\
$\xbf$ & vector of user's arrival rates\\
$\pbf$ & vector of user's (fixed) probabilities for channel contention\\
$\mathbf{1}$ & a vector with 1 in all its positions \\
$\ebf_{i}$ & unit vector with 1 in position $i \in [n]$ \\
$\mbf = m \mathbf{1}$ & ``all-rates-equal'' point $\mbf$ for $m \equiv \frac{1}{n} \left(1 - \frac{1}{n} \right)^{n-1}$ \\

$\pi (\pbf)$ \eqref{eq:piOfpbf} & product of $(1-p_{i})$'s \\
$\xbf (\pbf)$ \eqref{eq:xOfp} & $\xbf(\pbf)$ with components determined by $\pbf$ \\
$\pbf(\delta, \xbf)$ \eqref{eq:pOfdeltax} & $\pbf(\delta, \xbf)$ with components determined by $\xbf$ \\
&  and parameterized by $\delta > 0$ \\
\hline

$\mathrm{bd}$ & topological boundary of a set \\
$\mathrm{int}$ & interior of a set \\ 
$\mathrm{conv}$ & convex hull of a set \\ 
$\overline{A}$ & closure of set $A$ \\ 
$A^{c}$ & complement of set $A$ \\ 
$\| \cdot \| $ & $l_2$ norm \\
$d(x, y)$ & Euclidean distance between (geometric objects) $x$, $y$ \\
$\Ibb_{\rm S}$ & indicator function for boolean expression $\mathrm{S}$ \\
\hline
$\Smc$ & closed standard unit simplex \\
$\partial \Smc$ & the set of probability vectors \\
$\Hmc (\nbf, d)$ & hyperplane with normal vector $\nbf$ and displacement $d$ \\
$\Bmc(\cbf, r)$ & open ball centered at $\cbf$ with radius $r$ \\
$\Emc$ \eqref{eq:ellipsoidDef} & open ellipsoid centered at $\cbf$, in quadratic form \\
$\Emc(c, a_{1}, a_{2})$ & open ellipsoid centered at $c \mathbf{1}$ \\
& with semi-axis lengths $a_{1}$, $a_{2} = \cdots = a_{n}$ \\
\hline
$\Qbf$ (\S\ref{sec:eiAndeo}) & the rotation matrix used in Prop.\ \ref{prop:LambdaSymInRotatedSystem}, \\
& may also encode the direction of ellipsoid's axes \\
\hline
\end{tabular}
\label{tab:notationgeneral}
\end{table} 

\section{Polynomial membership testing, forms of $\Lambda$, and critical stabilizing controls} 
\label{sec:rootTesting}

This section introduces some seemingly distinct results which are presented together on account of the fact that their proofs rely upon closely related concepts.  First, Prop.\ \ref{prop:Lambdarootconditions} demonstrates that testing membership of a rate vector $\xbf$ in $\Lambda$ is equivalent to a certain polynomial equation having at least one positive root, the test of which can be performed very efficiently (Prop.\ \ref{prop:rootTestingBisectionSearch}).  Second, Prop.\ \ref{prop:Lambda2EqLambdaEqLambda4} establishes two set definitions similar to $\Lambda$ are in fact equivalent to $\Lambda$.  Finally, Prop.\ \ref{prop:augumentedRootTesting} identifies the ``critical'' stabilizing control(s) $\pbf(\xbf)$ (see Def.\ \ref{def:stabilizingcontrols}) for each $\xbf \in \Lambda$. Def.\ \ref{def:equivalentFormsLambda} gives three sets, related to $\Lambda$, that will be important for what follows.  

\begin{definition} 
\label{def:equivalentFormsLambda}
\begin{IEEEeqnarray}{rCl}
\Lambda_{\rm eq} & \equiv & \left\{ \sizecorr{x_i = p_i \prod_{j \neq i} (1-p_j),  \forall i \in [n]}  \xbf \in \Rbb_{+}^{n} : \exists \pbf \in [0,1]^n : \right. \nonumber \\
&& \qquad \qquad ~~  \left. x_i = p_i \prod_{j \neq i} (1-p_j),  \forall i \in [n] \right \} 
\label{eq:LambdaEq} \\
\Lambda_{\partial \Smc} & \equiv & \left\{ \sizecorr{x_i \leq p_i \prod_{j \neq i} (1-p_j), \forall i \in [n]} \xbf \in \Rbb_{+}^{n} : \exists \pbf \in [0,1]^n, \sum_{i}p_{i}=1 : \right. \nonumber \\
& & \qquad \qquad ~~  \left. x_i \leq p_i \prod_{j \neq i} (1-p_j), \forall i \in [n] \right\}
\label{eq:LambdaPartialS} \\
\partial \Lambda & \equiv & \left\{ \sizecorr{x_i \leq p_i \prod_{j \neq i} (1-p_j), \forall i \in [n]} \xbf \in \Rbb_{+}^{n} : \exists \pbf \in [0,1]^n, \sum_{i}p_{i}=1 : \right. \nonumber \\
& & \qquad \qquad ~~  \left. x_i = p_i \prod_{j \neq i} (1-p_j), \forall i \in [n] \right\}.
\label{eq:partialLambda}
\end{IEEEeqnarray}
\end{definition}

Comparison with $\Lambda$ in \eqref{eq:Lambda} makes clear that $\Lambda_{\rm eq}$ replaces all the inequalities in $\Lambda$ with equalities, $\Lambda_{\partial \Smc}$ adds to $\Lambda$ a restriction that the contention probabilities sum to one, and $\partial\Lambda$ adds both of these to $\Lambda$.  We denote the set of all sub-stochastic vectors as $\Smc \equiv \{ \zbf \geq \mathbf{0} : \sum_i z_i \leq 1\}$, and its facet in $\Rbb_{+}^{n}$, the set of all stochastic vectors (also called probability vectors) as $\partial \Smc \equiv \{ \zbf \geq \mathbf{0} : \sum_i z_i = 1\}$; this notation explains the label $\Lambda_{\partial \Smc}$.  The next definition introduces several quantities to be used. Let $\{\ebf_i\}_{i=1}^n$ denote the $n$ standard unit vectors in $\Rbb_+^n$.

\begin{definition} 
\label{def:fOfdeltaxAndpOfdeltaxAndxOfp}
The order-$n$ polynomial in $\delta \in \Rbb$ with coefficients determined by $\xbf \in [0,1]^{n} \setminus \{ \ebf_{i} \}_{i=1}^{n}$:
\begin{equation}
\label{eq:fdeltax}
f(\delta, \xbf)  \equiv  \prod_{i=1}^{n} (1 + x_{i} \delta) - \delta.
\end{equation}
The $n$-vector $\pbf(\delta,\xbf) \in [0,1]^{n}$ with components $p_{i} (\delta, \xbf)$ determined by $\xbf \in \Rbb_+^n$ and parameterized by $\delta > 0$:
\begin{equation} 
\label{eq:pOfdeltax}
p_{i} (\delta, \xbf) \equiv \frac{\delta x_{i}}{1 + \delta x_{i}}, i \in [n].
\end{equation}
The product of the component-wise complements of a given vector of contention probabilities $\pbf \in [0,1]^n$:
\begin{equation}
\label{eq:piOfpbf}
\pi(\pbf) \equiv \prod_{j} (1-p_{j}).
\end{equation}
The $n$-vector $\xbf(\pbf) \in [0,1]^n$ with components $x_i(\pbf)$ determined by $\pbf \in [0,1]^n$:
\begin{equation} 
\label{eq:xOfp}
x_{i} (\pbf) \equiv p_{i} \prod_{j \neq i} (1- p_{j}) = \frac{p_{i}}{1-p_{i}} \pi(\pbf), ~ i \in [n].
\end{equation}
\end{definition}

Note the equalities in $\Lambda_{\rm eq}$ in \eqref{eq:LambdaEq} are $\xbf = \xbf(\pbf)$ in \eqref{eq:xOfp}.  Our notation distinguishes between a generic vector of contention probabilities $\pbf \in [0,1]^{n}$ and a specific vector $\pbf(\delta,\xbf)$ determined by $\delta$ and $\xbf$, and likewise between a generic rate vector $\xbf \in \Lambda$ and a specific vector $\xbf(\pbf)$ determined by $\pbf$. 

\begin{definition}[stabilizability in the sense of $\Lambda$ or its equivalent forms]
\label{def:stabilizingcontrols}
A \textit{stabilizing control} for $\xbf \in \Lambda$ is a vector of contention probabilities $\pbf \in [0,1]^{n}$ that is ``compatible'' with $\xbf$, meaning the pair $(\xbf, \pbf)$ satisfies the definition of $\Lambda$. A {\em critical stabilizing control} is a stabilizing control $\pbf$ such that $(\xbf, \pbf)$ satisfies the definition of $\Lambda_{\rm eq}$ (or $\partial \Lambda$).
\end{definition}

A corollary of Prop.\ \ref{prop:Lambdarootconditions} below is that, given $\xbf \in [0,1]^n$, there exists a stabilizing control if and only if there exists a critical stabilizing control.  

Since $\Lambda \subseteq \Lambda_{\rm A}$, the non-existence of a stabilizing control for a given $\xbf$ in the sense of $\Lambda$ does not necessarily mean the non-existence of one for $\xbf$ in the sense of $\Lambda_{\rm A}$ (i.e., it does not necessarily mean $\xbf$ is not stabilizable under the Aloha protocol). Throughout this paper though, our usage of ``stabilizability'' and ``stability controls'' is tied to $\Lambda$ or its equivalent forms. 

The following proposition gives an alternative test for membership of a rate vector $\xbf$ in $\Lambda$ in terms of the existence of a positive root of the polynomial $f(\delta, \xbf)$ in \eqref{eq:fdeltax}, and furthermore establishes that in fact $\Lambda = \Lambda_{\rm eq}$.  The converse proof is constructive, meaning given a positive root $\delta$, one can construct a $\pbf(\delta,\xbf)$ compatible with $\xbf$.  In the forward direction, given $\xbf \in \Lambda$ and an associated compatible $\pbf$, we do not give an explicit expression for a positive root $\delta$ of $f(\delta, \xbf)$, although we can bound the interval containing $\delta$.  In the forward direction for $\xbf \in \Lambda_{\rm eq}$, however, given a compatible $\pbf$ such that $\xbf = \xbf(\pbf)$ as in \eqref{eq:xOfp}, we have that one of the positive roots of $f(\delta, \xbf)$ for $\xbf \in \Lambda_{\rm eq}$ will always equal $\delta = 1/\pi(\pbf)$.

\begin{proposition}[root testing] 
\label{prop:Lambdarootconditions}
Membership in $\Lambda$ (except $\{ \ebf_{i} \}_{i=1}^{n}$) is equivalent to the existence of a  positive root of the polynomial equation $f(\delta, \xbf) = 0$.
\begin{equation} 
\label{eq:Lambdarootconditions}
\xbf \in \Lambda \setminus \{ \ebf_{i} \}_{i=1}^{n} ~ \iff ~ \exists \delta > 0 ~ : ~  f(\delta, \xbf) = 0.
\end{equation}
\end{proposition}

\begin{IEEEproof}
``$\Leftarrow$'': Fix $\xbf \not\in \left\{ \ebf_{i} \right\}_{i=1}^{n}$ and suppose $\delta > 0$ satisfies $f(\delta, \xbf) = 0$.  Construct $\pbf(\delta,\xbf)$ as in \eqref{eq:pOfdeltax}, and observe the worst-case service rate for user $i$ is 
\begin{equation}
p_i(\delta, \xbf) \prod_{j \neq i} \left(1 - p_j(\delta,\xbf)\right) = \frac{\delta x_i}{\prod_j (1 + \delta x_j)}, 
\end{equation}
and for this choice of $\pbf$ the requirement $x_i \leq p_i \prod_{j \neq i} (1-p_j)$  simplifies to $\prod_j (1+\delta x_j) \leq \delta$, which is true with equality by the assumption that $f(\delta, \xbf) = 0$.  As this is true for each $i \in [n]$ it follows that $\xbf \in \Lambda$.  

``$\Rightarrow$'': First observe that if $p_{i} = 1$ for some $i \in [n]$ then the only way for $\xbf \in \Lambda$ is to let $\xbf \leq \ebf_{i}$ which means $\xbf \in \Lambda$.  Similarly, if $x_i = 0$ for some $i \in [n]$, then we can work with a reduced-dimensional $\xbf$ (i.e., the original $\xbf$ with zero component(s) removed).  Consequently, we now assume $p_{i} < 1$ and $x_{i} > 0$ for each $i \in [n]$.  Suppose $\xbf \in \Lambda \setminus \left\{ \ebf_{i} \right\}_{i=1}^{n}$ and let $\pbf$ be compatible with $\xbf$.  Define the ``inverse stability rank'' vector $\boldsymbol\Delta$ (Luo and Ephremides \cite[Thm.\ 2]{LuoEph1999}) with elements
\begin{equation}
\Delta_i = \frac{p_i}{x_i(1-p_i)}, ~ i \in [n].
\end{equation}
Then $\xbf \in \Lambda$ may be equivalently expressed in terms of $\boldsymbol\Delta$ via:
\begin{IEEEeqnarray}{rCl}
\label{eqn:keyeqPreliminaryRootTestingProof}
\xbf \in \Lambda &\Leftrightarrow& \exists \pbf : x_i \leq p_i \prod_{j \neq i}(1-p_j), i \in [n] \nonumber \\
&\Leftrightarrow& \exists \pbf : \frac{p_i}{x_i(1-p_i)} \!\geq\! \prod_{j \in [n]}\left( 1+ x_{j} \frac{p_{j}}{x_{j}\left(1-p_{j}\right)} \right), i \in [n]  \nonumber \\
&\Leftrightarrow& \exists \boldsymbol\Delta : \Delta_i \geq \prod_j (1+x_j \Delta_j),  i \in [n] ~~~ (*)\IEEEeqnarraynumspace
\end{IEEEeqnarray}
Define $\tilde{\Delta} \equiv \min_j \Delta_j$, and let $\tilde{\boldsymbol\Delta} = \tilde{\Delta} \mathbf{1}$ be the $n$-vector with all components equal to $\tilde{\Delta}$.  If $\boldsymbol\Delta$ obeys ($*$) in \eqref{eqn:keyeqPreliminaryRootTestingProof} then $\tilde{\boldsymbol\Delta}$ also obeys ($*$), because 
\begin{equation}
\tilde{\Delta} = \tilde{\Delta}_i = \min_j \Delta_j  \geq \prod_k (1+x_k \Delta_k) \geq \prod_k (1+x_k \tilde{\Delta}),  i \in [n].
\end{equation}
It follows that $f(\tilde{\Delta},\xbf) \leq 0$.  If $f(\tilde{\Delta},\xbf) = 0$ then $\tilde{\Delta}$ is the required positive root in \eqref{eq:Lambdarootconditions}.  Otherwise, notice $\lim_{\tilde{\Delta} \to \infty} f(\tilde{\Delta},\xbf) = \infty$, so by the intermediate value theorem there must exist some $\delta \in (\tilde{\Delta},\infty)$ so that $f(\delta, \xbf) = 0$.  This proves the equivalence \eqref{eq:Lambdarootconditions}, namely the membership testing of $\Lambda$ can be cast as the problem of searching for a positive root of $f(\delta, \xbf)$. 
\end{IEEEproof}

Building upon Prop.\ \ref{prop:Lambdarootconditions} including its proof ideas, we can establish the following equivalences.
\begin{proposition} 
\label{prop:Lambda2EqLambdaEqLambda4}
There exist the set equivalence relationships: $\Lambda_{\rm eq} = \Lambda = \Lambda_{\partial \Smc}$.
\end{proposition}
\begin{IEEEproof}
First we show $\Lambda_{\rm eq} = \Lambda$. Having established Prop.\ \ref{prop:Lambdarootconditions}, we only need to show a counterpart of \eqref{eq:Lambdarootconditions} for $\Lambda_{\rm eq}$, namely
\begin{equation} 
\label{eq:LambdaEQrootconditions}
\xbf \in \Lambda_{\rm eq} \setminus \{ \ebf_{i} \}_{i=1}^{n} ~ \iff ~ \exists \delta > 0 ~ : ~  f(\delta, \xbf) = 0.
\end{equation}

``$\Leftarrow$'': The same proof part used in Prop.\ \ref{prop:Lambdarootconditions} for membership testing for $\Lambda$ holds here.

``$\Rightarrow$'':  We must show that if $\xbf \in \Lambda_{\rm eq} \setminus \left\{ \ebf_{i} \right\}_{i=1}^{n}$ then there exists $\delta > 0$ such that $f(\delta, \xbf) = 0$.  But in the proof of Prop.\ \ref{prop:Lambdarootconditions} it has been shown such a $\delta$ always exists for each $\xbf \in \Lambda \setminus \left\{ \ebf_{i} \right\}_{i=1}^{n}$, and as $\Lambda_{\rm eq} \subseteq \Lambda$, a $\delta$ must likewise exist for each $\xbf \in \Lambda_{\rm eq} \setminus \left\{ \ebf_{i} \right\}_{i=1}^{n}$.  The fact that $f(1/\pi(\pbf),\xbf) = 0$ for $\pbf$ compatible with $\xbf \in \Lambda_{\rm eq} \setminus \left\{ \ebf_{i} \right\}_{i=1}^{n}$ follows by substitution.  This concludes the proof of the equivalence of root testing and membership testing for $\Lambda_{\rm eq}$ and hence establishes $\Lambda_{\rm eq} = \Lambda$. 

It then remains to show $\Lambda = \Lambda_{\partial \Smc}$.  As $\Lambda_{\partial \Smc} \subseteq \Lambda$, we only need to show $\Lambda \subseteq \Lambda_{\partial \Smc}$.  By Lem.\ 1 of \cite{MasMat1985}, 
given $\xbf \in \Lambda$ (with compatible $\pbf \in [0,1]^n$), there must exist a unique $\hat{\xbf} \in \partial \Lambda$ (with a unique compatible $\hat{\pbf} \in \partial \Smc$)  that ``dominates'' $\xbf$ in the sense that $\xbf \leq \hat{\xbf}$.  If in fact $\xbf = \hat{\xbf}$ then $\hat{\pbf}=\pbf$ as well \cite{MasMat1985}.  Since $\hat{\xbf} \in \partial \Lambda$ and $\xbf \leq \hat{\xbf}$, it follows that $\xbf \in\Lambda_{\partial \Smc}$, and thus $\Lambda \subseteq \Lambda_{\partial \Smc}$.
\end{IEEEproof}

We next present an augmented version of the root testing Prop.\ \ref{prop:Lambdarootconditions}, which makes clear how the roots of polynomial equation $f(\delta, \xbf) = 0$ map between compatible $\pbf$ and $\xbf \in \Lambda =\Lambda_{\rm eq}$. The proofs of the (critical) stabilizing controls $\pbf(\xbf)$ for a given $\xbf$ are constructive. 

\begin{proposition}[augmented root testing] 
\label{prop:augumentedRootTesting}
Fix $n \geq 2$ and let a rate vector $\xbf \in [0,1]^{n} \setminus \{ \ebf_{i} \}_{i=1}^{n}$ be given.
\begin{enumerate}
\item $\xbf \in \partial \Lambda$ if and only if there is a unique positive root $\delta$ of $f(\delta, \xbf) = 0$, denoted $\delta_{u}$. Furthermore given $\xbf \in \partial \Lambda$, then $\pbf_{u} = \pbf(\delta_{u}, \xbf)$ given by \eqref{eq:pOfdeltax} stabilizes $\xbf$.  Finally, $\pbf_{u} \in \partial \Smc$ and is the only (critical) stabilizing control for $\xbf$ among all $\pbf \in [0,1]^{n}$.

\item Let $\xbf \in \Lambda \setminus \partial \Lambda$ be given. Solving $f(\delta, \xbf) = 0$ on $(0, \infty)$ for $\delta$ yields exactly two positive roots denoted $\delta_{s}$, $\delta_{l}$.  Each root can be used to construct a vector of contention probabilities, $\pbf_{s} = \pbf(\delta_{s}, \xbf)$, $\pbf_{l} = \pbf(\delta_{l}, \xbf)$, according to \eqref{eq:pOfdeltax}, that stabilizes $\xbf$. Furthermore, $\pbf_{s}$ is such that $\sum_{i=1}^{n} p_{s,i} < 1$ (i.e., $\pbf_{s} \in \Smc \setminus \partial \Smc$) and $\pbf_{l}$ is such that $\sum_{i=1}^{n} p_{l,i} > 1$ (i.e., $\pbf_{l} \in [0,1]^{n} \setminus \Smc$).  Finally, $\pbf_{s}$, $\pbf_{l}$ are also the only two critical stabilizing controls for $\xbf$ among all $\pbf \in [0,1]^{n}$.
\end{enumerate}
\end{proposition}
\begin{IEEEproof}
See \S \ref{ssec:augmentedRootTestingProof} in the Appendix.
\end{IEEEproof}

\begin{corollary} \label{cor:bijectionbetweenSandLambda}
There exist the following bijections: $i)$ $\partial \Smc \leftrightarrow \partial \Lambda$, $ii)$ $\Smc \setminus \partial \Smc \leftrightarrow \Lambda \setminus \partial \Lambda$, $iii)$ $[0,1]^{n} \setminus \Smc \leftrightarrow \Lambda \setminus \partial \Lambda$, $iv)$ $\Smc \setminus \partial \Smc \leftrightarrow [0,1]^{n} \setminus \Smc$, $v)$ $\Smc \leftrightarrow  \Lambda$.
\end{corollary}
\begin{IEEEproof}
Massey and Mathys \cite{MasMat1985} showed $i)$. We now show $ii)$. From (the proof of) Prop.\ \ref{prop:augumentedRootTesting} there exists a function that maps from $\Lambda \setminus \partial \Lambda$ to $\Smc \setminus \partial \Smc$. We need to show this function mapping is one-to-one and onto. First, given two distinct points $\xbf, \ybf \in \Lambda \setminus \partial \Lambda$, the function maps to $\pbf_{s,x}$, $\pbf_{s,y}$ respectively, both in $\Smc \setminus \partial \Smc$. If $\pbf_{s,x} = \pbf_{s,y}$, since they are both critical stabilizing controls (according to Prop.\ \ref{prop:augumentedRootTesting}) meaning they determine the corresponding rate vectors $\xbf = \xbf(\pbf_{s,x})$, $\ybf = \ybf(\pbf_{s,y})$ according to \eqref{eq:xOfp}, this gives $\xbf = \ybf$, which contradicts the assumption $\xbf \neq \ybf$ and hence this function is one-to-one. Second, for any point $\pbf_{s} \in \Smc \setminus \partial \Smc$ it defines a rate vector $\xbf(\pbf_{s})$ according to \eqref{eq:xOfp}, which by definition is in $\Lambda_{\rm eq} = \Lambda$ and in fact is in $\Lambda \setminus \partial \Lambda$ (because of the bijection $i)$ \cite{MasMat1985}). Recall $\pbf_{s}$ is automatically a critical stabilizing control for $\xbf(\pbf_{s})$. That this function has to map $\xbf(\pbf_{s})$ back to $\pbf_{s}$ is due to the fact that a rate vector from $\Lambda \setminus \partial \Lambda$ has exactly two critical stabilizing controls (one in $\Smc \setminus \partial \Smc$, the other in $[0,1]^{n} \setminus \Smc$), as shown at the end of the proof of Prop.\ \ref{prop:augumentedRootTesting}. Therefore this function is onto. Thus we have shown the bijection $ii)$. The proof of $iii)$ is similar to that of $ii)$ and is omitted. The proof of $iv)$ follows from $ii)$ and $iii)$ due to transitivity. Finally $i)$ and $ii)$ together give $v)$.
\end{IEEEproof}

Fig.\ \ref{fig:RootTest} illustrates the three membership possibilities ($\xbf \in \Lambda \setminus \partial\Lambda$, $\xbf \in \partial \Lambda$, $\xbf \not\in \Lambda$) and the corresponding polynomials $f(\delta,\xbf)$ for the case $n=2$.  The case $n=2$ is the only (known) value of $n$ for which $\Lambda$ can be expressed explicitly (\cite{TsyMik1979}, \cite{RaoEph1988}, \cite{MasMat1985}), i.e., $\Lambda = \Lambda_{\rm A} = \left\{\xbf \in \mathbb{R}_{+}^{2}: \sqrt{x_{1}} + \sqrt{x_{2}} \leq 1 \right\}$.

\begin{figure}[!ht]
\centering
\includegraphics[width=0.49\textwidth]{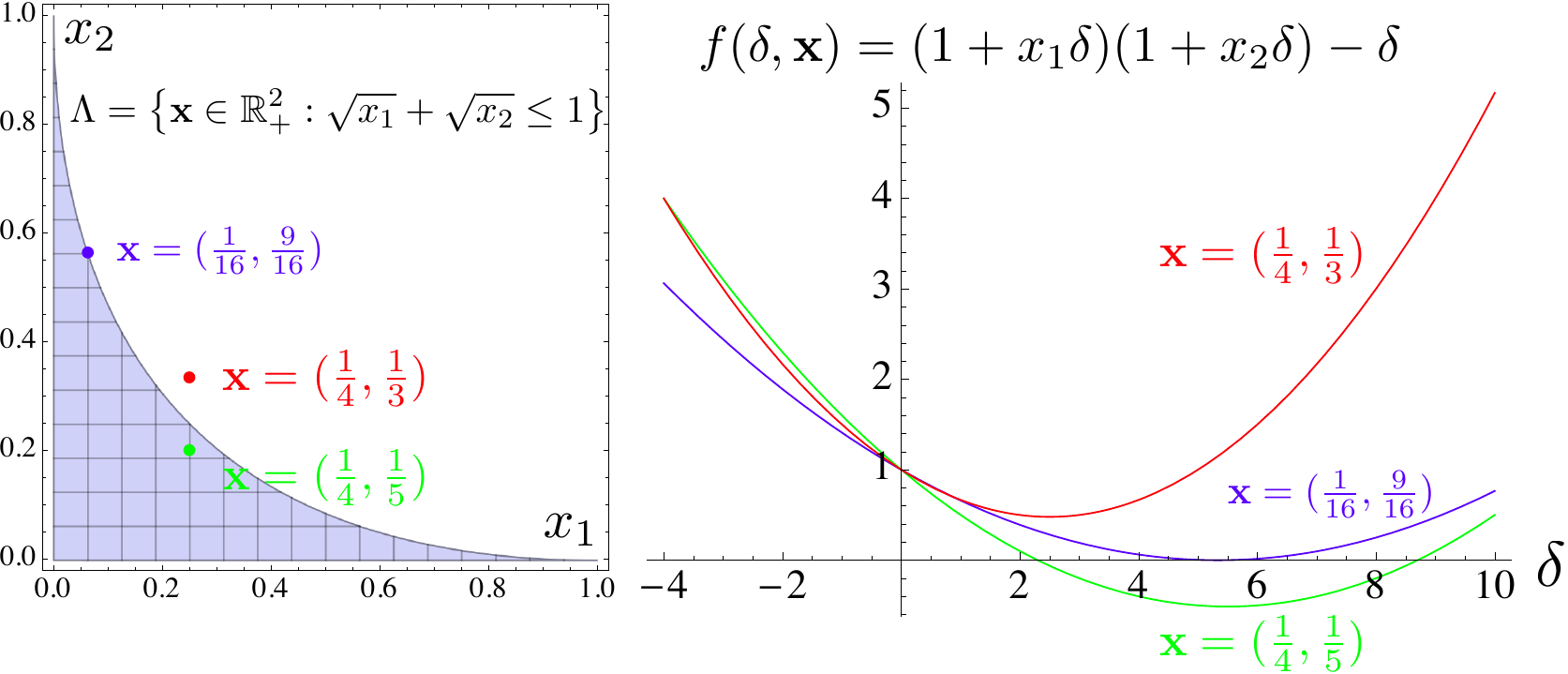}
\caption{Polynomial root membership testing when $n=2$. The green curve corresponds to the interior point $\xbf = (1/4, 1/5) \in \Lambda \setminus \partial \Lambda$, and has two positive roots: $\left(  (11-\sqrt{41})/2,  (11+\sqrt{41})/2 \right)  \approx (2.2984, 8.7016)$; the blue curve corresponds to the boundary point $\xbf = (1/16, 9/16) \in \partial \Lambda$ and has a unique positive root $16/3 \approx 5.3333$; the red curve corresponds to a point $\xbf = (1/4, 1/3) \not\in \Lambda$ and hence does not have any positive root.}
\label{fig:RootTest}
\end{figure} 

A natural concern is that root finding for an order-$n$ polynomial (especially for large $n$) could be non-trivial. The following proposition shows there is no such difficulty, since $f(\delta, \xbf)$ is convex in $\delta$. 

\begin{proposition}
\label{prop:rootTestingBisectionSearch}
The root testing (Prop.\ \ref{prop:Lambdarootconditions}) can be performed using Alg.\ \ref{alg:RootTesting}, which essentially requires simple bisection search within the interval $[1, 1/\min_{s \in [n]} x_{s}^{2}]$.
\end{proposition}
\begin{IEEEproof}
Simple algebra yields the first and second derivatives (w.r.t.\ $\delta$) of $f(\delta, \xbf)$
\begin{eqnarray} 
\label{eq:fdeltaxFirstAndSecondDerivatives}
\frac{\drm f}{\drm \delta} & = & \sum_{j=1}^{n} x_{j} \prod_{i \neq j} (1 + x_{i} \delta) - 1 \nonumber \\
\frac{\drm^{2} f}{\drm \delta^{2}} & = & \sum_{j=1}^{n} x_{j} \sum_{k \neq j} x_{k} \prod_{i \neq k, j} (1 + x_{i} \delta).
\end{eqnarray}
Observe for $\delta \in (0, \infty)$, $\frac{\drm^{2} f}{\drm \delta^{2}} \geq 0$ meaning $f(\delta, \xbf)$ is convex in $\delta$ and $\frac{\drm f}{\drm \delta}$ is increasing in $\delta$. Also, it is easy to prove by contradiction that any positive root(s) $f(\delta, \xbf)$ may have must be no less than $1$, and in fact, $f(\delta, \xbf)$ is always positive for $\delta \in (0,1)$ while at the boundary $f(0, \xbf) = 1$, $f(1, \xbf) = \prod_{i} (1+x_{i}) - 1$. We rewrite $\frac{\drm f}{\drm \delta}$ as
\begin{equation} 
\label{eq:fdeltaxFirstDerivativeRewritten}
\frac{\drm f}{\drm \delta} = \sum_{j = 1}^{n} \frac{x_{j}}{1 + x_{j} \delta} \prod_{i=1}^{n} (1 + x_{i} \delta) - 1,
\end{equation}
and compute $\left. \frac{\drm f}{\drm \delta} \right \vert_{\delta = 0} = \sum_{j} x_{j} - 1 < 0$ (for otherwise, unless $\xbf \in \{\ebf_{i}\}_{i = 1}^{n}$, we can immediate assert $\xbf \notin \Lambda$), and 
\begin{equation} 
\label{eq:fdeltaxFirstDerivativeEvaluAtOne}
\left. \frac{\drm f}{\drm \delta} \right \vert_{\delta = 1} = \sum_{j = 1}^{n} \frac{x_{j}}{1 + x_{j}} \prod_{i=1}^{n} (1 + x_{i}) - 1.
\end{equation}
Depending on the sign of the RHS of \eqref{eq:fdeltaxFirstDerivativeEvaluAtOne}, there are three cases. 

Case $1)$:$\left. \frac{\drm f}{\drm \delta} \right \vert_{\delta = 1} > 0$. Since $\frac{\drm f}{\drm \delta}$ is increasing (and continuous), its only real root, corresponding to the only stationary point (i.e., the global minimizer) of $f(\delta, \xbf)$, must lie in the interval $(0, 1)$. Recall from the above discussion that $f(\delta, \xbf) > 0$ for $\delta \in (0, 1)$ and is convex for $\delta \in (0, \infty)$, we can therefore conclude $f(\delta, \xbf)$ does not have any positive root.  

Case $2)$: $\left. \frac{\drm f}{\drm \delta} \right \vert_{\delta = 1} = 0$. In this case, $\delta = 1$ is the only real root of $\frac{\drm f}{\drm \delta}$, which is also the global minimizer of $f(\delta, \xbf)$, since $f(1, \xbf) \geq 0$ with equality only in the trivial case $\xbf = \mathbf{0}$.  We also find $f(\delta, \xbf)$ does not have any positive root (unless $\xbf = \mathbf{0}$).  

Case $3)$: $\left. \frac{\drm f}{\drm \delta} \right \vert_{\delta = 1} < 0$. In this case, the only real root of $\frac{\drm f}{\drm \delta}$ lies in $(1, \infty)$. This is the only case when $f(\delta, \xbf)$ could possibly have two positive roots. We now describe the root finding in more detail. The first step is to use bisection search on $[1, \delta_{M'}]$ to find the root (denoted $\delta^{*}$) of $\frac{\drm f}{\drm \delta}$, for $\delta_{M'}$ given below. Observe $f(\delta, \xbf)$ has two, one, and zero positive root(s), for $f(\delta^{*}, \xbf)$ less than, equal to, and greater than zero, respectively. Then, for the case when $f(\delta, \xbf)$ has two positive roots, the smaller ($\delta_{s}$) and larger ($\delta_{l}$) roots can be found by bisection search on $[1, \delta^{*}]$ and $[\delta^{*}, \delta_{M}]$ respectively, for $\delta_{M} \geq \delta_{M'}$ to be chosen. We claim it suffices to choose $\delta_{M} = \delta_{M'} = 1 / \min_{s \in [n]} x_{s}^{2}$. To see this, assuming w.l.o.g.\ $n \geq 2$, we have from \eqref{eq:fdeltaxFirstDerivativeRewritten}
\begin{IEEEeqnarray}{rCl}
\left. \frac{\drm f}{\drm \delta} \right \vert_{\delta = \delta_{M'}} & \geq & \sum_{j = 1}^{n} \frac{\min_{s} x_{s}}{1 + \min_{s} x_{s} \delta_{M'}} \prod_{i=1}^{n} (1 + \min_{s} x_{s} \delta_{M'}) - 1 \nonumber \\
& = & n \min_{s} x_{s} (1 + \min_{s} x_{s} \delta_{M'})^{n-1} - 1 \nonumber \\
& > & n \left( 1 / \min_{s} x_{s} \right)^{n-2} - 1 > 0,
\end{IEEEeqnarray}
which justifies the claim that $\delta_{M'} = 1 / \min_{s \in [n]} x_{s}^{2}$ can serve as an upper bound for the bisection search for $\delta^{*}$ (recall $\frac{\drm f}{\drm \delta} $ is increasing in $\delta$). Similarly we can verify from the definition of $f(\delta, \xbf)$ in \eqref{eq:fdeltax}
\begin{eqnarray} 
\left. f(\delta, \xbf) \right \vert_{\delta = \delta_{M}} & \geq & (1 + \min_{s} x_{s} \delta_{M})^{n} - \delta_{M} \nonumber \\
& > & (1 / \min_{s} x_{s})^{n} - 1 / \min_{s} x_{s}^{2} \geq 0,
\end{eqnarray}
which justifies the claim that $\delta_{M} = 1 / \min_{s \in [n]} x_{s}^{2}$ can serve as an upper bound for the bisection search for the larger root of $f(\delta, \xbf)$ (since $f(\delta, \xbf)$ is increasing for $\delta \geq \delta^{*}$). 
\end{IEEEproof}

The above discussion is distilled into Alg.\ \ref{alg:RootTesting} which shows how root testing can be implemented. Fig.\ \ref{fig:RootTestingBisectionSearch} illustrates the scenario when $f(\delta, \xbf)$ has two positive roots for case $3)$. First, observe that it suffices to just look at the sign of $f(\delta^{*}, \xbf)$ if one only wants to know whether $\xbf \in \Lambda$; further using bisection search for finding the positive roots of $f(\delta, \xbf)$ has the added benefits of constructing critical stabilizing controls (via \eqref{eq:pOfdeltax}). Second, observe that part of Prop.\ \ref{prop:augumentedRootTesting} also shows (albeit in a more complicated manner) $f(\delta, \xbf)$ can have no more than two positive roots, yet the above results do not completely supersede Prop.\ \ref{prop:augumentedRootTesting} since the latter provides more information and in particular it establishes further connections between critical stabilizing control(s) and the positive root(s) of $f(\delta, \xbf)$.

\begin{algorithm}
\caption{Root testing for $\xbf \in [0,1]^{n} \setminus \{\mathbf{0}, \ebf_{1}, \ldots, \ebf_{n} \}$}
\label{alg:RootTesting}
\begin{algorithmic}[1]
\If{ $\left. \frac{\drm f}{\drm \delta} \right \vert_{\delta = 1}  = \sum_{j = 1}^{n} \frac{x_{j}}{1 + x_{j}} \prod_{i=1}^{n} (1 + x_{i}) - 1 \geq 0$}
\State
\Return ``$\xbf \notin \Lambda$''
\Else
\State $\delta_{M}, \delta_{M'} \gets 1 / \min_{s \in [n]} x_{s}^{2}$
\State Bisection search for root $\delta^{*} \in [1, \delta_{M'}]$ of $\frac{\drm f}{\drm \delta}$ \eqref{eq:fdeltaxFirstAndSecondDerivatives}
\If{$f(\delta^{*}, \xbf) > 0$}
	\State \Return ``$\xbf \notin \Lambda$''
\ElsIf{$f(\delta^{*}, \xbf) = 0$}
	\State \Return ``Unique root $\delta^{*}$''
\Else
	\State Bisection search for root $\delta_{s}\! \in\! [1, \delta^{*}]$ of $f(\delta, \xbf)$\! \eqref{eq:fdeltax}
	\State Bisection search for root $\delta_{l}\! \in\! [\delta^{*}, \delta_{M}]$ \!of $f(\delta, \xbf)$\! \eqref{eq:fdeltax}
\EndIf
\EndIf
\end{algorithmic}
\end{algorithm}

\begin{figure}[!ht]
\centering
\includegraphics[width=0.35\textwidth]{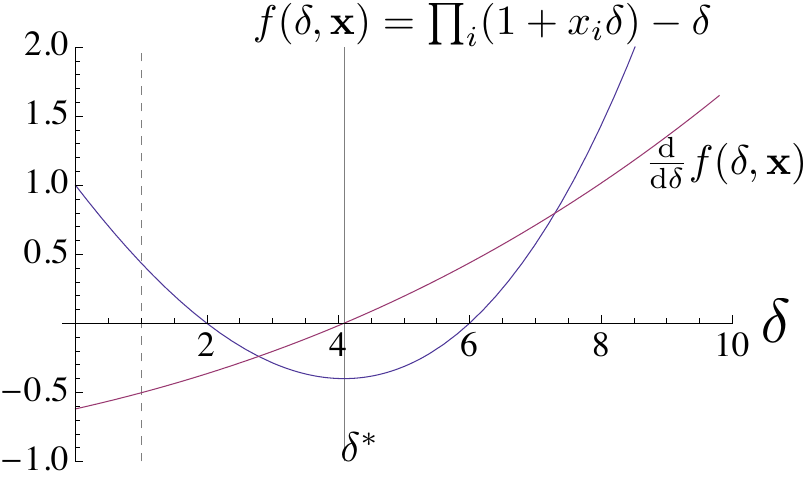}
\caption{Bisection search can be used for finding the positive root(s) of $\frac{\drm f}{\drm \delta}$ (red) and $f(\delta, \xbf)$ (blue). Here $n = 4$, $\xbf = (\frac{1}{8}, \frac{1}{10}, \frac{1}{12}, \frac{1}{14})$, and one may choose $\delta_{M} = \delta_{M'} = 196$. Any positive root this order-$4$ polynomial $f(\delta, \xbf)$ may have must lie between $\delta = 1$ (dashed gridline) and $\delta_{M}$ (not shown). First, the root ($\delta^{*} \approx 4.0890$) of $\frac{\drm f}{\drm \delta}$ is found on $[1, \delta_{M'}]$; since $f(\delta^{*}, \xbf) \approx -0.4008 < 0$ (which means $\xbf \in \Lambda$), then the smaller ($\delta_{s} = 2$) and larger ($\delta_{l} = 6$) positive roots of $f(\delta, \xbf)$ are found by bisection search on $[1, \delta^{*}]$ and $[\delta^{*}, \delta_{M}]$ respectively.}
\label{fig:RootTestingBisectionSearch}
\end{figure} 

\begin{remark}
In our example of underwater acoustic sensor networks (\S\ref{ssec:motivationAndproblemstatement}), sensors form clusters with a target sink, and may employ the slotted Aloha protocol for sending data to the sink.  Initially, each sensor must send to the sink its initial requested data transmission rate (determined by its rate of sensor data generation).  The sink may then perform the augmented root test (Prop.\ \ref{prop:augumentedRootTesting}).  If a positive root $\delta$ can be found (adjusting some proposed arrival rates if necessary), it will be used to compute a ``critical'' stabilizing control vector $\pbf$ to be sent back to the sensors. Changes in the arrival rate vector, e.g., due to an environmental change that affects the data generation rates, or the arrival or departure of one of the nodes in the cluster, will necessitate a new root test.
\end{remark}

As shown in Prop.\ \ref{prop:rootTestingBisectionSearch}, the polynomial root test is not hard in our setting, although arguably in extremely time-sensitive and/or energy-constrained scenarios (that may arise in a sensor network), one might prefer to use (approximate) membership testing based on simple non-parametric bounds (\S\ref{sec:piAndpo}, \S\ref{sec:siAndso}, and \S\ref{sec:eiAndeo}). Recall an important motivation for our investigation is to understand the set $\Lambda$ from a geometric perspective. Our attempts along this line are manifested in the next few sections where three sets of non-parametric bounds (all based on geometrically intuitive objects) are presented. In fact, in many cases we can show a family of bounds and if so we will optimize within this family. The optimization has the volume of a bound as the figure of merit. That is, the closer the volume of a bound to that of $\Lambda$, the better this bound is. Toward this end, the volume of $\Lambda$ itself, derived in the next section, is an indispensable result.

\section{Volume of $\Lambda$ and an inner bound on $\Lambda$}
\label{sec:volLambdaAndLsrs}

We first give a closed-form expression for the volume of $\Lambda$. Unfortunately its computation is a formidable task.

\begin{proposition}
\label{prop:volLambda}
The set $\Lambda$ defined in \eqref{eq:Lambda} has volume 
\begin{IEEEeqnarray}{rCl}
\IEEEeqnarraymulticol{3}{l}
{\vol(\Lambda) = 
}\nonumber \\* ~
& & \sum_{\kbf \in \Kmc_{2^n,n-2}}\!\!\!  \binom{n-2}{\kbf} (-1)^{\sum_{i=1}^{n} \alpha(\kbf)_{i}}  \frac{\prod_{i=1}^n \alpha(\kbf)_i !}{\left(n+1+\sum_{i=1}^n \alpha(\kbf)_i \right)!},\IEEEeqnarraynumspace
\label{eq:volLambda}
\end{IEEEeqnarray}
where $\binom{n-2}{\kbf}$ is a multinomial coefficient and $\Kmc_{r,s}$ $\equiv$ $\left\{ \kbf = (k_{1}, \ldots, k_{r}) \in \Zbb_{+}^{r}: \sum_{t=1}^{r} k_{t} = s \right\}$.  Furthermore, $\alpha(\kbf)_{i} = \sum_{t=1}^{2^n} \Vbf_{i,t} k_{t}$ for $\Vbf$ the $n \times 2^{n}$ matrix whose columns are the $2^n$ possible distinct length-$n$ binary vectors.\footnote{C.f., the binary Hamming matrix in coding theory.}
\end{proposition}

\begin{IEEEproof}
Recall there is a bijection from $\Smc$ to $\Lambda$ (Cor.\ \ref{cor:bijectionbetweenSandLambda}). Let $\tilde{\Jbf}(\pbf) \equiv \Jbf(\pbf)/\pi(\pbf)$, where $\Jbf(\pbf)$ is the Jacobian of this mapping, namely the mapping $\pbf \mapsto \xbf$ given by \eqref{eq:xOfp} in Def.\ \ref{def:fOfdeltaxAndpOfdeltaxAndxOfp}:
\begin{equation}
x_i(\pbf) = p_i \prod_{j \neq i}(1-p_j), ~ \text{for all } i \in [n], ~ \pbf \in \Smc.
\end{equation}
The fact that $\det(\alpha A) = \alpha^n \det A$ for any scalar $\alpha$ and any $n \times n$ matrix $A$ yields $\det \Jbf(\pbf) = \pi(\pbf)^n \det \tilde{\Jbf}(\pbf)$.  Abramson \cite{Abr1977} showed that $\pi(\pbf)^2 \det \tilde{\Jbf}(\pbf) = 1 -  \pbf^{\Tsf} \mathbf{1}$, which gives $\det \Jbf(\pbf) = \pi(\pbf)^{n-2} \left(1 - \pbf^{\Tsf} \mathbf{1} \right)$.  Substituting this into the general expression for volume yields
\begin{equation} 
\label{eq:volLambdaBasedOnAbr1977AndGenericExpression}
\mathrm{vol}(\Lambda)\! =\!  \int_{\Smc} \det \Jbf(\pbf) \drm \pbf = \int_{\Smc} \prod_{i=1}^{n} (1-p_i)^{n-2} \left(1-\sum_{j=1}^{n} p_j \right) \drm \pbf.
\end{equation}
In order to get a better closed-form expression, we leverage results in Grundmann and M\"{o}ller \cite{GruMol1978}, in particular (2.3) on integration of certain functions over the solid standard unit simplex $\Smc$: 
\begin{equation}
\label{eq:GruMol1978}
\int_{\Smc} \pbf^{\boldsymbol\alpha} \left(1 - \sum_i p_i \right)^{\alpha_0} \drm \pbf = \frac{\prod_{i=0}^n \alpha_i !}{\left(n+\sum_{i=0}^n \alpha_i \right)!},
\end{equation}
where $\pbf = (p_1,\ldots,p_n)$, $\boldsymbol\alpha = (\alpha_1,\ldots,\alpha_n)$, and $\pbf^{\boldsymbol\alpha} = \prod_{i=1}^n p_i^{\alpha_i}$. 
To apply this general expression to our case \eqref{eq:volLambdaBasedOnAbr1977AndGenericExpression}, we want to put $\prod_{i=1}^{n} (1-p_i)^{n-2}$ into a weighted sum of terms of the form $\pbf^{\boldsymbol\alpha}$.  The multi-binomial theorem states, for arbitrary $n$-vectors $\abf,\bbf$, and positive $n$-vector $\cbf$:
\begin{IEEEeqnarray}{rCl}
\IEEEeqnarraymulticol{3}{l}
{\prod_{i=1}^n (a_i + y_i)^{c_i} =
}\nonumber \\* \quad
\sum_{k_1=0}^{c_1} \cdots \sum_{k_n=0}^{c_n} \binom{c_1}{k_1} a_1^{c_1-k_1} y_1^{k_1} \cdots \binom{c_n}{k_n} a_n^{c_n-k_n} y_n^{k_n}.\IEEEeqnarraynumspace
\end{IEEEeqnarray}
Specializing the above expression to the case $\abf = \mathbf{1}$ and $\cbf = \mathbf{1}$ and arbitrary $n$-vector $\ybf$ yields:
\begin{IEEEeqnarray}{rCl}
\prod_{i=1}^{n} (1+y_{i})^{1} & = & \sum_{k_{1} = 0}^{1} \cdots \sum_{k_{n} = 0}^{1} \binom{1}{k_{1}} 1^{1-k_{1}} y_{1}^{k_{1}} \cdots \binom{1}{k_{n}} 1^{1-k_{n}} y_{1}^{k_{n}}  \nonumber \\
& = & \sum_{\vbf \in \{0,1\}^n} \binom{\mathbf{1}}{\vbf} \mathbf{1}^{\mathbf{1}-\vbf}\ybf^{\vbf}
= \sum_{\vbf \in \{0,1\}^n} \ybf^{\vbf},
\end{IEEEeqnarray}
where we employ the multi-index notation $\binom{\abf}{\bbf} = \prod_{i=1}^n \binom{a_i}{b_i}$ and $\abf^{\bbf} = \prod_{i=1}^n a_i^{b_i}$, for two $n$-vectors $\abf,\bbf$.  Consequently, for $\ybf = -\pbf$,
\begin{IEEEeqnarray}{rCl}
\label{eq:afterMultibinomial}
\prod_i (1-p_i)^{n-2} & = & \left( \sum_{\vbf \in \{0,1\}^n} (-\pbf)^{\vbf} \right)^{n-2} 
= \left( \sum_{t=1}^{2^n} (-\pbf)^{\vbf_{t}} \right)^{n-2} \nonumber \\
& = & \left( \sum_{t=1}^{2^n} \prod_{i=1}^{n} (-p_{i})^{\Vbf_{i,t}} \right)^{n-2},
\end{IEEEeqnarray}
where $\vbf_{t}$ is the $t^{\rm th}$ column of $\Vbf$.  The multinomial theorem states, for arbitrary $r$-vector $\ybf$ and positive integer $s$,
\begin{equation}
(y_1 + \cdots + y_r)^s = \sum_{\kbf \in \Kmc_{r,s}} \binom{s}{\kbf} \ybf^{\kbf}, 
\end{equation}
for $\Kmc_{r,s}$ defined in the proposition.  We apply the multinomial theorem to the RHS of \eqref{eq:afterMultibinomial} and get
\begin{IEEEeqnarray}{rCl}
\IEEEeqnarraymulticol{3}{l}
{\prod_i (1-p_i)^{n-2} 
}\nonumber \\* \quad
& = & \sum_{\kbf \in \Kmc_{2^n,n-2}} \binom{n-2}{k_{1}, \ldots, k_{2^n}} \prod_{t=1}^{2^n} \left( \prod_{i=1}^{n} (- p_{i})^{\Vbf_{i,t}} \right)^{k_{t}} \nonumber \\
& = & \sum_{\kbf \in \Kmc_{2^n,n-2}} \binom{n-2}{k_{1}, \ldots, k_{2^n}} (-1)^{\sum_{i=1}^{n} \alpha_{i}} \prod_{i=1}^{n} p_{i}^{\alpha_{i}},\IEEEeqnarraynumspace
\end{IEEEeqnarray}
for $\alpha_i$ defined in the proposition.  Finally substitution of this expression of $\prod_i (1-p_i)^{n-2} $ into \eqref{eq:volLambdaBasedOnAbr1977AndGenericExpression} and application of \eqref{eq:GruMol1978} with $\alpha_{0} = 1$ yields the desired volume expression in \eqref{eq:volLambda}.
\end{IEEEproof}

The number of summands in \eqref{eq:volLambda} is the number of multinomial coefficients. Equivalently, it is the number of ways to write $n-2$ as an ordered sum of $2^{n}$ non-negative integers, and is given by $\binom{n-2 + 2^{n} - 1}{n-2}$ (see e.g., Wilf \cite{Wil1994} Ex.\ 3 in Chapter 2). 
Applying an easy lower bound on the binomial coefficient $\binom{n}{k} \geq \left( \frac{n}{k}\right)^{k}$, we have $\binom{n-2 + 2^{n} - 1}{n-2} \geq \left( 1 + \frac{2^{n} - 1}{n-2}\right)^{n-2}$, meaning it grows super-exponentially in $n$, and hence calculation of $\mathrm{vol}(\Lambda)$ using Prop.\ \ref{prop:volLambda} requires substantial computation for even moderate $n$.

We now initiate our pursuit of non-parametric bounds on $\Lambda$, which is the focus of the next three sections.  Recall it is already known that when $n = 2$, $\Lambda$ equals a non-parametric set $\left\{  \xbf \in \Rbb_{+}^{2}  : \sqrt{x_1} + \sqrt{x_2} \leq 1 \right\}$ (\cite{TsyMik1979, RaoEph1988, MasMat1985}) for which membership testing is simple. Naturally one might wonder how the natural extension of this sum relates to $\Lambda$ for higher values of $n$. This motivates the following definition of the ``square root sum'' set. The proposition below shows in general this set is only an inner bound on $\Lambda$.  In the subsequent proof and elsewhere throughout the paper, we use the fact that $\Lambda$ is {\em coordinate convex}, meaning if $\xbf \in \Lambda$ then $\xbf' \in \Lambda$ for all $\mathbf{0} \leq \xbf' \leq \xbf$. 

\begin{definition} 
\begin{equation}
\Lsrs \equiv \left\{  \xbf \in \Rbb_{+}^{n}: \sum_{i=1}^n \sqrt{x_i} \leq 1  \right\}.
\end{equation}
\end{definition}

\begin{proposition}[``square root sum'' inner bound]
\label{prop:srsib}
The set $\Lsrs$ is an inner bound on $\Lambda$ for $n \geq 2$.
\end{proposition}
\begin{IEEEproof}
Fix a point $\xbf' \in \Lsrs$.  Due to the coordinate convexity of $\Lambda$ and $\Lsrs$, it suffices to produce a point $\xbf \in \Lambda$ so that $\xbf \geq \xbf'$.  Set $\pbf$ with $p_i = \sqrt{x_i'}$ for each $i$ and set $\xbf = \xbf(\pbf)$ according to  \eqref{eq:xOfp} in Def.\ \ref{def:fOfdeltaxAndpOfdeltaxAndxOfp}.  Clearly $\xbf \in \Lambda$.  It remains to show $x_i \geq x_i'$ for each $i \in [n]$.  Note $\xbf' \in \Lsrs$ ensures $\sum_{i=1}^{n} p_i \leq 1$.  
Define independent events $A_1,\ldots,A_n$ with $\Pbb(A_i) = 1 - p_i$ for each $i \in [n]$. Denote the complement of event $A_{i}$ by $A_{i}^{c}$. It follows that
\begin{equation} \label{eq:LambdasrsKeyEquality}
1 - \Pbb \left( \bigcup_{j \neq i} A_j^{c} \right) = \Pbb \left( \bigcap_{j \neq i} A_j \right)  = \prod_{j \neq i} \Pbb \left( A_{j} \right) = \prod_{j \neq i} (1-p_j).
\end{equation}
Then for any $i$, reversely applying \eqref{eq:LambdasrsKeyEquality} to $x_{i}$ followed by the union bound and then the fact $\sum_{i=1}^{n} p_i \leq 1$, we have
\begin{IEEEeqnarray}{rCl}
x_i & = & p_i \left( 1 - \Pbb \left( \bigcup_{j \neq i} A_j^{c} \right) \right)   \nonumber \\
& \geq & p_i \left( 1 - \sum_{j \neq i} \mathbb{P} (A_j^{c})\right)
 =  p_i \left( 1 - \sum_{j \neq i} p_{j} \right)  \geq p_i^2 = x_{i}'.\IEEEeqnarraynumspace
\end{IEEEeqnarray}
\end{IEEEproof}

Below we compute the exact volume of $\Lsrs$. It has been seen from Fig.\ \ref{fig:vol} (in \S \ref{sec:intro}) that, although simple, $\Lsrs$ is a very poor inner bound.
\begin{proposition}
\label{prop:volLambdasrs}
The volume of the inner bound $\Lsrs$ is
\begin{equation}
\label{eq:volLambdasrs}
\vol(\Lsrs) = \frac{2^n}{(2n)!}.
\end{equation}
\end{proposition}
\begin{IEEEproof}
Use the change of variable $y_{i} = \sqrt{x_{i}}$, $\forall i \in [n]$ so that the volume integration becomes
\begin{IEEEeqnarray}{rCl}
\IEEEeqnarraymulticol{3}{l}
{\vol(\Lsrs)
}\nonumber \\* ~
& = & \int_{[0,1]^{n}} \Ibb_{\sum_{i=1}^{n} \sqrt{x_{i}} \leq 1} ~ \drm x_{1} \cdots \drm x_{n}  \nonumber \\
& = & 2^{n} \int_{[0,1]^{n}} \Ibb_{\sum_{i=1}^{n} y_{i} \leq 1} ~ y_{1} \drm y_{1} \cdots y_{n} \drm y_{n} \nonumber \\
& = & 2^{n} \int_{0}^{1} y_{n} \int_{0}^{1-y_{n}} y_{n-1} \cdots \nonumber \\
& & \negmedspace{} \cdot \int_{0}^{1-y_{n} - \cdots -y_{3}} y_{2} \int_{0}^{1-y_{n} - \cdots -y_{2}}  y_{1} \drm y_{1} \drm y_{2} \cdots \drm y_{n-1} \drm y_{n}.\IEEEeqnarraynumspace
\label{eq:volLambdasrsStep1}
\end{IEEEeqnarray}
It will be useful to first compute an integral denoted $I(d,k) = \int_{0}^{d} y (d - y)^{k} \drm y$ for $d \geq 0$, $k \in \Zbb_{+}$; expansion of this integral would require using the binomial theorem and then handling the resulting alternating sum.  If instead we employ a change of variable $z = d - y$ and integrate with respect to $z$ we obtain directly
\begin{equation} \label{eq:Idk}
I(d,k) =  \frac{1}{(k+1)(k+2)} d^{k+2}.
\end{equation}
For $j \in [n]$ define $k_{j} = 2 (j-1)$ and for $j \in [n-1]$ define $d_{j} = 1 - y_{n} - \cdots - y_{j+1}$, and $d_{n} = 1$. Observe the recurrences $k_{j} + 2 = k_{j+1}$, $d_{j} = d_{j+1} - y_{j+1}$. Specializing \eqref{eq:Idk} with parameters $d_{j}$, $k_{j}$ and dummy integrating variable $y_{j}$, we have
\begin{equation} \label{eq:Idkforj}
I(d_{j}, k_{j}) = \int_{0}^{d_{j}} y_{j} (d_{j} - y_{j})^{k_{j}} \drm y_{j} = \frac{d_{j}^{k_{j}+2}}{(k_{j}+1)(k_{j}+2)}, \forall j \in [n].
\end{equation}
Now we are ready to resume the computation of $\vol(\Lsrs)$ in \eqref{eq:volLambdasrsStep1}. Using our new notation, we have:
\begin{IEEEeqnarray}{rCl}
\IEEEeqnarraymulticol{3}{l}
{\vol(\Lsrs)
}\nonumber \\* \quad
& = & 2^{n} \int_{0}^{d_{n}} y_{n} \int_{0}^{d_{n-1}} y_{n-1} \cdots \nonumber \\
& & \negmedspace{} \cdot \int_{0}^{d_{2}} y_{2} \int_{0}^{d_{1}} y_{1} (d_{1} - y_{1})^{k_{1}} \drm y_{1} \drm y_{2} \cdots \drm y_{n-1} \drm y_{n}. \IEEEeqnarraynumspace
\end{IEEEeqnarray}
We can then repeatedly apply \eqref{eq:Idkforj} with $j \in [n]$. To see this, observe after the $j^{\rm th}$ innermost integration, the new innermost integration is
\begin{IEEEeqnarray}{rCl}
& & \int_{0}^{d_{j+1}} \prod_{s=1}^{j} \frac{1}{(k_{s}+1)(k_{s}+2)} y_{j+1} d_{j}^{k_{j} + 2} \drm y_{j+1} = \nonumber \\
& & \prod_{s=1}^{j} \frac{1}{(k_{s}+1)(k_{s}+2)}  \int_{0}^{d_{j+1}}   y_{j+1}  \left( d_{j+1} - y_{j+1} \right)^{k_{j+1}} \drm y_{j+1},\IEEEeqnarraynumspace
\end{IEEEeqnarray}
which is $\prod_{s=1}^{j} \frac{1}{(k_{s}+1)(k_{s}+2)} I (d_{j+1}, k_{j+1})$. 

Therefore, after the $j = (n-1)^{\rm st}$ innermost integration, we have
\begin{IEEEeqnarray}{rCl}
\IEEEeqnarraymulticol{3}{l}
{\vol(\Lsrs) 
}\nonumber \\* \quad
& = & 2^{n} \prod_{s=1}^{n-1} \frac{1}{(k_{s}+1)(k_{s}+2)} \int_{0}^{d_{n}} y_{n} d_{n-1}^{k_{n-1}+2} \drm y_{n} \nonumber \\
& = & 2^{n} \prod_{s=1}^{n-1} \frac{1}{(k_{s}+1)(k_{s}+2)}  \int_{0}^{d_{n}} y_{n}  \left( d_{n} - y_{n} \right)^{k_{n}} \drm y_{n} \nonumber \\
& = & 2^{n} \prod_{s=1}^{n-1} \frac{1}{(k_{s}+1)(k_{s}+2)} I (d_{n}, k_{n}) \nonumber \\
& = & 2^{n} \prod_{s=1}^{n} \frac{1}{(k_{s}+1)(k_{s}+2)} d_{n}^{k_{n}+2} =  2^{n} \prod_{s=1}^{n} \frac{1}{(2s-1)(2s)} \nonumber \\
&=& \frac{2^{n}}{(2n)!}.
\end{IEEEeqnarray}
\end{IEEEproof}

As we have seen, the set $\Lambda$ is parametric, making it both algebraically cumbersome and geometrically unintuitive. The polynomial root test (\S\ref{sec:rootTesting}) can be leveraged for the purpose of testing membership in $\Lambda$ and finding stabilizing control(s), yet it fails to exhibit the geometric aspects of of the region defined by $\Lambda$.  In the following three sections we present various inner and outer bounds on $\Lambda$ based on ``simple'' polyhedra (\S \ref{sec:piAndpo}), spheres (\S \ref{sec:siAndso}), and ellipsoids (\S \ref{sec:eiAndeo}). Our approach to developing bounds is largely geometric, meaning we use simple geometric objects and their constructions are suggested by the shape of $\Lambda$ in low dimensions. Specifically, the hyperplane is one of the simplest objects and we use it to construct polyhedral bounds; the sphere is also a natural choice since the Pareto frontier of $\Lambda$ has a smooth symmetric curvature; and finally, the ellipsoid generalizes the sphere and is more versatile. In all cases, the positioning of these bound-inducing geometric objects is important and is guided by exploiting symmetries. Another factor is that the construction of bounds should be simple (e.g., by enforcing tangency/incidence with $\Lambda$ at some special points) to hopefully yield better analytical tractability in establishing the correctness of the bound for arbitrary dimensions.  We use the volume of a bound to measure its quality.  Observe from Fig.\ \ref{fig:vol} and Table \ref{tab:voltable} that the proposed inner bounds are in general better than the outer bounds.  The simple optimal polyhedral inner bound $\Lpi^{*}$ and spherical inner bound $\Lsi^{*}$ are both very tight; together they suggest the ``mass'' of $\Lambda$ is less concentrated towards the corner points $\ebf_{i}$'s. 

\section{Polyhedral inner and outer bounds on $\Lambda$}
\label{sec:piAndpo}

In this section we form inner and outer bounds on $\Lambda$ using polyhedra.  The inner bound is formed using a single hyperplane, i.e., a generalized simplex, while the outer bound is formed using the intersection of a collection of $n+1$ hyperplanes in $\Rbb_{+}^{n}$.

\begin{definition}
\begin{IEEEeqnarray}{rCl}
\label{eqn:Lambdapifaimly}
\Lpi (\pbf) \equiv \left\{  \xbf \in \Rbb_{+}^{n}: (\mathbf{1}-\pbf)^{\Tsf} \xbf \leq \prod_{i}(1-p_{i})\right\}, \pbf \in \partial \Smc.\IEEEeqnarraynumspace
\end{IEEEeqnarray}
\end{definition}

Geometrically, the set $\Lpi (\pbf)$ is a generalized simplex bounded by the $n$ coordinate hyperplanes and the hyperplane with normal vector $\mathbf{1}-\pbf$.  All such hyperplanes are tangent to $\partial\Lambda$ with a tangency point at $\xbf(\pbf)$ in \eqref{eq:xOfp} in Def.\ \ref{def:fOfdeltaxAndpOfdeltaxAndxOfp}. Recognizing this geometric property and the fact that the complement of $\Lambda $ in $\Rbb_{+}^{n}$ is convex (both shown in \cite{Pos1985}),  we can in principle construct an \textit{arbitrarily accurate} polyhedral outer (and inner) bound, as we briefly describe here: one can choose a collection of $M$ points ($\pbf_{s}$, $s \in [M]$) from $\partial \Smc$. Use the mapping $\pbf \mapsto \xbf$ (in the sense of $\Lambda_{\rm eq}$) namely \eqref{eq:xOfp} to construct the corresponding $\xbf_{s}(\pbf_{s}), s \in [M]$ which are points on $\partial \Lambda$. Form a polyhedron having these $M$ points $\xbf_{s}$ as extreme points and $n$ extreme rays along each of the coordinate axis (in the positive orthant). This polyhedron (denoted $\Psf_{o}$) is an outer bounding polyhedron for $\Lambda$ in that $\Lambda \subseteq \Smc \setminus \Psf_{o}$. By increasing the number of the chosen points, the induced outer bound can be made arbitrarily close to $\Lambda$. For the inner bound, note the intersection of the halfspaces (leaning toward $\Rbb_{+}^{n}$) associated with the tangent hyperplane of $\partial \Lambda$ at $\xbf_{s}(\pbf_{s})$, truncated by $\Rbb_{+}^{n}$, is an inner bounding polyhedron (denoted $\Psf_{i}$) for $\Lambda$ in that $\Smc \setminus \Psf_{i} \subseteq \Lambda$. This inner bound can also be made arbitrarily close to $\Lambda$. In fact, each outer bound so constructed can be thought of having a ``matching'' inner bound (and vice versa), in that the same set of $M$ points can be used to induce both an outer bound (these points treated as extreme points namely the vertices of $\Psf_{o}$) and an inner bound (these points treated as tangent points on $\partial \Lambda$ which in turn define the facets of $\Psf_{i}$). Algorithmically, a stopping criterion could be one simply measuring some type of ``gap'' between the matched outer and inner bounds, and once it is met we know both bounds track (the volume of) $\Lambda$ reasonably well even without knowing much about the properties (e.g., volume) of $\Lambda$ in advance.

As the number of users $n$ grows, however, these arbitrarily accurate polyhedral bounds will be computationally infeasible to obtain.  This is in contrast with the root test (for the purpose of membership testing) and non-parametric bounds (collected in Table \ref{tab:voltable}), with the latter scaling with $n$ very well and only requiring simple computations.  Moreover, it would be hard to gain geometric insight about $\Lambda$ from these bounds. 

All that said, as will be shown in the following proposition, using only a single ``best'' point on $\partial \Lambda$ induces a good inner bound. Specifically, Prop.\ \ref{prop:Lambdapi} given below asserts for each given $\pbf$ the set $\Lpi (\pbf)$ is an inner bound on $\Lambda$, indicates the $\pbf^{*}$ that achieves the largest volume bound over this family of inner bounds, and also computes the corresponding volume.

\begin{proposition}[polyhedral inner bound]
\label{prop:Lambdapi}
For each $\pbf \in \partial \Smc$, the set $\Lpi (\pbf)$ is an inner bound on $\Lambda$ for $n \geq 2$. 
Among these, the tightest is given when $\pbf = \pbf^{*} = \frac{1}{n} \mathbf{1}$, namely,
\begin{equation} \label{eq:Lambdapistar}
\Lpi^{*} = \Lpi(\pbf^{*}) = \left\{  \xbf \in \Rbb_{+}^{n}: ~ \sum_{i=1}^{n} x_{i} \leq \left(1-\frac{1}{n}\right)^{n-1} \right\}.
\end{equation}
and the corresponding volume of this set is $\vol (\Lpi^{*}) = \frac{1}{n!}\left( 1 - \frac{1}{n} \right)^{n(n-1)}$.
\end{proposition}

\begin{IEEEproof}
Recall Post \cite{Pos1985} established that the complement of $\Lambda $ in $\Rbb_{+}^{n}$ is convex, and gave the tangent hyperplane at a point $\xbf(\pbf)$ on $\partial \Lambda$: $\left\{\xbf: (\mathbf{1}-\pbf)^{\Tsf} \xbf = \prod_{i}(1-p_{i}) \right\}$, where $\pbf \in \partial \Smc$ is the unique control associated with a point $\xbf \in \partial \Lambda$.
Since this hyperplane is a supporting hyperplane, this open convex set $\Lambda^{c} \cap \Rbb_{+}^{n}$ lies entirely on one ``side'', i.e., the open halfspace $\left\{  \xbf: (\mathbf{1}-\pbf)^{\Tsf} \xbf > \prod_{i}(1-p_{i})\right\}$, of the hyperplane. This means points on the other side of this hyperplane are not in $\Lambda^{c} \cap \Rbb_{+}^{n}$, and hence are in $\Lambda$, i.e., $\Lpi(\pbf) \subseteq \Lambda$.

Now notice $\Lpi(\pbf)$ is a generalized simplex, and its volume is given by (\cite{Ell1976}):
\begin{equation}
\mathrm{vol}(\Lpi(\pbf)) = \frac{1}{n!} \prod_{i=1}^{n} \prod_{j \neq i}(1-p_{j}) = \frac{1}{n!} \left( \prod_{i=1}^{n}(1-p_{i})\right)^{n-1}.
\end{equation}
It is easily shown that the function $\prod_{i=1}^{n}(1-p_{i})$ is maximized over $\pbf \in \partial \Smc$ at $\pbf =\frac{1}{n} \mathbf{1}$, and hence the best $\Lpi$ (in terms of achieving the largest volume) is given by \eqref{eq:Lambdapistar}.
\end{IEEEproof}

\begin{remark} \label{remark:LpiStarVersusLsrs}
Using lower and upper bounds on the factorial \cite{Bat2008}, one can show $\vol(\Lpi^{*}) \geq \vol(\Lsrs)$ for all $n \geq 3$.
\end{remark}

\begin{remark}
For such a simple bound, its quality seems better than one might previously think. Also, observe that if we ``duplicate'' $\Lambda$ and the corresponding inner bound $\Lpi^{*}$ in each of the $2^{n}$ orthants of $\Rbb^{n}$, the union of these $2^{n}$ inner bounds $\Lpi^{*}$'s is an inscribed maximum volume $L_{1}$-ball of the union of these $2^{n}$ sets $\Lambda$'s. This suggests the inscribed maximum volume $L_{1}$-ball contains a non-vanishing (in $n$) fraction of the volume of $\Lambda$.
\end{remark}

Next we construct a polyhedral outer bound.  If we restrict ourselves to only using a single halfspace, the best choice is the standard simplex, $\Smc$, which is a very loose outer bound.  Consequently, we consider a specific construction using $2n+1$ hyperplanes. The convex polytope given below is a subset of $\Smc$ (and in fact a subset of $\Smc \setminus \Lambda$), has $\partial\Smc$ as a facet, and an additional $n$ facets each defined by a hyperplane, $(\Hmc_1^{\alpha},\ldots,\Hmc_n^{\alpha})$, where $\Hmc_i^{\alpha}$ is the hyperplane passing through $\ebf_i,\mbf$, and $\alpha \ebf_j$ for all $j \neq i$, for $\alpha$ given below.  

\begin{definition}
The halfspace representation of the convex polytope $\mathsf{P}$ in $\Rbb^{n}$ consists of the following halfspaces:
\begin{eqnarray}
&& \Hmc_{i}^{\alpha+} \equiv \left\{\xbf \in \Rbb^{n}: ~ x_{i} + \frac{1}{\alpha(n)} \sum_{j \neq i} x_{j} \geq 1 \right\}, ~ i \in [n], \nonumber \\
&& \Hmc_{i}^{\rm c+} \equiv \left\{\xbf \in \Rbb^{n}: ~ x_{i} \geq 0 \right\}, ~ i \in [n], \nonumber \\
&& \Hmc^{\partial \Smc-} \equiv \left\{ \xbf \in \Rbb^{n}:  ~ \sum_{i=1}^{n} x_{i} \leq 1 \right\},
\end{eqnarray}
where 
\begin{equation}
\label{eq:alphan}
\alpha(n) \equiv \frac{n-1}{\frac{n}{\left( 1 - \frac{1}{n} \right)^{n-1}} - 1} = \frac{n-1}{\frac{1}{m(n)} - 1}
\end{equation}
and the superscript $^{+}$ indicates an ``upward'' halfspace and $^{-}$ indicates a ``downward'' halfspace. More compactly,
\begin{equation}
\mathsf{P} \equiv \left\{\xbf \in \Rbb^{n}: ~ \xbf \in \bigcap_{i \in [n]} \Hmc_{i}^{\alpha+}  \bigcap_{i \in [n]} \Hmc_{i}^{\rm c+} \cap \Hmc^{\partial \Smc-} \right\}.
\end{equation}
Furthermore, the corresponding hyperplane is denoted by dropping these superscripts, meaning the inequality in the definition holds with equality. For example, $\Hmc_{i}^{\rm c}$ denotes the \textit{coordinate hyperplane} $\left\{\xbf \in \Rbb^{n}: x_{i} = 0 \right\}$.
\end{definition}

\begin{proposition}[polyhedral outer bound]
\label{prop:Lambdapo}
The convex polytope $\mathsf{P}$ defined above induces an outer bound on $\Lambda$. More precisely, $\Lambda \subseteq \Lpo \equiv \Smc \setminus \mathsf{P}$.
\end{proposition}

To prove the correctness of this bound it will be essential to establish the  monotonicity of $\alpha(n)$ \eqref{eq:alphan}.
\begin{lemma} \label{lem:alphamono}
The function $\alpha(n)$ \eqref{eq:alphan} is monotone increasing for $n \geq 2$. In particular, $\alpha(2) = 1/3$, $\alpha(3) = 8/23$, $\alpha(\infty) = 1/\erm$.
\end{lemma}

\begin{IEEEproof}
The derivative of $\alpha(n)$ \eqref{eq:alphan} is
\begin{equation}
\frac{\drm \alpha(n)}{\drm n} = \frac{\left( n-1 \right)^{2} \frac{1}{n}\left( 1 - \frac{1}{n} \right)^{n-1}}{\left( \left( 1 - \frac{1}{n} \right)^{n} - \left( n-1\right) \right)^{2}} \cdot g(n),
\end{equation}
where $g(n) \equiv  \left( 1 - \frac{1}{n}\left( 1 - \frac{1}{n} \right)^{n-1} + \left( n - 1 \right) \log \left( 1 - \frac{1}{n} \right) \right)$.
The sign of $\frac{\drm \alpha(n)}{\drm n}$ is determined by that of $g(n)$. To show the positivity of $g(n)$ for all $n \geq 2$, first observe $\lim_{n \to \infty} g(n) = 0$, therefore it suffices to show $g(n)$ is itself monotone decreasing in $n$, which is shown below:
\begin{IEEEeqnarray}{rCl}
\IEEEeqnarraymulticol{3}{l}
{n \frac{\drm g(n)}{\drm n}  = 1 + n \left( 1 - \frac{1}{n} \left( 1 - \frac{1}{n} \right)^{n-1} \right) \log \left( 1 - \frac{1}{n} \right) 
}\nonumber \\* ~
& \stackrel{(a)}{\leq} & 1 + n \left( 1 - \frac{1}{n} \left( 1 - \frac{1}{n} \right)^{n-1} \right) \left( -\frac{1}{n} - \frac{1}{2n^{2}} - \frac{1}{3n^{3}} \right) \nonumber \\
& \stackrel	{(b)}{\leq} & 1 + n \left( 1 - \frac{1}{2n} \right) \left( -\frac{1}{n} - \frac{1}{2n^{2}} - \frac{1}{3n^{3}} \right) = - \frac{n-2}{12 n^{3}} \leq 0,\IEEEeqnarraynumspace
\end{IEEEeqnarray}
where we use in $(a)$ the inequality $\log(1+x) \leq x - \frac{x^{2}}{2} + \frac{x^{3}}{3}$ for all $x \in (-1,0]$ and $(b)$ the property that $\left( 1 - \frac{1}{n} \right)^{n-1}$ is monotone decreasing in $n$ from $1/2$ (when $n=2$) to $1 / \erm$ (when $n = \infty$). 
\end{IEEEproof}

\begin{IEEEproof}(of Prop.\ \ref{prop:Lambdapo})
Our approach is to show all the vertices of the convex polytope $\mathsf{P}$ are in $\overline{\Lambda^{c}}$, it then follows from the convexity of $\overline{\Lambda^{c}} \cap \Rbb_{+}^{n}$ that $\mathsf{P} \subseteq \overline{\Lambda^{c}} \cap \Rbb_{+}^{n}$ which implies $\Lpo \equiv \Smc \setminus \mathsf{P} \supseteq \Lambda$.

To find a vertex, we first choose $n$ out of the $2n+1$ hyperplanes defining $\mathsf{P}$.  If there is a solution to this linear system that is a single point that \textit{also} obeys the remaining $n+1$ halfspace constraints, then this solution is a valid vertex (indicated below as underlined cases). Furthermore, since our primary goal is to show all the vertices are in $\overline{\Lambda^{c}}$, rather than to list all the vertices, for simplicity we only consider the scenario where those $n$ selected hyperplanes do not include $\Hmc^{\partial \Smc}$. This is justified since the intersection of $\Hmc^{\partial \Smc}$ and $n$ halfspaces $\Hmc_{i}^{\rm c+}$, namely $\partial \Smc$, lies completely in $\overline{\Lambda^{c}}$, so if there exists a valid vertex on $\Hmc^{\partial \Smc}$ it is guaranteed to be in $\overline{\Lambda^{c}}$.

Consequently, we choose a set of hyperplanes from $\left\{ \Hmc_{i}^{\alpha} \right\}_{i=1}^{n}$ (denoted $S$) and a set of hyperplanes from $\left\{ \Hmc_{i}^{\rm c} \right\}_{i=1}^{n}$ (denoted $T$) so that their cardinalities $\vert S \vert$ and $\vert T \vert$ sum to $n$. We also assume we choose the first $\vert S \vert$-indexed hyperplanes from $\left\{ \Hmc_{i}^{\alpha} \right\}_{i=1}^{n}$; this holds with no loss of generality as we may always permute the indices of the hyperplanes, and the polytope $\Psf$ is symmetric with respect to such permutations. For notational convenience define $\Imc_{S}$ and $\Imc_{T}$ as the set of indices appearing as subscripts of the elements in the set $S$ and $T$ respectively. For example, if $S = \left\{\Hmc_{1}^{\alpha}, \Hmc_{2}^{\alpha} \right\}$, then $\Imc_{S} = \left\{ 1, 2\right\}$; if $T = \left\{\Hmc_{1}^{\rm c} \right\}$, then $\Imc_{T} = \left\{ 1 \right\}$.  Recall, $m(k) = \frac{1}{k}\left(1 - \frac{1}{k} \right)^{k-1}$ for $k \in [n]$ stands for the coordinate of the all-rates-equal point on $\partial\Lambda$ in a $k$-dimensional space.

We discuss cases based on the pair $\left( \vert \Imc_{S} \cap \Imc_{T} \vert, \vert S \vert \right)$
\begin{itemize}
\item \underline{case 1}: $\vert \Imc_{S} \cap \Imc_{T} \vert$ = 0, $\vert S \vert = n$. Namely we choose all $n$ $\Hmc_{i}^{\alpha}$'s, to which the only solution is the all-rates-equal point $\mbf = m \mathbf{1}$, which is in $\overline{\Lambda^{c}}$.

\item \underline{case 2}: $\vert \Imc_{S} \cap \Imc_{T} \vert$ = 0, $\vert S \vert = k$ for $1 \leq k < n$. Namely $S = \left\{ \Hmc_{1}^{\alpha}, \ldots, \Hmc_{k}^{\alpha} \right\}$, $T = \left\{ \Hmc_{k+1}^{\rm c}, \ldots, \Hmc_{n}^{\rm c} \right\}$. The only solution can be shown to be $\xbf = \left( 1 + \frac{k-1}{\alpha(n)} \right)^{-1} \sum_{j=1}^{k} \ebf_{j}$. To verify this point $\xbf$ is in $\mathsf{P}$, we first verify it satisfies the halfspace constraint $\Hmc_{k+1}^{\alpha+}$, i.e., $x_{k+1} + \frac{1}{\alpha\left(n\right)} \sum_{j \neq k+1} x_{j} \geq 1$, which applied to this point becomes $\alpha\left(n\right) \leq 1$. Similarly $\xbf$ also satisfies the halfspace constraints associated with $\Hmc_{k+2}^{\alpha}$, $\ldots$, $\Hmc_{n}^{\alpha}$. Next, for $\xbf$ to satisfy the halfspace constraint $\Hmc^{\partial \Smc-}$, we again only need $\alpha\left(n\right) \leq 1$. Finally the nonnegativity constraint for each coordinate axis $\Hmc_{i}^{\rm c+}$ is satisfied, so this solution is a valid vertex. We now need to show $\xbf \in \overline{\Lambda^{c}}$. Observe this vertex's effective length is $k$ so we need to check $\overline{\Lambda^{c}}$ in the corresponding $k$-dimensional space; furthermore, all the non-zero components of $\xbf$ are identical meaning $\xbf$ lies along the all-rates-equal ray in this $k$-dimensional space so we only need to show $\xbf$ extends beyond the corresponding all-rates-equal point $\mbf = m(k) \mathbf{1}$ for $\mathbf{1}$ a $k$-vector of all $1$'s. Applying Lem.\ \ref{lem:alphamono}, we have $\left( 1 + \frac{k-1}{\alpha(n)} \right)^{-1} \geq \left( 1 + \frac{k-1}{\alpha(k)} \right)^{-1} = m(k)$. Thus we've shown this case does produce a valid vertex in $\overline{\Lambda^{c}}$.

\item case 3: $\vert \Imc_{S} \cap \Imc_{T} \vert$ = 0, $\vert S \vert = 0$. Namely we choose all $n$ coordinate hyperplane $\Hmc_{i}^{c}$'s. The only solution is the origin $\obf$, which is not in $\mathsf{P}$, hence this is an invalid vertex.

\item case 4: $\vert \Imc_{S} \cap \Imc_{T} \vert  = 1$, $\vert S \vert = k$ for $1 \leq k < n$. In this case, in order to further satisfy $\vert S \vert + \vert T \vert = n$, there must exist some index $k' \geq k +1$ such that $\Hmc_{k'}^{c} \notin T$. In fact if we assume $\Hmc_{1}^{\rm c} \in T$, this determines $T = \left\{ \Hmc_{1}^{\rm c}\right.$, $\Hmc_{k+1}^{\rm c}$, $\ldots$, $\left. \Hmc_{n}^{\rm c} \right\} \setminus \left\{ \Hmc_{k'}^{c} \right\}$. The solution can be shown to be $\xbf = \alpha \ebf_{k'}$, which is not in $\mathsf{P}$, and hence is not a valid vertex. Note this conclusion does not depend on our choice of $\Hmc_{1}^{\rm c}$ to be included in $T$.

\item case 5: $\vert \Imc_{S} \cap \Imc_{T} \vert  > 1$, $\vert S \vert = k$ for $1 < k < n-1$. In this case, in order to further satisfy $\vert S \vert + \vert T \vert = n$, there must exist $ l = \vert \Imc_{S} \cap \Imc_{T} \vert$ indices such that the corresponding coordinate hyperplanes are not in $T$. Attempt to solve this system shows this is an underdetermined system because the solution is given by a hyperplane instead of a point. Furthermore, if we want to ensure the solution is in $\Psf$, we find there is no consistent solution. As $S = \left\{ \Hmc_{1}^{\alpha}\right.$, $\ldots$, $\left. \Hmc_{k}^{\alpha} \right\}$, suppose $T$ does not include, say, $\Hmc_{k+1}^{\rm c}$, \ldots, $\Hmc_{k+l}^{\rm c}$ (as well as $\Hmc_{l+1}^{\rm c}$, $\ldots$, $\Hmc_{k}^{\rm c}$), so $T = \left\{ \Hmc_{1}^{\rm c}\right.$, $\ldots$, $\Hmc_{l}^{\rm c}$, $\Hmc_{k+l+1}^{\rm c}$, $\ldots$, $\left. \Hmc_{n}^{\rm c} \right\}$. Solving these $n$ equations gives an $l$-dimensional hyperplane: $\left\{\xbf: x_{k+1} + \cdots + x_{k+l} = \alpha \right\}$. Satisfying the halfspace constraint $\Hmc_{k+1}^{\alpha+}$ as well as each nonnegativity component constraints $\Hmc_{i}^{\rm c+}$ requires $x_{k+1} = 0$.  Similarly, due to each other halfspace constraint $\Hmc_{k+2}^{\alpha+}$, $\ldots$, $\Hmc_{k+l}^{\alpha+}$, each other component $x_{k+2}$, $\ldots$, $x_{k+l}$ would also need to be set to zero. These together lead to no valid (vertex) solution.
\end{itemize}

To summarize, each valid vertex solution we have found is such that: $i)$ its non-zero components are all equal, and $ii)$ this non-zero component value is no smaller than the coordinate of all-rates-equal point in the corresponding possibly reduced-dimensional space, which means all those vertices are in $\overline{\Lambda^{c}}$. More precisely, each vertex extends beyond (or coincides with) the corresponding all-rates-equal point (which lies on the boundary of $\Lambda$), and can be written as (up to permutation of the indices) $\xbf = \left( 1 + \frac{k-1}{\alpha(n)}\right)^{-1} \sum_{j=1}^{k} \ebf_{j}$ for $k \in [n]$. In particular, when $k=1$, $\xbf = \ebf_{1}$; when $k=n$, $\xbf = \mbf$.  
\end{IEEEproof}

\begin{remark}
One can perform a similar analysis considering the scenario where $\Hmc^{\partial \Smc}$ is selected. The only valid vertex solutions consist of just $\ebf_{i}$'s for all $i \in [n]$.
\end{remark}

We give the vertex representation of the convex polytope $\Psf$ when $n = 2$ and $n=3$ in the following example. The bounds $\Lpi^{*}$ and $\Lpo$ together with $\partial \Lambda$ are illustrated in Fig.\ \ref{fig:LpiLpoCombined}.
\begin{example}
When $n = 2$, since $\alpha(2) = 1/3$, the vertices of $\Psf$ are $\{ \ebf_{1}, \ebf_{2}, \mbf \}$, here $\mbf = m(2) \mathbf{1}$ for $ m(2)= 1/4$.
When $n = 3$, since $\alpha(3) = 8/23$, the vertices of $\Psf$ are $\{ \ebf_{1}, \ebf_{2}, \ebf_{3}, (8/31)(\ebf_{1}+\ebf_{2}),
(8/31)(\ebf_{1}+\ebf_{3}), (8/31)(\ebf_{2}+\ebf_{3}), \mbf \}$, here $\mbf = m(3) \mathbf{1}$ for $ m(3)= 4/27$. Those vertices are also shown in Fig.\ \ref{fig:LpiLpoCombined}. Note $8/31 > 1/4$ thus each of the three green points in the bottom subfigure extends beyond the all-rates-equal point on the corresponding $2$-dimensional plane namely the orange point in the top subfigure.
\end{example}

\begin{figure}[!ht]
\centering
\includegraphics[width=0.4\textwidth]{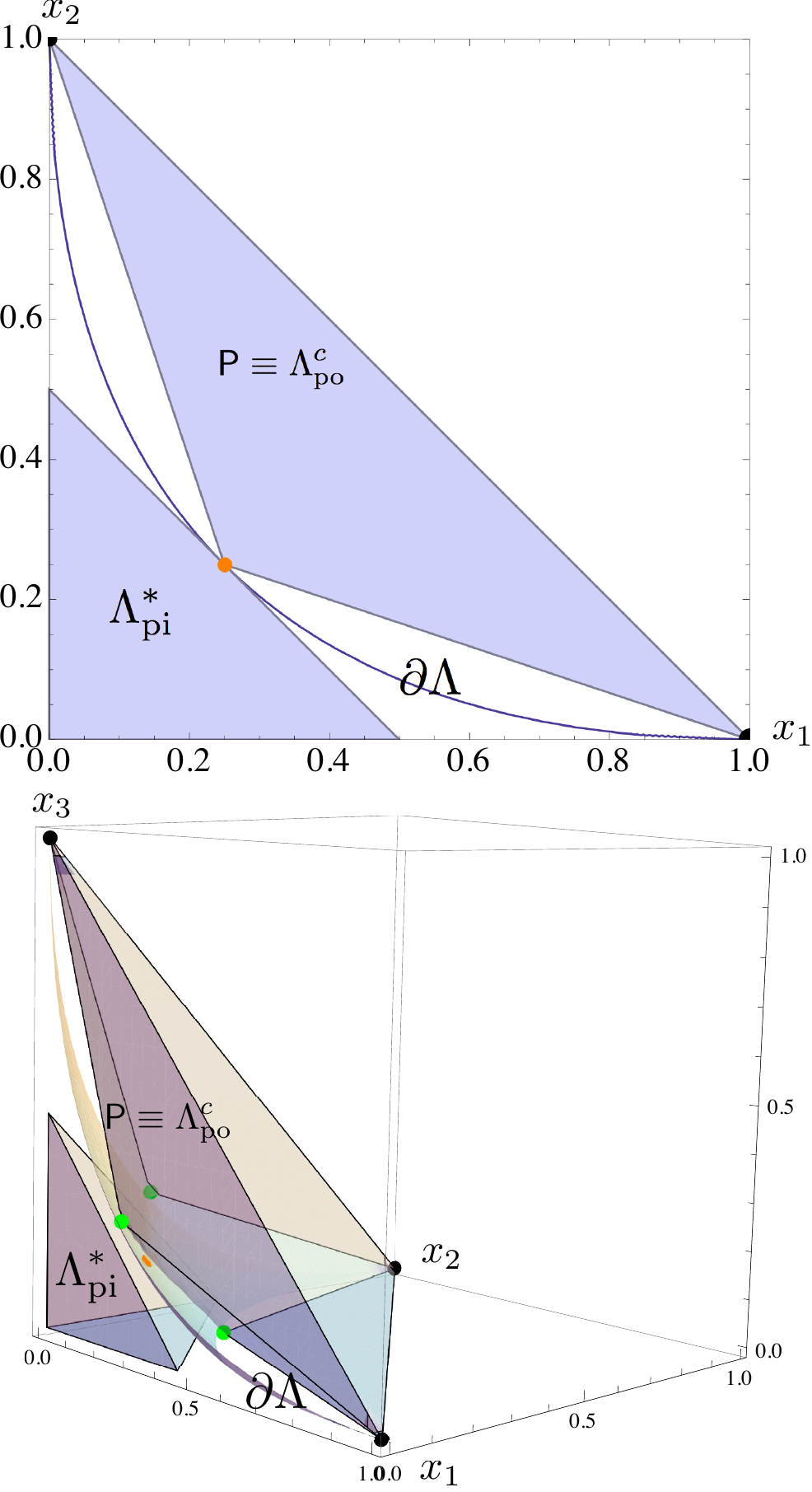}
\caption{Polyhedral bounds for $n=2$ (top) and $3$ (bottom).  The inner bound and the complement (w.r.t.\ $\Smc$) of the outer bound (namely $\Psf$, recall $\Lpo \equiv \Smc \setminus \Psf$) are shown in solid, and in between sits $\partial \Lambda$. Also shown are the vertices of the polytope used in $\Lpo$: black points are $\ebf_{i}$'s, the orange point is the all-rates-equal point $\mbf$, the green points on the right have all non-zero components equal $8/31$.}
\label{fig:LpiLpoCombined}
\end{figure} 

\section{Spherical inner and outer bounds on $\Lambda$}
\label{sec:siAndso}

In this section we consider bounds induced by spheres. More specifically we want $\Smc \setminus \Bmc(\cbf, r)$ to be included in (for inner bounding) or to include (for outer bounding) $\Lambda$, where $\Bmc(\cbf, r)$ denotes an open ball in $\Rbb^n$ with center $\cbf$ and radius $r$, and $\partial\Bmc(\cbf,r)$ is its boundary.  By symmetry, we restrict our attention to balls with centers on the all-rates-equal ray, i.e., $\cbf = c \mathbf{1}$ for some $c > 0$.  In the following, Prop.\ \ref{prop:Lambdasi} establishes a family of inner bounds induced by balls centered at $\cbf = c \mathbf{1}$ with radius $r(c) = d(\cbf, \mbf)$ (i.e., $\mbf \in \partial \Bmc(\cbf,r)$), and among them the best one (in the sense of giving the best approximation of the volume of $\Lambda$) is obtained by $c = (1-n m^2)/(2(1-nm))$, which is indeed the minimum $c$ in order to \textit{possibly} produce a valid spherical inner bound in this family. A parallel result, Prop.\ \ref{prop:Lambdaso}, establishes a family of outer bounds induced by balls centered at $\cbf = c \mathbf{1}$ with radius defined as $r(c) = d(\cbf, \ebf_{i})$ (i.e., $\ebf_i \in \partial \Bmc(\cbf,r)$ for $i \in [n]$), and among them the best one is given when $c = 1$, which is also the minimum $c$ in order to produce a valid spherical outer bound in this family.  

\begin{definition} \label{def:Lsi}
$\Lsi(c) \equiv \Smc \setminus \Bmc(\cbf, r_{\rm in}(c))$, where the center of the ball is $\cbf = c \mathbf{1}$ for all $c \geq c_{\rm in}^{*} \equiv (1-n m^2)/(2(1-nm))$, and its radius $r_{\rm in}(c) \equiv d(\cbf, \mbf) = \sqrt{n}(c-m)$.
\end{definition}

\begin{proposition}[spherical inner bound]
\label{prop:Lambdasi}
For each $c \geq c_{\rm in}^{*}$, the set $\Lsi(c)$ is an inner bound on $\Lambda$ for $n \geq 2$. Among these, the tightest is given when $c = c_{\rm in}^{*}$:
\begin{equation}
\label{eq:Lambdasi}
\Lsi^{*} =  \Lsi(c_{\rm in}^{*})  =  \left\{ \xbf \in \Smc : \|\xbf - c_{\rm in}^{*} \mathbf{1}\| \geq \sqrt{n}\left(c_{\rm in}^{*}-m\right) \right\}.
\end{equation}
\end{proposition}

\begin{IEEEproof}
Here is an overview of the proof. First, we observe the correctness of the spherical inner bound with some $c$ implies the correctness of an inferior bound with a larger $c$ (Lem.\ \ref{lem:LambdasiSupportingLemma}), so we only need to show the correctness of the bound with the minimum $c$ namely $c_{\rm in}^{*}$. Second, by inspecting the Karush-Kuhn-Tucker (KKT) conditions, we argue a potential local extremizer can have at most two distinct non-zero component values, and also obtain a condition the components of this extremizer must satisfy \eqref{eq:condIfTwoDistinctNonzeroValues}.  Third, we address the case when a potential extremizer does not have zero component and has exactly two distinct non-zero component values, and we show such a point can be safely ruled out for the optimization problem set up in Step 2. Fourth, we consider the case when a potential extremizer has zero component(s), and show this point can be removed too (unless it reduces to $\ebf_{i}$). Finally, it is clear we only need to evaluate the objective function at $\mbf$ and $\ebf_{i}$.

\textbf{Step 1}: correctness of the bound with small $c$ implies correctness and inferiority of the bound with larger $c$.  Lem.\ \ref{lem:LambdasiSupportingLemma} below establishes that $c_{2} \geq c_{1} \geq m$ implies $\Bmc(\cbf_{2}, d(\cbf_{2}, \mbf)) \supseteq \Bmc(\cbf_{1}, d(\cbf_{1}, \mbf))$, i.e., the balls in this family are nested in $c$.  Since $\Lsi(c) = \Smc \setminus \Bmc(\cbf,r_{\rm in}(c))$, it follows that $c_{2} \geq c_{1} \geq m$ implies $\Lsi(c_1) \supseteq \Lsi(c_2)$, i.e., the induced bounds are likewise nested, and thus the optimal (largest) bound in this family is obtained by the smallest $c$ in the family.  Because of this, we need only establish that $\Lsi(c) \subseteq \Lambda$ for this smallest $c$ in the family.  To establish $c_{\rm in}^{*}$ is the minimum $c$, it suffices to verify the following:
\begin{equation}
\label{eq:cinstarOneToOne}
d(\cbf, \mbf) <, =, > d(\cbf, \ebf_{i}) \text{ iff } c <, =, > c_{\rm in}^{*} \text{ respectively.}
\end{equation}
This is straightforward to establish, and the proof is omitted.  Assuming \eqref{eq:cinstarOneToOne} to be true, it follows that if $c < c_{\rm in}^{*}$ then each $\ebf_{i}$ will not be included in the closed ball $\overline{\Bmc}(\cbf, r_{\rm in})$, implying the induced bound $\Smc \setminus \Bmc(\cbf, r_{\rm in})$ is invalid (since each $\ebf_i \in \Lambda$).  

\begin{lemma}
\label{lem:LambdasiSupportingLemma}
$\Bmc(\cbf_{2}, d(\cbf_{2}, \mbf)) \supseteq \Bmc(\cbf_{1}, d(\cbf_{1}, \mbf))$ for $c_{2} \geq c_{1} \geq m$.
\end{lemma}
$\indent$ {\em Proof of Lem.\ \ref{lem:LambdasiSupportingLemma}:}
For simplicity we shift the origin of coordinate system along the all-rates-equal ray $\mathbf{1}$ so that it overlaps with $\mbf$ in the original system. In the new system we have $c_{2}' = c_{2} - m$, $c_{1}' = c_{1} - m$, and $m' = 0$, and we need to show $\Bmc(\cbf_{2}', d(\cbf_{2}', \obf')) \supseteq \Bmc(\cbf_{1}', d(\cbf_{1}', \obf'))$ for $c_{2}' \geq c_{1}' \geq 0$, where $\obf'$ denotes the origin of the new system. Observe $d(\cbf_{j}', \obf') = \sqrt{n} c_{j}'$ for $j = 1,2$. So we need to verify for all $\xbf$ satisfying $\sum_{i} (x_{i} - c_{1}')^{2} \leq n c_{1}'^{2}$, it holds that $\sum_{i} (x_{i} - c_{2}')^{2} \leq n c_{2}'^{2}$. Toward this, we write
\begin{IEEEeqnarray}{rCl}
\IEEEeqnarraymulticol{3}{l}
{\sum_{i} (x_{i} - c_{2}')^{2}
}\nonumber \\* \quad
& = & \sum_{i} (x_{i} - c_{1}')^{2} + n (c_{2}' - c_{1}')^{2} - 2 (c_{2}' - c_{1}') \sum_{i} (x_{i} - c_{1}') \nonumber \\
& \leq & n c_{1}'^{2} + n (c_{2}' - c_{1}')^{2} - 2 (c_{2}' - c_{1}') \sum_{i} (x_{i} - c_{1}').
\end{IEEEeqnarray}
So it suffices to show the RHS is no larger than $n c_{2}'^{2}$, which is equivalent to showing $\sum_{i} x_{i} \geq 0$. We claim this is true, because the hyperplane $\{\xbf \in \Rbb^{n}: \sum_{i} x_{i} = 0 \}$ is tangent with $\Bmc(\cbf_{1}', d(\cbf_{1}', \obf'))$ at $\obf'$, and in fact it is a supporting hyperplane of the convex body $\Bmc(\cbf_{1}', d(\cbf_{1}', \obf'))$.
\hfill $\square$

Steps 2, 3, and 4 are actually valid for $\Lsi(c)$ for all $c$, not just $c = c_{\rm in}^*$, and so we consider an arbitrary $\Lsi(c)$ in what follows.

\textbf{Step 2}: properties of a potential extremizer for $\Lsi$.  By Prop.\ \ref{prop:innerboundReducedtoBoundaryChecking} in \S \ref{ssec:simplification}, it suffices to establish $\partial\Lambda \subseteq \overline{\Bmc}(\cbf, r_{\rm in})$, i.e., given any point $\xbf \in \partial \Lambda$, its distance to the center of the sphere is no larger than the sphere's radius:
\begin{equation}
\label{eq:si_opt_first}
\max_{\xbf \in \partial \Lambda}d(\cbf, \xbf)^{2} \leq r_{\rm in}^{2}. 
\end{equation}
Recall from Cor.\ \ref{cor:bijectionbetweenSandLambda} (\S \ref{sec:rootTesting}) the bijection between $\partial \Lambda$ and $\partial \Smc$, and write $\xbf(\pbf)$ to denote the unique $\xbf \in \partial\Lambda$ associated with each $\pbf \in \partial\Smc$.  Under this bijection, the LHS of \eqref{eq:si_opt_first} becomes
\begin{equation}
\label{eq:si_opt_second}
\max_{\pbf \in \partial\Smc} f (\pbf) \equiv d(\cbf, \xbf(\pbf))^{2} = \sum_{i=1}^{n} \left(c - p_{i}\prod_{j \neq i}^{n} \left(1-p_{j} \right) \right)^{2}.
\end{equation}
Introducing Lagrange multipliers $\mu$, $(\lambda_i, i \in [n])$ for the equality constraint and $n$ inequality constraints in $\partial\Smc$, respectively, the Lagrangian of this maximization problem becomes:
\begin{equation}
\Lmc (\pbf, \mu, \boldsymbol\lambda) = f (\pbf) + \mu \left(\sum_{i=1}^{n} p_{i} - 1\right) + \sum_{i=1}^{n} \lambda_{i} \left(-p_{i}\right).
\end{equation}
The first-order Karush-Kuhn-Tucker (KKT) necessary conditions for a local maximizer are:
\begin{IEEEeqnarray}{rCl}
\text{stationarity}~ & &  \frac{\drm \Lmc}{\drm p_{i}} = 0,  i \in [n]  \\
\text{primal  feasibility}~ & &   \sum_{i=1}^{n} p_{i} - 1 = 0, -p_{i} \leq 0,  i \in [n] \IEEEeqnarraynumspace \\
\text{dual  feasibility}~ & &  \lambda_{i} \leq 0,  i \in [n]  \\
\text{complementary slackness}~ & &   \lambda_{i} \left(-p_{i}\right) = 0,  i \in [n].
\end{IEEEeqnarray}
Note the regularity condition LICQ (linear independence constraint qualification) is satisfied.

Observe $\frac{\drm \Lmc}{\drm p_{i}} = \frac{\drm f}{\drm p_{i}} + \mu - \lambda_{i}$.  Therefore, if a potential local maximizer $\pbf$ has two distinct non-zero components $0 < p_{k} < p_{l}$, then by complementary slackness, stationarity of the Lagrangian reduces to the equality of derivatives of the objective function w.r.t. $p_{k}$ and $p_{l}$:
\begin{equation}
\frac{\drm \Lmc}{\drm p_{k}} = \frac{\drm \Lmc}{\drm p_{l}} = 0 ~ \iff ~ \frac{\drm f}{\drm p_{k}} = \frac{\drm f}{\drm p_{l}} = - \mu.
\end{equation}
The derivative of $f$ w.r.t. $p_k$ is
\begin{equation}
\frac{1}{2}\frac{\drm f}{\drm p_{k}} = -\frac{\pi_{k}}{1-p_{k}} \left(c - p_{k} \pi_{k} \right) + \frac{1}{1-p_{k}} \sum_{i=1}^{n} p_{i} \left(c - p_{i} \pi_{i}\right) \pi_{i},
\end{equation}
where $\pi_{i} = \pi_{i}(\pbf) \equiv \pi (\pbf) /(1-p_{i})$.
Similarly we can write out the derivative w.r.t. $p_{l}$. Equating the two by further multiplying both sides by $(1-p_{k})(1-p_{l})$ gives
\begin{equation}
(1-p_{l}) (\eta_{c} - (c-p_{k}\pi_{k}) \pi_{k}) = (1-p_{k}) (\eta_{c} - (c-p_{l}\pi_{l}) \pi_{l}),
\end{equation}
where $\eta_{c} \equiv \sum_{i=1}^{n} p_{i} (c-p_{i}\pi_{i}) \pi_{i}$ and hence $\eta_{c}$ can be viewed as the expectation of a discrete random variable $Z$ with support $\{(c-p_{i}\pi_{i}) \pi_{i}, i \in [n]\}$ and associated PMF $\Pbb(Z = (c-p_{i}\pi_{i}) \pi_{i}) = p_{i}$ for each $i \in [n]$. So the only way to satisfy the above equality (for all $k, l$ such that $0 < p_{k} < p_{l}$) is by requiring $(c-p_{i}\pi_{i}) \pi_{i}$ to be all equal for $i$'s such that $p_{i} \neq 0$ (because otherwise we can always choose $k'$, $l'$ so that $\eta_{c}$ lies between $(c-p_{k'}\pi_{k'}) \pi_{k'}$ and $(c-p_{l'}\pi_{l'}) \pi_{l'}$). In particular, $(c-p_{l}\pi_{l}) \pi_{l} = (c-p_{k}\pi_{k}) \pi_{k}$, which simplifies, after some algebra, to:
\begin{equation} 
\label{eq:condIfTwoDistinctNonzeroValues}
c = \frac{\pi (\pbf)}{(1-p_{k}) (1-p_{l})} (1-p_{k}p_{l}).
\end{equation}

Because of the constraint enforced by \eqref{eq:condIfTwoDistinctNonzeroValues}, we claim there are \textit{at most two} distinct values among all the non-zero components of a potential local extremizer. To see this, we prove by contradiction. Assume there exist $p_{j}, p_{k}, p_{l}$ such that $0 < p_{j} < p_{k} < p_{l} < 1$. Then \eqref{eq:condIfTwoDistinctNonzeroValues} must hold with indices $\{j,k\}$ replacing indices $\{k,l\}$.  Equating the two resulting expressions for $c$ gives $p_{j} = p_{l}$, a contradiction.

With the above claim, we only need to consider points that have at most two distinct non-zero component values.  Define $\Vmc(\pbf) = \{ a \in (0,1] : \exists i \in [n] : p_i = a\}$ as the set of non-zero values taken by a $\pbf \in \partial\Smc$, and $\Zmc(\pbf) = \{ i \in [n] : p_i = 0\}$ as the set of indices where $\pbf$ has a zero value.  The set of probability vectors taking at most two distinct non-zero component values is then denoted $\Pmc^{(2)} = \{ \pbf \in \partial\Smc : |\Vmc(\pbf)| \in \{1,2\}\}$.  We partition this set into two subsets, $\Pmc^{(2)} = \Pmc^{(2)}_a \cup \Pmc^{(2)}_b$, which are in turn each partitioned into two subsets, $\Pmc^{(2)}_a = \Pmc^{(2)}_{a,1} \cup \Pmc^{(2)}_{a,2}$ and $\Pmc^{(2)}_b = \Pmc^{(2)}_{b,1} \cup \Pmc^{(2)}_{b,2}$, where 
\begin{IEEEeqnarray}{rCl}
& & \Pmc^{(2)}_a = \{ \pbf \in \Pmc^{(2)} : \Zmc(\pbf) = \emptyset\}; \Pmc^{(2)}_b  = \{ \pbf \in \Pmc^{(2)} : \Zmc(\pbf) \neq \emptyset\} \nonumber \\
& & \Pmc^{(2)}_{a,1} = \{ \pbf \in \Pmc^{(2)}_a : |\Vmc(\pbf)| = 1\},  \nonumber \\
& & \Pmc^{(2)}_{a,2} = \{ \pbf \in \Pmc^{(2)}_a : |\Vmc(\pbf)| = 2\}; \nonumber  \\
& & \Pmc^{(2)}_{b,1} = \{ \pbf \in \Pmc_{b}^{(2)} : |\Zmc(\pbf)| = n-1\},  \nonumber \\
& & \Pmc^{(2)}_{b,2} = \{ \pbf \in \Pmc_{b}^{(2)} : |\Zmc(\pbf)| \in \{1,\ldots,n-2\}\}.
\end{IEEEeqnarray}
In words, $\Pmc^{(2)}_a$ holds $\pbf \in \partial\Smc$ with no component equal to zero and at most two distinct (non-zero) values, while $\Pmc^{(2)}_b$ holds those with at least one component equal to zero and at most two distinct non-zero values.  Likewise, $\Pmc^{(2)}_{a,1}$ holds $\pbf$ with no zero components and only one (non-zero) value, meaning $\Pmc^{(2)}_{a,1} = \{\frac{1}{n}\mathbf{1}\}$, and $\Pmc^{(2)}_{a,2}$ holds $\pbf$ with all components taking one of two non-zero values, and both values held by some component.  Finally, $\Pmc^{(2)}_{b,1}$ holds $\pbf$ with all but one of the $n$ entries holding value zero, meaning $\Pmc^{(2)}_{b,1} = \{\ebf_1,\ldots,\ebf_n\}$, and $\Pmc^{(2)}_{b,2}$ holds $\pbf$ with between one and $n-2$ components taking value zero, and all non-zero components taking at most two distinct (non-zero) values.   The next step (Step $3$) in the proof focuses on $\Pmc^{(2)}_{a,2}$, while Step $4$ focuses on $\Pmc^{(2)}_{b,2}$; the simpler cases $\Pmc^{(2)}_{a,1}$ and $\Pmc^{(2)}_{b,1}$ will be left until the end.

\textbf{Step 3}: any $\pbf \in \Pmc^{(2)}_{a,2}$ cannot be a global maximizer.  We define the subset $\Pmc^{(2),*}_{a,2} \subseteq \Pmc^{(2)}_{a,2}$ as the collection of points from $\Pmc^{(2)}_{a,2}$ that also satisfies \eqref{eq:condIfTwoDistinctNonzeroValues}, which is a necessary condition for any such $\pbf$ to be a potential extremizer. In order to rule out the possibility that a point from $\Pmc^{(2)}_{a,2}$ can be a global maximizer, based on the KKT condition analysis, we only need to show the original objective function $f$ maximized over $\pbf \in \Pmc^{(2),*}_{a,2}$ is no larger than say $f(\ebf_{i}) = d(\cbf, \ebf_{i})^{2}$, equivalently we show another function $\tilde{f}$ maximized over $\pbf \in \Pmc^{(2),*}_{a,2}$ is no larger than $f(\ebf_{i})$ where $\tilde{f} = f$ for all $\pbf \in \Pmc^{(2),*}_{a,2}$. It suffices to work with an enlarged feasible set, meaning we shall show $\tilde{f}$ maximized over $\pbf \in \Pmc^{(2)}_{a,2}$ is still smaller than $f(\ebf_{i})$.

As $\pbf \in \Pmc^{(2)}_{a,2}$ by assumption, there is no loss in generality in denoting the two non-zero values it takes by $\Vmc(\pbf) = \{p_s,p_l\}$ for $0 < p_s < p_l < 1$, where $s$ stands for small and $l$ for large (and do not denote indices).  Assume there are $k$ ($1 \leq k \leq n-1$) components that equal $p_{s}$ and hence $(n-k)$ components equal $p_{l}$.  Then $k p_{s} + (n-k) p_{l} = 1$, and it follows from these assumptions that $0 < p_{s} < \frac{1}{n} < p_{l} < 1$, where we emphasize the strictness of each of the above inequalities.  Because of the assumption of exactly two distinct non-zero values for $\pbf$, \eqref{eq:condIfTwoDistinctNonzeroValues} simplifies to
\begin{equation} 
\label{eq:eqvForc}
c = (1-p_{s})^{k-1} (1-p_{l})^{n-k-1} (1-p_{s}p_{l}).
\end{equation}

Recall all the points from the set $\Pmc^{(2)}_{a,2}$ satisfy $k p_{s} + (n-k) p_{l} = 1$, only points from the subset $\Pmc^{(2),*}_{a,2}$ also satisfy \eqref{eq:eqvForc}.  We now express the original objective function $f$ from \eqref{eq:si_opt_second} as another function, $\tilde{f}(p_s,k)$, where $f(\pbf) = \tilde{f}(p_s,k)$ for all $\pbf \in \Pmc^{(2),*}_{a,2}$, i.e., for all $\pbf$ for which both $k p_{s} + (n-k) p_{l} = 1$ and \eqref{eq:eqvForc} hold:
\begin{IEEEeqnarray}{rCl}
\label{eqn:fobj}
\tilde{f} (p_{s}, k) & = & k \left(c - p_{s}(1-p_{s})^{k-1} (1-p_{l})^{n-k} \right)^{2} \nonumber \\
& & \negmedspace{} + (n-k) \left(c - p_{l} (1-p_{s})^{k} (1-p_{l})^{n-k-1} \right)^{2} \nonumber \\
& = & c^{2} (n-k) \frac{(n-k)(n-2) + 1 - 2 k p_{s} + k n p_{s}^{2}}{(n-k-p_{s} + k p_{s}^{2})^{2}}.\IEEEeqnarraynumspace
\end{IEEEeqnarray}

Fixing $p_{s} \in (0, \frac{1}{n})$ temporarily, we now show $\tilde{f}$ is monotone increasing in $k$ for $k \in \{1,\ldots,n-1\}$. Denote $u = u(p_s,k) = (n-k)(n-2) + 1 - 2 k p_{s} + k n p_{s}^{2}$ and $v = v(p_s,k) = n-k-p_{s} + k p_{s}^{2}$ so that 
\begin{equation}
\tilde{f} (p_{s}, k) = c^2(n-k) \frac{u(p_s,k)}{v(p_s,k)^2}.  
\end{equation}
It is straightforward to establish that $u \geq 0$, $v \geq 0$ under the given assumptions.  Taking the derivative of $\tilde{f}$ w.r.t.\ $k$: 
\begin{IEEEeqnarray}{rCl}
\frac{\drm }{\drm k} \tilde{f}(p_s,k)
= \frac{c^{2}}{v(p_{s}, k)^{3}}  \cdot h(p_{s},k),
\end{IEEEeqnarray}
where
\begin{IEEEeqnarray}{rCl}
\IEEEeqnarraymulticol{3}{l}
{h(p_{s},k) \equiv - u v + (n-k) v (n p_{s}^{2} - 2 p_{s} - (n-2)) 
}\nonumber \\* \qquad \qquad ~~
& & \negmedspace{} - 2 (n-k) u (p_{s}^{2}-1).\IEEEeqnarraynumspace
\end{IEEEeqnarray}
Therefore showing $\frac{\drm \tilde{f}}{\drm k} > 0$ is equivalent to showing $h (p_{s},k) > 0$.  Toward this, observe the third summand in $h (p_{s},k)$ can be split evenly to be combined with the first and second summands, thus
\begin{equation}
h (p_{s},k) = (1 - n p_{s}) \left( u p_{s} + (n-k) (1-p_{s})^{2} \right) > 0.
\end{equation}
It follows that, for fixed $p_{s} \in (0, \frac{1}{n})$, $\tilde{f}(p_{s}, k)$ is maximized at $k = n-1$.  The global maximum of $\tilde{f}(p_s,k)$ is obtained by further optimizing $\tilde{f}(p_{s}, n-1)$ over $p_{s} \in (0, \frac{1}{n})$. Setting $k = n-1$ in \eqref{eqn:fobj} gives
\begin{equation} 
\label{eq:fpkn-1}
\tilde{f}(p_{s}, n-1) = c^{2} (n-1) \frac{n p_{s}^{2} - 2 p_{s} + 1}{\left( (n-1) p_{s}^{2} - p_{s} + 1 \right)^{2}},
\end{equation}
for which 
\begin{equation}
\left. \frac{\partial \tilde{f}(p_{s},k)}{\partial p_s} \right|_{k=n-1} = - 2 \frac{c^{2} (n-1)^{2}}{v} p_{s} \left( n p_{s}^{2} - 3 p_{s} + 1 \right) < 0.
\end{equation}
The inequality holds since the quadratic $n p_{s}^{2} - 3 p_{s} + 1$ can be verified to be positive for $n \geq 2$ and $p_{s} \in (0, \frac{1}{n})$. Therefore, the maximum (indeed supremum) of $\tilde{f}(p_{s}, n-1)$ is obtained when $p_{s} \to 0$ (meaning in the limit $\ebf_{i}$ is the maximizer although $\ebf_{i}$ itself does not satisfy \eqref{eq:condIfTwoDistinctNonzeroValues}), which according to \eqref{eq:fpkn-1} is $(n-1)c^{2}$. These monotonicity properties are illustrated in Fig.\ \ref{fig:monotonicitiesLsi-withLegend}.

\begin{figure}[!ht]
\centering
\includegraphics[width=0.45\textwidth]{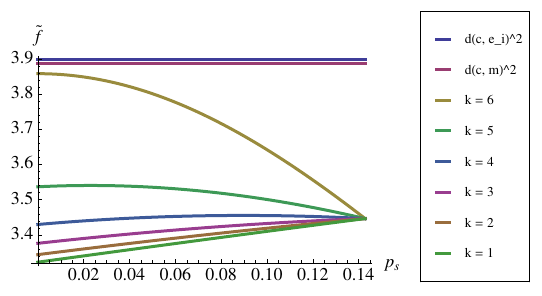}
\caption{Monotonicity of function $\tilde{f}(p_s,k)$ w.r.t.\ $k$ and $p_{s}$ for $\pbf \in\Pmc^{(2)}_{a,2}$ (Step 3) in the spherical inner bound $\Lsi$ proof. Horizontal axis denotes $p_{s} \in (0, 1/n)$ with $n=7$. Here $c = 0.99 c_{\rm in}^{*} < c_{\rm in}^{*}$ so $d(\cbf, \ebf_{i}) > d(\cbf, \mbf)$. Given $p_{s}$, $\tilde{f}$ is increasing in $k$; when $k = n-1$, $\tilde{f}$ is decreasing in $p_{s}$ so the supreme of $\tilde{f}$ over the set $\Pmc^{(2)}_{a,2}$ is achieved as $p_{s} \to 0$, which is $(n-1)c^{2} \approx 3.85807$.}
\label{fig:monotonicitiesLsi-withLegend}
\end{figure}
Observe when we maximize $\tilde{f}(p_s,k)$ we effectively enlarge the feasible set from $\Pmc^{(2),*}_{a,2}$ to $\Pmc^{(2)}_{a,2}$ because we do not check whether \eqref{eq:condIfTwoDistinctNonzeroValues} is satisfied. Recall $f$ is identically equal to $\tilde{f}$ only for $\pbf \in \Pmc^{(2),*}_{a,2}$ because $\tilde{f}$ is derived from $f$ by applying \eqref{eq:condIfTwoDistinctNonzeroValues}. Therefore we have 
\begin{IEEEeqnarray}{rCl}
f(\pbf') = \tilde{f}(\pbf') \leq \max_{\pbf \in \Pmc^{(2),*}_{a,2}} \tilde{f}(\pbf) & \leq & \max_{\pbf \in \Pmc^{(2)}_{a,2}} \tilde{f}(\pbf) \nonumber \\
& = & (n-1) c^{2}, \forall \pbf' \in \Pmc^{(2),*}_{a,2}.\IEEEeqnarraynumspace
\label{eq:si_ub}
\end{IEEEeqnarray}

Summarizing, so far we have shown, suppose there exists a potential extremizer $\pbf$ whose components are all non-zero but not all identical, then in order to satisfy the first-order KKT necessary conditions, the original objective function $f$ evaluated at such a point is \textit{upper bounded} by $\tilde{f}(0,n-1) = (n-1) c^{2}$. Now since $f(\ebf_{i}) = d(\cbf, \ebf_{i})^{2} = (n-1)c^{2} + (c-1)^{2}$, this means no $\pbf \in\Pmc^{(2)}_{a,2}$ can achieve a higher objective value than $\ebf_{i}$ does (i.e., case $\Pmc^{(2)}_{b,1}$) in terms of globally maximizing the original objective function. In fact, this property does not depend on choosing the thresholding $c_{\rm in}^{*} = (1 - n m^2)/(2(1 - n m))$. This property is useful in Step 4 below.

\textbf{Step 4}: Any $\pbf \in \Pmc^{(2)}_{b,2}$ cannot be a global maximizer.  Fix $\pbf \in \Pmc^{(2)}_{b,2}$ and let $s = n - |\Zmc(\pbf)| \in \{2,\ldots,n-1\}$ be the number of non-zero components.  Evaluating the original objective function, \eqref{eq:si_opt_second}, for such a point yields
\begin{equation}
f (\pbf) = (n-s) (c - 0)^{2} + \sum_{i=1}^{s} \left(c - p_{i} \prod_{j \neq i}^{s} (1-p_{j}) \right)^{2}.
\end{equation}
For each given $s$, $f(\pbf)$ is maximized if and only if the second summand above is maximized.  Maximizing the above second summand can be thought of as performing the same optimization problem in an $s$-dimensional space where the $s$-vector $\pbf$ duplicates all the $s$ non-zero components from the original $n$-vector $\pbf$.  Then, one may view this $s$-vector with no zero components as a member of $\Pmc^{(2)}_{a}$, but with the dimension reduced from $n$ to $s$. There are two possibilities: this point is either in $\Pmc^{(2)}_{a,2}$ or in $\Pmc^{(2)}_{a,1}$. 

Consider the first possibility, i.e., $\Pmc^{(2)}_{a,2}$.  Based on the analysis of this case in Step 3 (with the dimension reduced from $n$ to $s$), and the upper bound \eqref{eq:si_ub} in particular, it follows that
\begin{equation}
\sum_{i=1}^{s} \left(c - p_{i} \prod_{j \neq i}^{s} (1-p_{j}) \right)^{2} \leq (s-1)c^{2},
\end{equation}
and hence $f (\pbf) \leq (n-s) (c - 0)^{2} + (s-1)c^{2} = (n-1) c^{2}$, which is the same upper bound for candidates in case $\Pmc^{(2)}_{a,2}$ in the original $n$-dimensional space. It follows that, in this reduced dimensional space, points in $\Pmc^{(2)}_{a,2}$ cannot achieve a higher objective value than that achieved by the points $\ebf_{i}$ in the original space.

Consider the second possibility, i.e., $\Pmc^{(2)}_{a,1}$, namely the all-rates-equal point in this $s$-dimensional space.  There are two subcases: $i)$ $c \leq c_{\rm in}^{*}(s) = (1 - s m(s)^2)/(2(1 - s m(s)))$, and $ii)$ $c > c_{\rm in}^{*}(s)$. Note we write $c_{\rm in}^{*}(s)$ to highlight it is a function of $s$, the corresponding dimension. Case $i)$ can be skipped, due to \eqref{eq:cinstarOneToOne} and the observation that the $\ebf_{i}$ in this $s$-dimensional space is also the $\ebf_{i}$ in the original $n$-dimensional space. Recall, $\ebf_{i}$ (in the set $\Pmc_{b,1}^{(2)}$) will be addressed in the final step. For case $ii)$ we now directly show the all-rates-equal point in this $s$-dimensional space cannot achieve a higher objective value than the all-rates-equal point in the original space does.   First, it is straightforward to establish the inequality 
\begin{IEEEeqnarray}{rCl}
\IEEEeqnarraymulticol{3}{l}
{n \left( c - \frac{1}{n} \left(1 - \frac{1}{n} \right)^{n-1}\right)^{2} = f(\mbf(n)) 
}\nonumber \\* \quad
& \geq & f(\mbf(s)) = (n-s) c^{2} + s \left( c - \frac{1}{s} \left(1 - \frac{1}{s} \right)^{s-1}\right)^{2} \IEEEeqnarraynumspace
\label{eq:sinproof1}
\end{IEEEeqnarray}
holds if and only if 
\begin{equation}
\label{eq:sinproof2}
2 c (s m(s) - n m(n)) \geq s m^{2}(s) - n m^{2}(n).
\end{equation}
Since $c > c_{\rm in}^{*}(s)$, \eqref{eq:sinproof2} is equivalent to
\begin{IEEEeqnarray}{rCl}
\IEEEeqnarraymulticol{3}{l}
{\left[ s m(s) (1 - m(s)) - n m(n) (1 - m(n)) \right] 
}\nonumber \\* \qquad \quad ~~
& & \negmedspace{} +  s n m(s) m(n) (m(s) - m(n)) \geq 0.
\label{eq:desiredInequalityReducedDim}
\end{IEEEeqnarray}
Since $s n m(s) m(n) (m(s) - m(n)) \geq 0$, to show \eqref{eq:desiredInequalityReducedDim} it suffices (as the terms in the brackets would be non-negative) to show the function $g(n) \equiv n m(n) (1 - m(n))$ is monotone decreasing in $n$ for $n \geq 3$ (recall $2 \leq s \leq n-1$). Toward this we find the derivative of $g(n)$ as 
\begin{IEEEeqnarray}{rCl}
\frac{\drm g(n)}{\drm n}
 =  \frac{1}{2 n^{2}} \left(1-\frac{1}{n}\right)^{n-1} \cdot \tilde{g}(n),
\end{IEEEeqnarray}
where
\begin{IEEEeqnarray}{rCl}
\tilde{g}(n) & \equiv & -2 \left(1 - \frac{1}{n} \right)^{n-1}+ 2n \nonumber \\
& & \negmedspace{} + 2 n^{2} \left(1 - \frac{2}{n} \left(1-\frac{1}{n}\right)^{n-1}\right) \log\left(1 - \frac{1}{n}\right).\IEEEeqnarraynumspace
\end{IEEEeqnarray}
It now suffices to show $\tilde{g}(n) < 0$. For $n = 3, 4$ this can be verified; for $n \geq 5$ we apply the inequality $\log(1+x) \leq x - \frac{x^{2}}{2}$ for all $x \in (-1,0]$ and get
\begin{equation}
\label{eq:tildegOfnsiAndso}
 \tilde{g}(n) \leq 2 \left(1 - \frac{1}{n} \right)^{n-1} \left(1 + \frac{1}{n} \right) - 1 < 0,
\end{equation}
where the second inequality follows from the monotonicity in $n$ of the upper bound (in \eqref{eq:tildegOfnsiAndso}) on $\tilde{g}(n)$. Therefore we have shown the desired inequality \eqref{eq:desiredInequalityReducedDim}. This means that, even if a point is from $\Pmc^{(2)}_{a,1}$, it cannot be a global maximizer as it cannot achieve a higher objective value than $\ebf_{i}$ (case $i)$) and/or $\mbf$ (case $ii)$) in the original $n$-dimensional space does. This concludes Step 4.

\textbf{Finally}, we are left with only cases $\Pmc^{(2)}_{a,1}$ and $\Pmc^{(2)}_{b,1}$. We can verify by checking the KKT conditions that $\mbf$ is always eligible to be a local extremizer, while $\ebf_{i}$ is eligible to be a local maximizer if and only if $c \leq 1$. Therefore, we conclude the global maximum of the original optimization problem can be obtained by evaluating and comparing $f$ at two points $\mbf$, $\ebf_{i}$. Furthermore, recall the objective function is defined as $d(\cbf, \xbf)^{2}$ for $\xbf \in \partial \Lambda$. Then as a consequence of \eqref{eq:cinstarOneToOne}, we can actually conclude in a more general manner: the global maximum occurs at $i)$ any $\ebf_i$ when $c < c_{\rm in}^*$, $ii)$ any $\ebf_i$ and $\mbf$ when $c = c_{\rm in}^*$, and $iii)$ $\mbf$ when $c > c_{\rm in}^*$.  See Fig.\ \ref{fig:sphericalInnerBoundsUnderDifferentcs} for an illustration. For $\Lsi^{*}$, the global maximizers are both $\ebf_{i}$ and $\mbf$ giving the maximum of $f$ as $\left(n-1\right)c_{\rm in}^{*2} + \left(c_{\rm in}^{*} - 1\right)^{2} = $ $n\left(c_{\rm in}^{*}-m\right)^{2}$, as desired in \eqref{eq:Lambdasi}.
\end{IEEEproof}

\begin{figure}[!ht]
\centering
\includegraphics[width=0.3\textwidth]{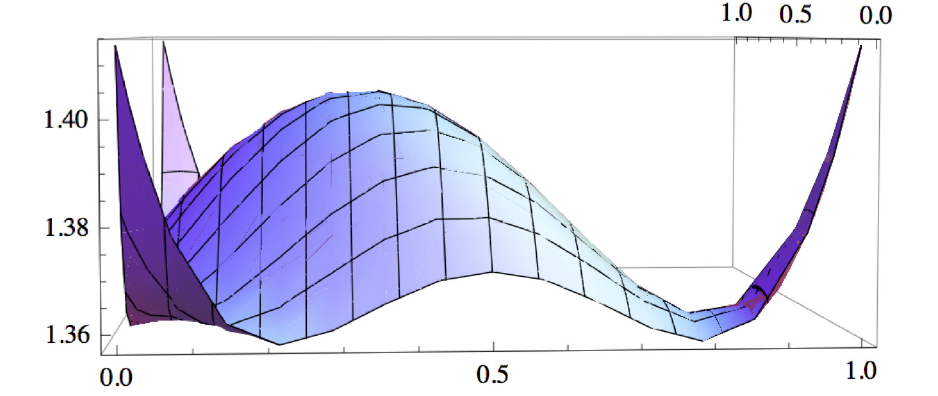}
\includegraphics[width=0.3\textwidth]{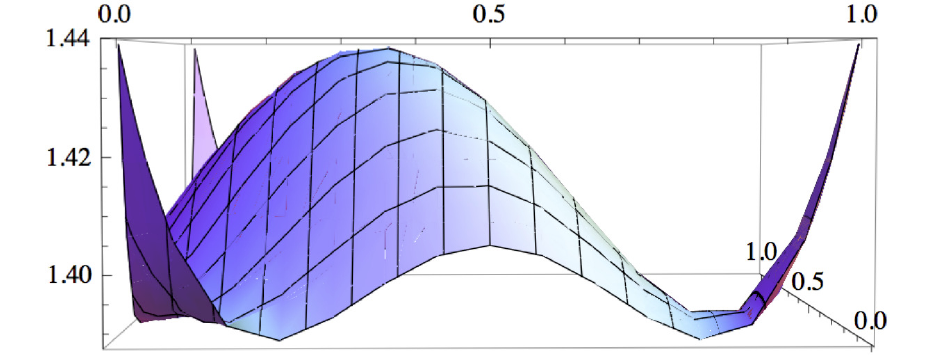}
\includegraphics[width=0.3\textwidth]{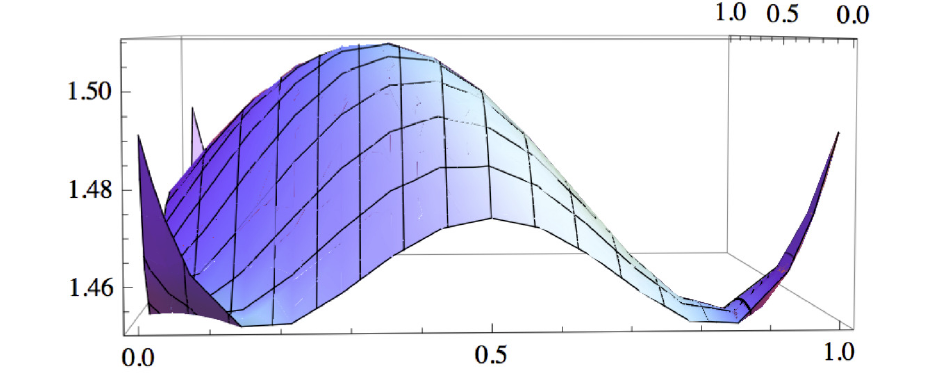}
\caption{Original objective function $f(\pbf)$ for $\Lsi$ when $n=3$ for various $c$.  Horizontal axes are $p_{2}, p_{3}$.  
{\bf Top:} $c = 0.99 c_{\rm in}^{*}$ with global maximizers: $\ebf_{i}$.
{\bf Middle:} $c = c_{\rm in}^{*}$, with global maximizers: $\ebf_{i}$ \& $\mbf$.
{\bf Bottom:} $c = 1.02 c_{\rm in}^{*}$, with global maximizer:$\mbf$.}
\label{fig:sphericalInnerBoundsUnderDifferentcs}
\end{figure}

\begin{remark} \label{remark:LsiStarBetterThanLpiStar}
Since the hyperplane inducing the optimal polyhedral inner bound $\Lpi^{*}$ is a supporting hyperplane of the convex body $\Bmc(\cbf, r_{\rm in}(c))$, due to the constructions of $\Lpi^{*}$ and $\Lsi$ it follows that $\Lsi$ is always tighter than $\Lpi^{*}$.
\end{remark}

We now proceed to the spherical outer bound. There are many similarities between the definitions, propositions and proof techniques for the spherical inner and outer bounds. In both cases there exists a set inclusion relationship which implies the optimal bound arises when $c$ is chosen to be the minimum possible.  

\begin{definition} 
\label{def:Lso}
$\Lso(c) \equiv \Smc \setminus \Bmc(\cbf, r_{\rm out}(c))$, where the center of the ball $\cbf = c \mathbf{1}$ for all $c \geq 1$, and its radius $r_{\rm out}(c) \equiv d(\cbf, \ebf_{i}) = \sqrt{(c-1)^{2} + (n-1)c^{2}}$.
\end{definition}

\begin{proposition}[spherical outer bound]
\label{prop:Lambdaso}
For each $c \geq 1$, the set $\Lso(c)$ is an outer bound on $\Lambda$ for $n \geq 2$. Among these, the tightest is given when $c = c_{\rm out}^{*} = 1$:
\begin{equation}
\Lso^{*} =  \Lso(c_{\rm out}^{*})  =  \left\{ \xbf \in \Smc : \|\xbf - \mathbf{1}\| \geq \sqrt{n-1} \right\}.
\end{equation}
\end{proposition}

\begin{IEEEproof}
By Def.\ \ref{def:Lso}, in order to induce an outer bound we must have $c > m$.  Moreover, $d (\cbf, \mbf) > d (\cbf, \ebf_{i}) = r_{\rm out}$, which is equivalent to $c > (1 - n m^2)/(2(1 - n m))$.  This means there remain two possible intervals for $c$: $i)$ $(1 - n m^2)/(2(1 - n m)) < c < 1$ and $ii)$ $c \geq1$. For each one we investigate whether a ball with parameter $c$ in that interval induces a valid outer bound on $\Lambda$.

For case $i)$, we can compute that $\ebf_1,\xbf_b \in \partial \Bmc(\cbf, r_{\rm out})$ for $\xbf_b \equiv (2c-1) \ebf_{1}$, i.e., the boundary of the ball intersects the first coordinate axis at these two points.  As $c > (1 - n m^2)/(2(1 - n m)) > \frac{1}{2}$, the line segment $\overline{\xbf_{b} \ebf_{1}} \subseteq \Lambda $ due to $\Lambda$'s coordinate convexity.  Furthermore, since the open line segment $\overline{\xbf_{b} \ebf_{1}} \subseteq \Bmc(\cbf, r_{\rm out})$, we find $\Lambda \not\subseteq \Smc \setminus \Bmc(\cbf, r_{\rm out})$ namely $\Bmc(\cbf, r_{\rm out})$ does not induce a valid outer bound.

It remains to investigate case $ii)$. In the rest of this proof we first show every $c$ in this category gives a valid outer bound and furthermore, $c = c_{\rm out}^{*} = 1$ yields the tightest bound.  We first show $c = c_{\rm out}^{*} = 1$ yields the smallest set, i.e., we show $\Lso^{*} \subseteq \Lso(c)$ for each $c \geq 1$.
Note the equivalence
\begin{IEEEeqnarray}{rCl}
\IEEEeqnarraymulticol{3}{l}
{\Lso^{*} \subseteq \Lso(c) 
}\nonumber \\* \quad
& \iff &  \Smc \setminus \left( \Bmc(\mathbf{1}, \sqrt{n-1}) \cap \Smc \right) \subseteq \Smc \setminus \left( \Bmc(\cbf, r_{\rm out}(c)) \cap \Smc \right) \nonumber \\
& \iff & \Bmc(\cbf, r_{\rm out}(c)) \cap \Smc \subseteq \Bmc(\mathbf{1}, \sqrt{n-1}) \cap \Smc.
\end{IEEEeqnarray}
Therefore we seek to prove: $\forall \xbf \in \Smc$, if $d(\cbf, \xbf)^{2} < r_{\rm out}(c)^{2}$ (i.e., $\xbf \in \Bmc(\cbf, r_{\rm out}(c)) \cap \Smc$), then $d(\mathbf{1}, \xbf)^{2} < \left(\sqrt{n-1}\right)^{2}$ (i.e., $\xbf \in \Bmc(\mathbf{1}, \sqrt{n-1}) \cap \Smc$). Suppose $\xbf \in \Smc$ is such that $d(\cbf, \xbf)^{2} < r_{\rm out}(c)^{2}$.  Then, we compute:
\begin{IEEEeqnarray}{rCl}
\IEEEeqnarraymulticol{3}{l}
{d(\mathbf{1}, \xbf)^{2} = \sum_{i=1}^{n} (1-x_{i})^{2} =  \sum_{i=1}^{n} \left( \left( c - x_{i} \right) - \left( c - 1 \right) \right)^{2}
}\nonumber \\* \quad
& = & \sum_{i=1}^{n} \left( c - x_{i} \right)^{2} + n (c-1)^{2} - 2 (c-1) \left( n c - \sum_{i=1}^{n} x_{i} \right) \nonumber \\
& < & (c-1)^{2} + (n-1) c^{2} + n (c-1)^{2} - 2 (c-1) (n c - 1) \nonumber \\
& = & n-1,
\end{IEEEeqnarray}
where the inequality follows from $d(\cbf, \xbf)^{2} < r_{\rm out}(c)^{2}$ and $\sum_{i=1}^{n} x_{i} \leq 1$. This shows the desired set inclusion, meaning $\Lso^{*}$ is the smallest set among $\{\Lso(c), c \geq 1\}$. 

It remains to show that $\Lso^*$ is a valid outer bound on $\Lambda$.  For any $\xbf \in \Lambda = \Lambda_{\rm eq}$ we must show $\xbf \in \Lso^{*}$, namely $\xbf$ is outside the open ball $\Bmc(\cbf,1)$, or equivalently $\min_{\xbf \in \Lambda} \|\mathbf{1} - \xbf \|^{2} \geq n-1$.  This latter expression may be cast as an optimization problem w.r.t. $\pbf$:
\begin{equation} 
\label{eq:LsostarObjFunc}
\min_{\pbf \in [0,1]^n} f(\xbf(\pbf)) \equiv \sum_{i=1}^{n} \left(1 - p_i \prod_{j \neq i}(1-p_j) \right)^2 \geq n-1.
\end{equation}
Observe if any component of $\pbf$ equals $1$ or if $\pbf = \mathbf{0}$ then $f \geq n-1$ immediately holds. So below we assume $\pbf < \mathbf{1}$ and $\pbf$ has non-zero component(s).  Recall we defined $\Vmc(\pbf) = \{ a \in (0,1] : \exists i \in [n] : p_i = a\}$ in the proof of Prop.\ \ref{prop:Lambdasi} as the set of non-zero values taken by a vector $\pbf$.  We now categorize based on how many distinct non-zero values the components of $\pbf$ assume: $a)$ $|\Vmc(\pbf)| > 1$ or $b)$ $|\Vmc(\pbf)| = 1$. 

Consider first case $a)$ ($|\Vmc(\pbf)| > 1$), i.e., $\pbf$ has two or more distinct non-zero component values, say $0 < p_k < p_l < 1$.  We will show that all such $\pbf$'s cannot be local extremizers due to the violation of Karush-Kuhn-Tucker (KKT) conditions required for optimality (note regularity is guaranteed in this case).  An equivalent form of the KKT stationarity condition is that 
\begin{equation}
\frac{1}{2} \left( \frac{\partial f}{\partial p_k} - \frac{\partial f}{\partial p_l}\right) \left( 1-p_{k}\right)\left( 1-p_{l}\right)  = 0, ~ 0< p_k < p_l < 1, 
\end{equation}
which, after some algebra, may be shown to be equivalent to:
\begin{equation}
\left(1-p_{l}\right) \left( \eta - \pi_{k} (1-p_{k}\pi_{k})\right) = \left(1-p_{k}\right) \left( \eta - \pi_{l} (1-p_{l}\pi_{l})\right),
\end{equation}
where $\eta \equiv \sum_{i} \pi_{i} (1-p_{i}\pi_{i}) p_{i}, ~\pi_{i} = \pi_{i}(\pbf) \equiv \pi (\pbf) /(1-p_{i})$. Note $\eta$ can be interpreted as the expectation of a discrete random variable $Z$ with support $\{\pi_{i}(1-p_{i}\pi_{i}), i \in [n]\}$ and associated PMF $\Pbb\left(Z = \pi_{i}(1-p_{i}\pi_{i}) \right) = p_{i}$ for each $i \in [n]$.  Therefore, stationarity will not be satisfied as long as we can choose indices $k', l'$ such that $0< p_{k'} < p_{l'} < 1$, and $\eta$ lies strictly between $\pi_{k'} (1-p_{k'}\pi_{k'})$ and $\pi_{l'} (1-p_{l'}\pi_{l'})$.  But, we {\em can} always find such indices since, following Lem.\ \ref{lem:monoLso}, we can show the ordering: $\pi_{k'} (1-p_{k'}\pi_{k'}) < \pi_{l'} (1-p_{l'}\pi_{l'})$. This rules out the possibility that an extremizer can come from case $a)$.

Consider next case $b)$ ($|\Vmc(\pbf)| = 1$), i.e., $\pbf$ has only one distinct non-zero component value (we call such $\pbf$ ``quasi-uniform'').  Lem.\ \ref{lem:QULambdaso} below states that for any such $\pbf$, the objective function $f(\pbf) \geq n-1$, the desired lower bound in \eqref{eq:LsostarObjFunc}.

These two cases establish the validity of the inequality \eqref{eq:LsostarObjFunc}, and thereby establish the fact that $\Lso^*$ is a valid outer bound for $\Lambda$.
\end{IEEEproof}

The following two lemmas are used in the preceding proof of Prop.\ \ref{prop:Lambdaso}. The proof of Lem.\ \ref{lem:monoLso} is straightforward and is omitted.

\begin{lemma} \label{lem:monoLso}
If two non-zero components of $\pbf$ satisfy $p_{l} > p_{k}$, then for $\xbf = \xbf (\pbf)$ as defined in \eqref{eq:xOfp}:
\begin{equation}
x_{l} > x_{k}, ~~~ \pi_{l} > \pi_{k}, ~~~ \frac{x_{l}}{x_{k}} > \frac{p_{l}}{p_{k}}, ~~~ x_{l}-x_{k} < p_{l} - p_{k}.
\end{equation}
\end{lemma}

\begin{lemma}
\label{lem:QULambdaso}
Fix $t \in (0,1]$ and $k \in [n]$.  Suppose $\pbf < \mathbf{1}$ ($\pbf \neq \mathbf{0}$) takes only one non-zero value (i.e., $|\Vmc(\pbf)|=1$, and $p_i \in \{0,t\}$ for $i \in [n]$), and this value is taken by $k$ components of $\pbf$.  Then $f(\pbf) \geq n-1$ (for $f$ in \eqref{eq:LsostarObjFunc}), with equality if and only if $k=t=1$, i.e., $f(\pbf) = n-1$ if and only if $\pbf \in \{\ebf_i\}_{i=1}^n$.
\end{lemma}

\begin{IEEEproof}
W.l.o.g.\ let $\pbf = t \cdot \sum_{i=1}^{k} \ebf_{i}$ for $k \in [n]$, $t \in (0,1]$. Substitution of such a $\pbf$ into \eqref{eq:LsostarObjFunc} yields the following equivalent inequality
\begin{equation} 
\label{eq:LsostarObjFuncVer2}
t (1-t)^{k-1} \leq 1 - \sqrt{1 - 1 / k}, 
\end{equation}
meaning the lemma will be established if we can show \eqref{eq:LsostarObjFuncVer2} holds for all valid $(t,k)$, and holds with equality if and only if $k=t=1$.  The inequality \eqref{eq:LsostarObjFuncVer2} is easily verified to hold strictly for $a)$ $k = 1 \neq t$ and $b)$ $t = 1 \neq k$.  If $k=t=1$ (namely $\pbf = \ebf_{1}$) the original objective function in \eqref{eq:LsostarObjFunc} evaluates to $n-1$, the desired minimum.  It remains to study the case $k \in \{2,\ldots,n\}$ and $0 < t < 1$.  Define $g(t) \equiv t (1-t)^{k-1}$. The only stationary point of $g$ on $t \in [0,1)$ is $t^{*} = 1/k$, at which the second derivative can be verified to be strictly negative, meaning $t^{*}$ is the unique maximizer. And hence we need to show \eqref{eq:LsostarObjFuncVer2} when $t = 1/k$, namely $\left(1 - 1/k\right)^{k-1} \leq  k \left(1-\sqrt{1- 1/k} \right)$. The derivative of its LHS can be shown to be negative using the inequality $\log (1+x) \leq x$ for $x > -1$. Thus the sequence $\{(1-1/k)^{k-1}\}$ is upper bounded by $ \left(1- 1/ 2\right)^{2-1} = 1 / 2$. On the other hand, using AM-GM inequality $\sqrt{1-1/k} < (1-1/k + 1)/2$, one can see the RHS is strictly lower bounded by $1 / 2$. This shows the desired inequality \eqref{eq:LsostarObjFuncVer2}, thus proving this lemma.
\end{IEEEproof}

\begin{remark}
The Cauchy-Schwarz inequality gets close to proving the desired inequality \eqref{eq:LsostarObjFunc}, but is insufficient by itself (note $\sum_{i} x_{i} \leq 1$ as $\Lambda \subseteq \Smc$):
\begin{IEEEeqnarray}{rCl}
\sum_{i=1}^{n} \left(1-x_i\right)^2 & \geq & \frac{\left(\sum_{i=1}^{n} 1 \cdot (1-x_i)\right)^2}{\sum_{i=1}^{n} 1^{2}} \nonumber \\
& = & \frac{\left(n - \sum_{i=1}^{n} x_i\right)^2}{n} \geq \frac{(n-1)^2}{n},
\end{IEEEeqnarray}
which is slightly weaker than the bound of $n-1$ required to show \eqref{eq:LsostarObjFunc}.
\end{remark}

The optimal spherical inner and outer bounds $\Lsi^{*}$, $\Lso^{*}$ together with $\partial \Lambda$ are shown in Fig.\ \ref{fig:LsiLso-Combined-Marked-ver20140802} for $n=2$ and $3$.

It seems hard to obtain the volume of these spherical bounds in closed-form for arbitrary $n$. Essentially, the problem is one of integrating over the intersection between a (solid) hypersphere and $[0,1]^n$.  It is natural to attempt to bound the volume. Below, we illustrate such an attempt using $\Lso^{*}$ as an example.

We take a probabilistic approach.  As the volume of the unit box $[0,1]^{n}$ always equals $1$, and as $\Lso^{*} \subseteq [0,1]^{n}$, it follows that its volume can be interpreted as the probability that a point uniformly distributed over $[0,1]^{n}$ falls into the set $\Lso^{*}$.  More precisely, for i.i.d.\ $\mathrm{Unif}[0,1]$ random variables (RV) $y_{1}, \ldots, y_{n}$,
\begin{IEEEeqnarray}{rCl}
\IEEEeqnarraymulticol{3}{l}
{\vol (\Lso^{*})
}\nonumber \\* \quad\!
& = & \Pbb\left( \ybf \in \Smc, ~ \ybf \notin \Bmc(\mathbf{1}, \sqrt{n-1}) \right)
\stackrel{(a)}{=} \Pbb\left( \ybf \notin \Bmc(\mathbf{1}, \sqrt{n-1}) \right) \nonumber \\
& = & \Pbb\left( \sum_{i} \left( 1 - y_{i} \right)^{2} \geq n-1 \right) 
\stackrel{(b)}{=} \Pbb\left( \sum_{i} y_{i}^{2} \geq n-1 \right),\IEEEeqnarraynumspace
\label{eqn:volLsoProb}
\end{IEEEeqnarray}
where $(a)$ follows from Lem.\ \ref{lem:boundingVolLsoUsingChernoff} given below and $(b)$ is due to the observation that $1 - y_{i}$ are i.i.d. $\mathrm{Unif}[0,1]$ RV's too. We note the \textit{uniform sum distribution} (also known as \textit{Irwin-Hall distribution}), $\sum_i y_i$, has a known closed-form density function, yet this does not seem to be the case for $\sum_i y_i^2$. Then one natural thing to do is to bound this tail probability.  A typical form of the Chernoff bound states that for a random variable $Z$ (usually expressed as a sum of independent RVs), an upper bound on the (upper) tail probability is $\Pbb\left(Z \geq t \right) \leq \inf_{s \geq 0} \erm^{-s t} \Ebb\left[\erm^{s Z}\right]$.  Substituting $\sum_{i} y_{i}^{2}$ for $Z$ yields:
\begin{IEEEeqnarray}{rCl}
\vol (\Lso^{*})  & \leq & \inf_{s \geq 0} \erm^{-s (n-1)} \Ebb\left[ \erm^{s \sum_{i} y_{i}^{2} }\right] \nonumber \\
& = & \inf_{s \geq 0} \erm^{-s (n-1)} \prod_{i} \Ebb\left[\erm^{s y_{i}^{2}}\right] \nonumber \\
& = & \inf_{s \geq 0} \erm^{-s (n-1)} \left( \int_{0}^{1} \erm^{s y_{i}^{2}} \drm y_{i} \right)^{n} .
\end{IEEEeqnarray}
The minimizer $s^{*}$ is hard to be obtained in closed-form.  Worse still, the numerically optimized upper bound is not close to the actual tail probability (i.e., the volume of $\Lso^{*}$).\footnote{For $n = 2$ through $7$, the ratios between the optimized upper bound and the true volume $\vol (\Lso^{*})$ are $3.4924$, $5.1904$, $6.0964$, $6.6436$, $7.0569$, and $7.6899$ respectively.}

\begin{lemma} \label{lem:boundingVolLsoUsingChernoff}
$[0,1]^{n} \setminus \Smc \subseteq \Bmc\left(\mathbf{1}, \sqrt{n-1}\right)$. In words this says the unit box with the unit simplex subtracted lies completely inside the ball $\Bmc\left(\mathbf{1}, \sqrt{n-1}\right)$.
\end{lemma}
\begin{IEEEproof}
Given $\xbf \in [0,1]^{n}, \sum_{i} x_{i} > 1$, we need to show $\sum_{i} \left( 1 - x_{i} \right)^{2} < n-1$, which is easily verifiable since
\begin{IEEEeqnarray}{rCl}
\IEEEeqnarraymulticol{3}{l}
{\sum_{i} \left( 1 - x_{i} \right)^{2} 
}\nonumber \\* \quad
& = & n - 2 \sum_{i} x_{i} + \sum_{i} x_{i}^{2} \leq n - 2 \sum_{i} x_{i} + \sum_{i} x_{i} \nonumber \\
& = & n - \sum_{i} x_{i} < n-1.
\end{IEEEeqnarray}
Note this lemma can be equivalently stated as $[0,1]^{n} \setminus \Bmc\left(\mathbf{1}, \sqrt{n-1}\right) \subseteq \Smc$ which implies if a point is from $[0,1]^{n}$ but not in $\Bmc\left(\mathbf{1}, \sqrt{n-1}\right)$ then it's guaranteed to be in $\Smc$. This observation is used step $(a)$ in \eqref{eqn:volLsoProb}.
\end{IEEEproof}

\begin{figure}[!h]
\centering
\includegraphics[width=0.4\textwidth]{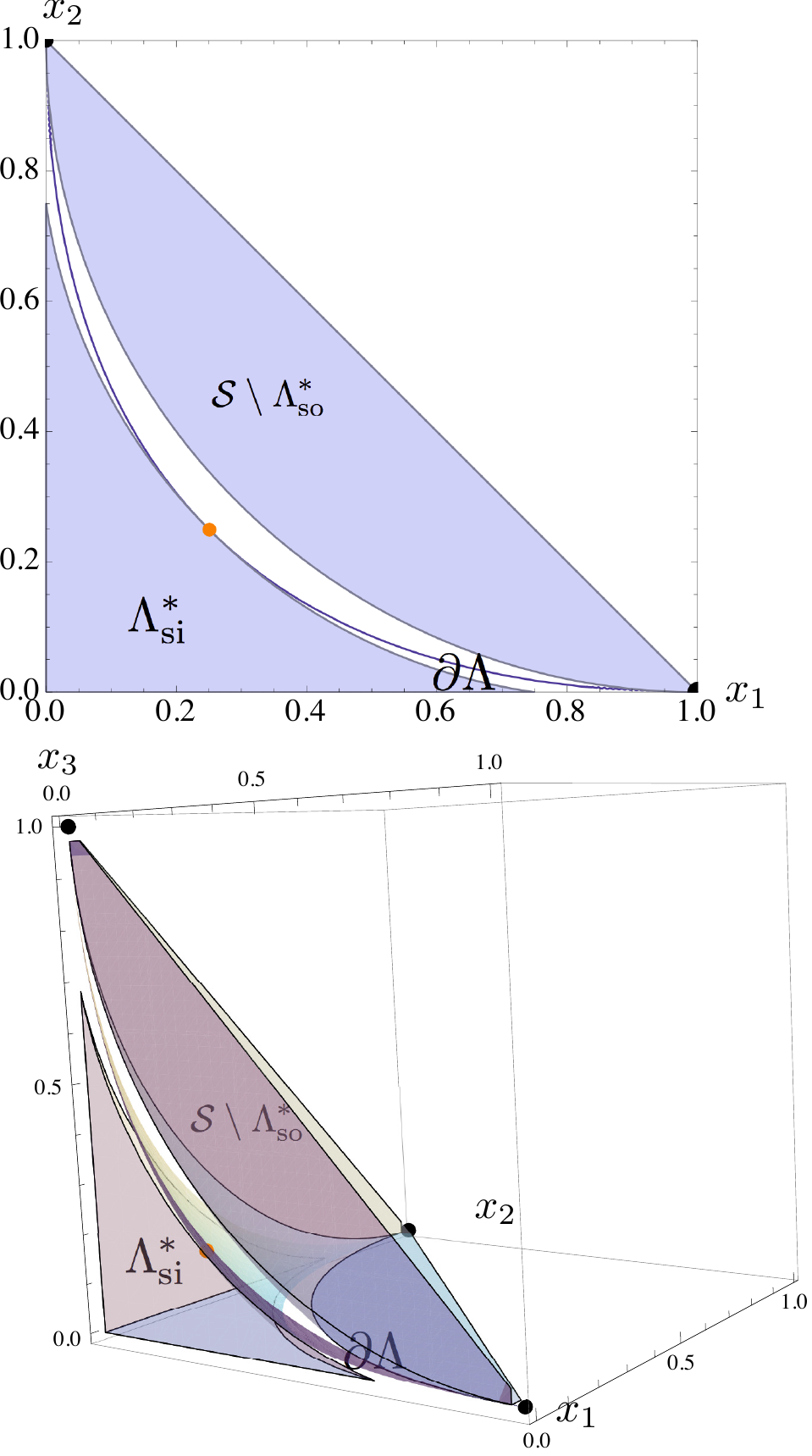}
\caption{Optimal spherical bounds for $n=2$ (top) and $3$ (bottom): inner bound and the complement (w.r.t.\ $\Smc$) of outer bound are shown in solid; in between sits $\partial \Lambda$. Also shown are the all-rates-equal point $\mbf$ (orange) and corner points $\ebf_{i}$'s (black).}
\label{fig:LsiLso-Combined-Marked-ver20140802}
\end{figure} 

\section{Ellipsoid inner and outer bounds on $\Lambda$} 
\label{sec:eiAndeo}

We now turn to the third, and final, class of bounds on $\Lambda$.   In this section we establish inner and outer bounds, each induced by a parameterized family of ellipsoids.  This section is organized into three subsections.  First, in \S \ref{ssec:simplification} we prove three results: $i)$ the set of ellipsoids that inherit all the permutation symmetries of $\Lambda$ are characterized by three scalars $(c, a_{1}, a_{2})$ (Prop.\ \ref{prop:InvariantEllipsoid}), $ii)$ the sufficiency of working only with $\partial \Lambda$ for the purpose of proving the correctness of the induced bound (Props.\ \ref{prop:outerboundReducedtoBoundaryChecking} and \ref{prop:innerboundReducedtoBoundaryChecking}), and $iii)$ a property of a local extremizer from $\partial \Lambda$ (Prop.\ \ref{prop:atMostTwoComponentValues}). Next, in \S \ref{ssec:derivationAndproofEllipsoidalBounds} we present the parameterized families of ellipsoid inner (Prop.\ \ref{prop:Ei}) and outer (Prop.\ \ref{prop:Eo}) bounds.  The derivation is based on the Karush-Kuhn-Tucker (KKT) optimality conditions.  Finally, in \S \ref{ssec:alternateProofEllipsoidalOuterBound}, we provide an alternative proof of the  ellipsoid outer bound by working in a transformed space and leveraging Schur-convexity. Although the ellipsoid bounds include the spherical bounds as special cases (Remark \ref{remark:LeiLeoRecoverLsiLso}) and are harder to prove, many of the results, and more importantly, the proof techniques in this section are similar to those already presented in \S\ref{sec:siAndso}, and to avoid redundancy we have left out all the proofs in this section; they are available in their entirety in \cite[Chapter 2.6]{Xie2014}.

\subsection{Simplification of parameter space}
\label{ssec:simplification}
We consider open ellipsoids $\Emc$ of the form (\cite{BoyVan2004}):
\begin{equation}
\label{eq:ellipsoidDef}
\Emc = \left\{ \xbf : (\xbf - \cbf)^{\Tsf} \Rbf^{-1} (\xbf - \cbf) < 1 \right\}.
\end{equation}
Here $\cbf$ is the center of the ellipsoid and the $n \times n$ symmetric and positive definite matrix $\Rbf$ has the spectral decomposition $\Rbf = \Qbf \Dbf \Qbf^{\Tsf}$ where $\Qbf = [\qbf_1 \cdots \qbf_n]$ is orthonormal and holds the eigenvectors of $\Rbf$ (which are the directions of the $n$ axes of the ellipsoid), and $\Dbf = \mathrm{diag}(\lambda_{1},\ldots,\lambda_{n})$ holds the eigenvalues of $\Rbf$.  Each $a_i = \sqrt{\lambda_{i}}$ is the semi-axis length  in the direction $\qbf_i$.  Denote the boundary of $\Emc$ by $\partial \Emc = \left\{ \xbf : (\xbf - \cbf)^{\Tsf} \Rbf^{-1} (\xbf - \cbf) = 1 \right\}$.

Our approach is to approximate the surface $\partial \Lambda$ with part of the surface of an ellipsoid, and then form inner and outer bounds on $\Lambda$ by subtracting these ellipsoids from the unit simplex $\Smc$.  This is the same approach that was used in constructing the spherical bounds (\S \ref{sec:siAndso}).  More concretely, we want to find inner and outer bounding ellipsoids $\Emc_{\rm in}, \Emc_{\rm out}$ such that $\Lei \subseteq \Lambda \subseteq \Leo$, where $\Lei \equiv \Smc \setminus \Emc_{\rm in}, ~\Leo \equiv \Smc \setminus \Emc_{\rm out}$.  

Although we are not able to characterize them further, we define the optimal inner and outer bounding ellipsoids:
\begin{IEEEeqnarray}{rCl}
 \Emc_{\rm in}^{*} 
 & \equiv & \argmin_{\Emc_{\rm in} : \overline{\Emc}_{\rm in} \cap \Smc \supseteq \overline{\Lambda^{c}} \cap \Smc} \text{vol}(\Emc_{\rm in} \cap \Smc) \nonumber \\
& = & \argmin_{\Emc_{\rm in} : \overline{\Emc}_{\rm in} \cap \Smc \supseteq \partial \Lambda} \text{vol}(\Emc_{\rm in} \cap \Smc) 
\label{eqn:Eistar} \\
\Emc_{\rm out}^{*} 
& \equiv & \argmin_{\Emc_{\rm out} : \Emc_{\rm out} \cap \Lambda = \emptyset} \text{vol}(\Smc \setminus \Emc_{\rm out}) \nonumber \\
& = & \argmax_{\Emc_{\rm out} : \Emc_{\rm out} \cap \Lambda = \emptyset} \text{vol}(\Emc_{\rm out} \cap \Smc), 
\label{eqn:Eostar}
\end{IEEEeqnarray}
where the second equality in \eqref{eqn:Eistar} follows from a lemma used in the proof of Prop.\ \ref{prop:innerboundReducedtoBoundaryChecking}.

A result in convex geometry states that for any convex body there exists a unique maximum (resp. minimum) volume inscribed (resp. circumscribing) ellipsoid, called the \textit{L\"{o}wner-John ellipsoid}.  Although we have a convex body $\Lambda^{c} \cap \Smc$, our objective is {\em not} to identify an inscribed/circumscribing ellipsoid with extremized volume for this set.  Rather, our figure of merit (in \eqref{eqn:Eistar} and \eqref{eqn:Eostar}) is to extremize the volume of the \textit{intersection} between the ellipsoid and the simplex.  For example, our $\Emc_{\rm out}$ need not lie \textit{entirely} within the convex body.

In general, analytical characterization of the L\"{o}wner-John ellipsoid is hard (see e.g., \cite{GulGur2007}  \cite[\S 8.4]{BoyVan2004}). One constructive result, though, is that the L\"{o}wner-John ellipsoid is an \textit{invariant ellipsoid}, meaning it inherits \textit{all} the symmetries of the convex body \cite{GulGur2007}. The intuition is that if there were some symmetry that the volume optimal ellipsoid is not endowed with, then using that particular symmetry one can construct another distinct volume optimal ellipsoid, hence contradicting the uniqueness of the L\"{o}wner-John ellipsoid.

In the spirit of the above result, we restrict our attention to ellipsoids that inherit all the symmetries of the convex body $\overline{\Lambda^{c}} \cap \Smc$. Note $\overline{\Lambda^{c}} \cap \Smc$, $\Lambda$ and $\Smc$ all have full permutation symmetry. Prop.\ \ref{prop:InvariantEllipsoid} below states some consequences of inheriting this permutation symmetry. Its proof uses the following three lemmas. 

\begin{lemma} \label{lem:EllipsoidReflectingHyperplane}
Fix an ellipsoid $\Emc = \left\{ \xbf: (\xbf-\cbf)^{\Tsf}\Rbf^{-1}(\xbf - \cbf) < 1 \right\}$ in $\Rbb^{n}$ with center $\cbf$.  A hyperplane passing through $\cbf$ with its normal vector being one of the axes/eigenvectors of $\Rbf$ is a reflecting hyperplane for $\Emc$, i.e.,   $\Emc$ is symmetric w.r.t.\ this hyperplane.  Conversely, the normal vector of any reflecting hyperplane of $\Emc$ can be considered as an axis/eigenvector of $\Rbf$.
\end{lemma}

Note that in the context of a full-dimensional ellipsoid, the axis (direction), eigenvector, and reflecting hyperplane's normal vector are all essentially the same thing.

\begin{lemma} 
\label{lem:distinctEigenvalues}
For a symmetric matrix $\Rbf$, if the two eigenvalues associated with two of the  eigenvectors of $\Rbf$ are distinct, then these two eigenvectors must necessarily be orthogonal (not just linearly independent).
\end{lemma}

\begin{lemma} \label{lem:eigensubspace}
The linear combination of some eigenvectors associated with the same eigenvalue is also an eigenvector (with the same eigenvalue).  In fact, all such eigenvectors are in the same eigen-subspace.
\end{lemma}

\begin{proposition} \label{prop:InvariantEllipsoid}
The class of ellipsoids invariant under permutations of the coordinate axes is the set of ellipsoids parameterized by $(c,a_1,a_2) \in \Rbb_+^3$ with the properties that
\begin{enumerate}
\itemsep=-2pt
\item the center is at $\cbf = c \mathbf{1}$
\item one axis is along the all-rates-equal ray with direction $\mathbf{1}$ and has semi-axis length $a_1$ \item the $n-1$ remaining axes are arbitrary (provided they, together with the axis aligned with $\mathbf{1}$, form an orthonormal set) and have common semi-axis lengths $a_2$. 
\end{enumerate} 
\end{proposition}

Lem.\ \ref{lem:invarianceofR} below gives an explicit construction for $\Rbf^{-1}$, and its Cor.\ \ref{cor:passingei} allows us to characterize the ellipsoid that passes through $\{\ebf_i\}_{i=1}^n$.  Finally, Lem.\ \ref{lem:commontangency} gives an expression that must be satisfied in order for $\partial \Emc$ and $\partial \Lambda$ to share a common tangent point.

\begin{lemma} \label{lem:invarianceofR}
For any ellipsoid in the form of \eqref{eq:ellipsoidDef}, if $\qbf_1 = \frac{1}{\sqrt{n}} \mathbf{1}$ and $\Dbf = \mathrm{diag}(a_1^2,a_2^2,\ldots,a_2^2)$, then $\Rbf^{-1} = \zeta \mathbf{1}_{n \times n} + a_2^{-2} \mathbf{I}_{n \times n}$, where $\zeta \equiv \frac{1}{n} \left(a_1^{-2}-a_2^{-2} \right)$, $\mathbf{1}_{n \times n}$ is an $n \times n$ matrix with each element being $1$ and $\Ibf_{n \times n}$ is the $n \times n$ identity matrix.
\end{lemma}

\begin{corollary} \label{cor:passingei}
For any $c > 1/n$ and $a_1 > \sqrt{n}(c-1/n)$, setting $a_2^2$ as below ensures $\ebf_i \in \partial \Emc$ for $i \in [n]$
\begin{equation}
a_2^2 = \frac{(n-1) a_1^2}{n a_1^2  - (nc-1)^2}.
\end{equation}
\end{corollary}

\begin{lemma} 
\label{lem:commontangency}
Define $\bar{\pbf}_{t} \equiv \mathbf{1} - \pbf_{t}$.  Then $\partial \Emc$ and $\partial \Lambda$ share a point of tangency at $\xbf_t = \xbf(\pbf_t)$ if for each $i \in [n-1]$:
\begin{equation}
\left( \frac{a_{2}}{a_{1}}\right)^{2} = \frac{(\bar{p}_{t,i}-\bar{p}_{t,n}) \Sigma_{j=1}^{n} x_{t,j} + n (\bar{p}_{t,n} x_{t,i} - \bar{p}_{t,i} x_{t,n})}{(\bar{p}_{t,i}-\bar{p}_{t,n})(\Sigma_{j=1}^{n} x_{t,j} - nc)}.
\end{equation}
\end{lemma}

Because of the symmetries present in $\Emc(c,a_{1},a_{2})$, if we use a hyperplane with normal vector $\mathbf{1}$ to ``slice'' $\Emc(c,a_{1},a_{2})$ we get an $(n-1)$-dimensional ball. This is particularly intuitive in light of the following Prop.\ \ref{prop:LambdaSymInRotatedSystem} for which we need to introduce the concept of rigid rotation of coordinate system. 

\begin{definition}[rigid rotation of coordinate system]
\label{def:rigidrotation}
Denote the original coordinate system as $\Xbf$. A \textit{rigid rotation} of the coordinate system about the origin is specified by two vectors $\vbf_{s}$ and $\vbf_{t}$ such that after the rotation $\vbf_{s}$ overlaps with $\vbf_{t}$ while during the rotation the relative position of $\vbf_{s}$ w.r.t.\ the coordinate axes remains unchanged. Denote this rotated coordinate system as $\Ubf$, in which $\vbf_{s}$ was denoted (before rotation) as $\vbf_{t}$ in the original system $\Xbf$. Alternatively, a rigid rotation is specified by a \textit{rotation matrix} $\Mbf_{\rm rot}$ (that can be determined by the starting and terminating vectors $\vbf_{s}, \vbf_{t}$ \cite{Mor2001}) satisfying $\vbf_{t} = \Mbf_{\rm rot} \vbf_{s}$. All rotation matrices are orthonormal.
\end{definition}

According to this definition, we perform a rigid rotation of the original coordinate system $\Xbf$ such that $\ebf_{1}$ overlaps with $ \mathbf{1} / \sqrt{n}$ (i.e., $\vbf_{s} = \ebf_{1}$, $\vbf_{t} = \mathbf{1} / \sqrt{n}$), we then shift the origin to $\cbf = c \mathbf{1}$. Denote this rotated and translated system as $\Ubf$, then the two systems are related by $\Xbf = \Qbf \Ubf + \cbf$ where $\Qbf$ is the associated rotation matrix. 

\begin{proposition} \label{prop:LambdaSymInRotatedSystem}
In the above rotated and translated coordinate system $\Ubf$, $\Lambda$ has permutation symmetry among coordinates $\{2,\ldots,n\}$.  If we let $\Tbf_{u}(\cdot)$ denote this transformation from $\Xbf$ to $\Ubf$, namely $\ubf = \Tbf_{u}(\xbf_{}) = \Qbf^{-1}(\xbf_{} - \cbf)$, then:
\begin{IEEEeqnarray}{rCl}
\IEEEeqnarraymulticol{3}{l}
{\ubf =  (u_{1}, u_{2}, \ldots, u_{n})^{\Tsf} \in \Tbf_{u}(\Lambda) \iff 
}\nonumber \\* ~~
& & \sigma(\ubf) = (u_{1}, u_{\sigma(2)}, \ldots, u_{\sigma(n)})^{\Tsf}  \in \Tbf_{u}(\Lambda), \forall \sigma \in S_{[n] \setminus \{1\}},\IEEEeqnarraynumspace
\label{eq:GoalLambdaSymInRotatedSystem}
\end{IEEEeqnarray}
where $S_{[n]\setminus\{1\}}$ is the permutation group whose members permute coordinate indices from $\{2,\ldots,n\}$ arbitrarily.
\end{proposition}

The previous proposition seems to suggest some loss of permutation symmetry. This is answered in the negative by the following proposition.

\begin{proposition}[conservation of symmetry]
\label{prop:conservationOfsymmetry}
The permutation symmetry of the set $\Lambda$ is preserved under rotation of the coordinate system.
\end{proposition}

The following two propositions comprise the second major contribution in this subsection. They justify why, for both $\Emc_{\rm out}$ (Prop.\ \ref{prop:outerboundReducedtoBoundaryChecking}) and $\Emc_{\rm in}$ (Prop.\ \ref{prop:innerboundReducedtoBoundaryChecking}), it suffices (necessity is clear) to only consider $\partial\Lambda$ in verifying an ellipsoid induces a valid inner or outer bound on $\Lambda$. This observation reduces the set of points that must be checked from $\xbf \in \Lambda$ to $\xbf \in \partial\Lambda$.  Furthermore, the characterization of $\partial \Lambda$ is amenable to analysis using majorization inequalities, which enables us to provide an alternative proof of the proposed ellipsoid outer bound in \S \ref{ssec:alternateProofEllipsoidalOuterBound}.

\begin{proposition} 
\label{prop:outerboundReducedtoBoundaryChecking}
Fix $c > \frac{1}{n}\left(1-\frac{1}{n}\right)^{n-1}$.  Then \\ $\Lambda \subseteq \Leo \iff \partial \Lambda \cap \Emc_{\rm out} = \emptyset$.
\end{proposition}

\begin{proposition} 
\label{prop:innerboundReducedtoBoundaryChecking}
Fix $c > \frac{1}{n}\left(1-\frac{1}{n}\right)^{n-1}$.  Then \\ $\Lei \subseteq \Lambda \iff \partial \Lambda \subseteq \overline{\Emc_{\rm in}}$.
\end{proposition}

The following proposition concludes this subsection; it further reduces the search space of points that must be checked to establish the correctness of a bound on $\Lambda$ induced by an ellipsoid $\Emc(c,a_1,a_2)$.  Recall $\Vmc(\pbf) = \{ a \in (0,1] : \exists i \in [n] : p_i = a\}$ as the set of non-zero values taken by a $\pbf \in \partial\Smc$.  Further define $\Pmc^{(2)} = \{ \pbf \in \partial\Smc : |\Vmc(\pbf)| \in \{1,2\}\}$ as those $\pbf$ taking at most two distinct non-zero values, and $\Pmc^{(1)} = \{ \pbf \in \partial\Smc : |\Vmc(\pbf)| = 1\}$ as those $\pbf$ taking exactly one non-zero value; we refer to $\Pmc^{(1)}$ as the set of quasi-uniform (QU) vectors.  The following proposition reduces the search space from $\partial\Lambda$ to $\{\xbf(\pbf) : \pbf \in \Pmc^{(2)}\}$.  

\begin{proposition} 
\label{prop:atMostTwoComponentValues}
Fix an ellipsoid $\Emc(c,a_{1},a_{2})$.  A potential extremizer of the maximization or minimization problem  
\begin{equation}
\label{eq:essec1objf}
\max_{\pbf \in \partial \Smc} ~ (\text{or} \min_{\pbf \in \partial \Smc})~ f (\xbf(\pbf)) \equiv (\xbf(\pbf) - \cbf)^{\Tsf} \Rbf^{-1} (\xbf(\pbf) - \cbf)
\end{equation}
can have \textit{at most two} distinct values among all its non-zero component(s), i.e., $\pbf^{*} \in \Pmc^{(2)}$.  
\end{proposition}

\subsection{Explicit ellipsoid induced bounds}
\label{ssec:derivationAndproofEllipsoidalBounds}
The key contributions of this subsection are Props.\ \ref{prop:Ei} and \ref{prop:Eo}, which leverage the results from the previous subsection to provide an explicit construction for ellipsoid induced inner and outer bounds on $\Lambda$.  As with the proofs of the spherical inner and outer bounds, $\Lsi$ and $\Lso$, the key proof technique we use is by exploiting the implications of Karush-Kuhn-Tucker first order necessary conditions.  In the next subsection (\S \ref{ssec:alternateProofEllipsoidalOuterBound}) we establish the ellipsoid outer bound using a different proof technique.

\begin{proposition}[ellipsoid outer bound]
\label{prop:Eo}
The ellipsoid $\Emc_{\rm out}(c, a_{1,\mathrm{out}}(c),a_{2,\mathrm{out}}(c))$ (in the class of $\Emc(c, a_1,a_2)$ ellipsoids) with 
\begin{IEEEeqnarray}{rCl}
& & a_{1,\mathrm{out}}(c) = \sqrt{(nc-1) c} \nonumber \\
& & a_{2,\mathrm{out}}(c) = \sqrt{\frac{(n-1)}{n a_{1,\mathrm{out}}^{2} - (nc-1)^{2}}} a_{1,\mathrm{out}} = \sqrt{(n-1)c}\IEEEeqnarraynumspace
\label{eqn:Eo}
\end{IEEEeqnarray}
induces an outer bound $\Leo = \Smc \setminus \Emc_{\rm out}$ on $\Lambda$, for all $c > 1/n$. 
\end{proposition}

\begin{proposition}[ellipsoid inner bound]
\label{prop:Ei}
The ellipsoid $\Emc_{\rm in}(c, a_{1,\mathrm{in}}(c),a_{2,\mathrm{in}}(c))$ (in the class of $\Emc(c, a_1,a_2)$ ellipsoids) with 
\begin{IEEEeqnarray}{rCl}
& & a_{1,\mathrm{in}}(c) = \sqrt{n} (c-m) \nonumber \\
& & a_{2,\mathrm{in}}(c) = \sqrt{\frac{(n-1)}{n a_{1,\mathrm{in}}^{2} - (nc-1)^{2}}} a_{1,\mathrm{in}}\IEEEeqnarraynumspace
\label{eqn:Ei}
\end{IEEEeqnarray}
induces an inner bound $\Lei = \Smc \setminus \Emc_{\rm in}$ on $\Lambda$, for all $c > 1/n$. 
\end{proposition}

\begin{remark} \label{remark:LeiLeoRecoverLsiLso}
The ellipsoid inner bound $\Lei$ recovers the optimal spherical inner bound $\Lsi^{*}$ by setting $c$ to be the critical $c^{*}_{\rm in} = (1 - n m^{2})/(2(1 - n m)$ (note $c^{*}_{\rm in} > 1/n$ for all $n \geq 2$); the ellipsoid outer bound $\Leo$ recovers the optimal spherical outer bound $\Lso^{*}$ by setting $c = c^{*}_{\rm out} = 1$. Furthermore, as a consequence of the tightness monotonicity Prop.\ \ref{prop:tightnessLeiLeo}, any $\Lei$ with $c > c^{*}_{\rm in}$ is better than $\Lsi^{*}$ and any $\Leo$ with $c > 1$ is better than $\Lso^{*}$.
\end{remark}

In words, $\Emc_{\rm in}$ is such that its boundary $\partial \Emc_{\rm in}$ passes through $\ebf_1,\ldots,\ebf_n$ and the all-rates-equal point $\mbf$.  It can be shown that $\partial \Emc_{\rm in}$ is tangent at $\mbf$ with $\partial \Lambda$ (recall Lem.\ \ref{lem:commontangency}). $\Emc_{\rm out}$ is such that its boundary $\partial \Emc_{\rm out}$ passes through each $\ebf_i$ and is also tangent with $\partial \Lambda$ at each $\ebf_i$ (recall Lem.\ \ref{lem:commontangency}).

The ellipsoid bounds $\Lei$, $\Leo$ together with $\partial \Lambda$ are shown in Fig.\ \ref{fig:LeiLeo-Combined-Marked} for $n=2$ and $3$, where the center is chosen to be $c=2$. Improvements can be seen by comparing this figure with the optimal spherical bounds shown in Fig.\ \ref{fig:LsiLso-Combined-Marked-ver20140802} (\S\ref{sec:siAndso}). Furthermore, the quality of the ellipsoid bounds are improved by increasing $c$, as stated in the following proposition (illustrated in Fig.\ \ref{fig:LeiLeo-Combined-Marked-TightnessMonotonictyWRTc-Neq2}).

\begin{figure}[!ht]
\centering
\includegraphics[width=0.4\textwidth]{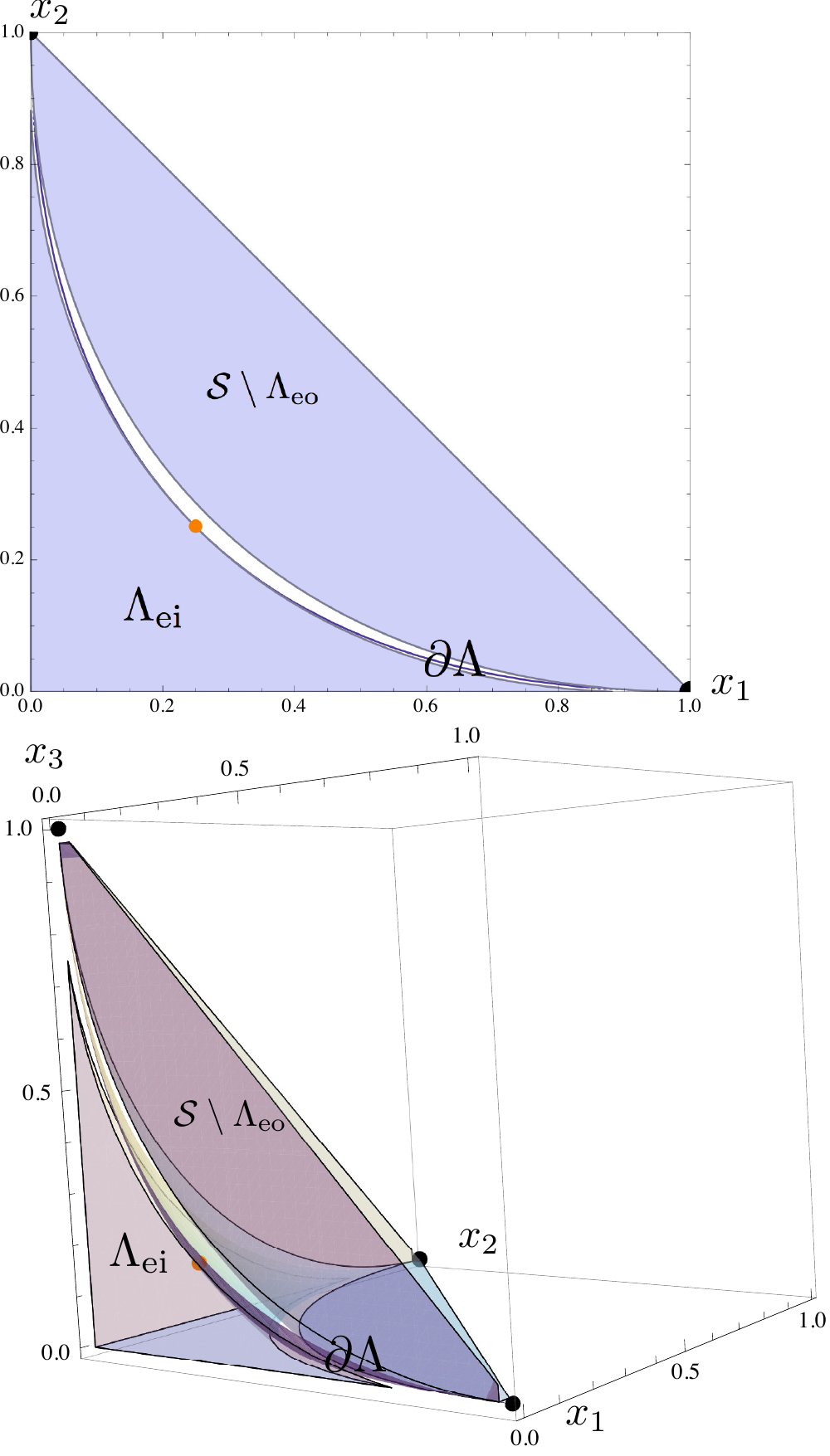}
\caption{Ellipsoid bounds for $n=2$ (top) and $3$ (bottom) when $c=2$: inner bound and the complement (w.r.t.\ $\Smc$) of outer bound are shown in solid; in between sits $\partial \Lambda$. Also shown are the all-rates-equal point $\mbf$ (orange) and corner points $\ebf_{i}$'s (black).}
\label{fig:LeiLeo-Combined-Marked}
\end{figure} 

\begin{proposition} \label{prop:tightnessLeiLeo}
For all $c > 1/n$, the tightness of both $\Lei$ and $\Leo$ increases in $c$.
\end{proposition}

\begin{figure}[!h]
\centering
\includegraphics[width=0.47\textwidth]{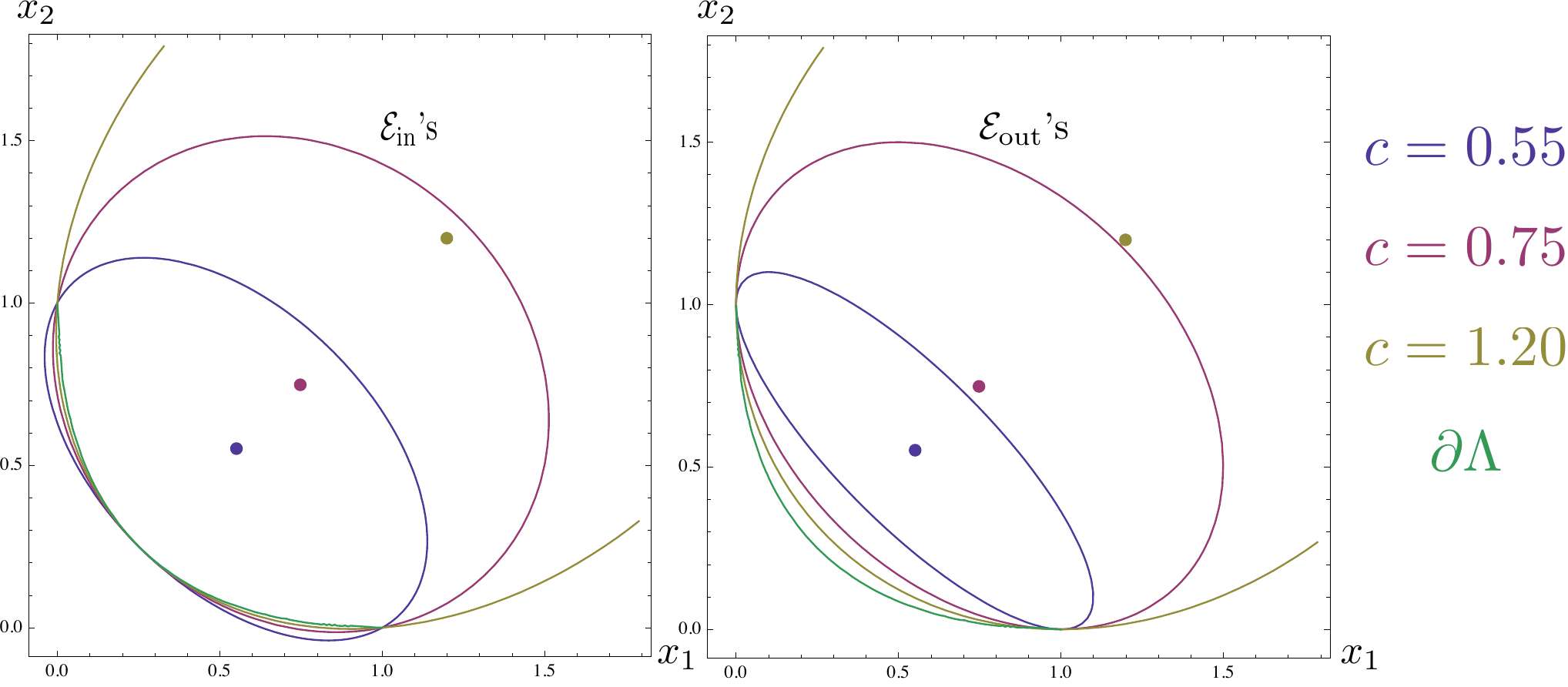}
\caption{Illustration (for $n=2$) of the fact that tightness of both $\Lei$ (left) and $\Leo$ (right) increases in $c$ for $c > 1/n$. For $\Lei$ the set inclusion relationship is reversed once one goes beyond $\Smc$, whereas for $\Leo$ there is a complete set inclusion relationship among this family of ellipsoids. Also shown are $\partial \Lambda$ and the center of each ellipsoid.}
\label{fig:LeiLeo-Combined-Marked-TightnessMonotonictyWRTc-Neq2}
\end{figure}

Thus $\partial \Lambda$ is increasingly tightly ``sandwiched'' between part of $\partial \Emc_{\rm in}$ and $\partial \Emc_{\rm out}$ as $c \rightarrow \infty$.  This sandwiching is asymptotically tight at the all-rates-equal point $\mbf$ (and by construction always tight at each $\ebf_i$).  Specifically, the ellipsoids $\Emc_{\rm in},\Emc_{\rm out}$ viewed in the limit as $c \to \infty$ have axes ratios given by 
\begin{IEEEeqnarray}{rCl}
\lim_{c \to \infty} \frac{a_{1,\mathrm{in}}(n,c)}{a_{1,\mathrm{out}}(n,c)} & = & 1 \nonumber \\
\lim_{c \to \infty} \frac{a_{2,\mathrm{in}}(n,c)}{a_{2,\mathrm{out}}(n,c)} & = & \left(2\left(1-\left(1-\frac{1}{n}\right)^{n-1}\right)\right)^{-\frac{1}{2}}\IEEEeqnarraynumspace
\end{IEEEeqnarray}
for each $n$, where the $a_2$ axes ratio is a monotone decreasing function of $n$, starting  with value $1$ at $n=2$.  Thus when $n=2$, the $\Emc_{\rm in},\Emc_{\rm out}$ are asymptotically equal as $c \to \infty$.  It follows that the asymptotic (in $n$ and $c$) $a_2$ axes ratio is
\begin{equation}
\lim_{n \to \infty} \lim_{c \to \infty} \frac{a_{2,\mathrm{in}}(n,c)}{a_{2,\mathrm{out}}(n,c)} = \sqrt{\frac{\erm}{2(\erm-1)}} \approx 0.8894.
\end{equation}

\subsection{An alternative proof of the outer bound}
\label{ssec:alternateProofEllipsoidalOuterBound}
This subsection supplies an alternative approach to proving the outer bound in Prop.\ \ref{prop:Eo}. It originates from this observation: membership testing for a ball is even simpler than that for an ellipsoid in that the Euclidean distance to the center of the ball is the sole indicator of set membership. Also, recall our ellipsoid parameterized by $(c,a_{1},a_{2})$ is already spherical in the subspace spanned by its $2^{\rm nd}, \ldots, n^{\rm th}$ axes, so the required transformation converting the ellipsoid to a ball is expected to be simple, too. Of course, ultimately we care about bounding $\Lambda$ rather than membership testing for ellipsoid. The gap is filled in by the following lemma, which uses Prop.\ \ref{prop:outerboundReducedtoBoundaryChecking}.

Let $\Tbf_{\mathrm{out}}$ be an affine transformation that transforms $\Emc_{\rm out}$ (with center $\cbf = c \mathbf{1}$ and $c > \frac{1}{n}\left(1-\frac{1}{n}\right)^{n-1}$) to the unit ball at the origin, i.e., $\Tbf_{\mathrm{out}} \left( \Emc_{\rm out} \right) = \Bmc \left( \obf, 1 \right)$.  For any $\xbf_{t} \in \partial \Lambda$ denote by $\Hmc_{t} = \Hmc_t(\xbf_t)$ the hyperplane tangent to $\Lambda$ at $\xbf_t$.  In the transformed space, this hyperplane is denoted $\hat{\Hmc}_{t} = \Tbf_{\mathrm{out}} (\Hmc_{t})$. 
\begin{lemma}
\label{lem:EO_necessaryAndsufficientTestingCdnsUsingTransformations}
Under the above transformation,  
\begin{equation} 
\label{eq:lemTransformApproachForwardPart}
d(\hat{\Hmc}_{t}, \obf) \geq 1 \text{ holds for \textit{all} }\xbf_{t} \in \partial \Lambda ~ \Rightarrow ~  \Lambda \subseteq \Leo.
\end{equation}
If $\Emc_{\rm out}$ is confined within $\Rbb_{+}^{n}$, then we also have a converse, namely
\begin{IEEEeqnarray}{rCl}
\IEEEeqnarraymulticol{3}{l}
{\text{If } \Emc_{\rm out} \subseteq \Rbb_{+}^{n}, \text{ then } 
}\nonumber \\* \quad
& & \Lambda \subseteq \Leo ~ \Rightarrow ~ d(\hat{\Hmc}_{t}, \obf) \geq 1 \text{ holds for \textit{all} }\xbf_{t} \in \partial \Lambda.\IEEEeqnarraynumspace
\label{eq:lemTransformApproachConversePart}
\end{IEEEeqnarray}
\end{lemma}

\begin{remark}
The intuition behind Lem.\ \ref{lem:EO_necessaryAndsufficientTestingCdnsUsingTransformations} is that the intersection between any tangent hyperplane of $\partial \Lambda$ and $\mathbb{R}_{+}^{n}$ is included in $\Lambda$ (proof similar to part of the proof of Prop.\ \ref{prop:Lambdapi} in \S\ref{sec:piAndpo}), hence the intersection is disjoint from (or at most tangent with) any $\overline{\Emc_{\rm out}}$. If further $\Emc_{\rm out}$ is entirely confined within $\mathbb{R}_{+}^{n}$, then for any $\xbf_t \in \partial\Lambda$, any part of the associated tangent hyperplane $\Hmc_t(\xbf_t)$ beyond $\mathbb{R}_{+}^{n}$ cannot touch $\Emc_{\rm out}$. Consequently, the entire tangent hyperplane is disjoint from (or at most tangent with) such $\overline{\Emc_{\rm out}}$.
\end{remark}

The following proposition (Prop.\ \ref{prop:necessaryAndsufficientTestingCdnsUsingTransformations}) will be used in our alternative proof of the $\Leo$ bound.  It is essentially an application of Lem.\ \ref{lem:EO_necessaryAndsufficientTestingCdnsUsingTransformations}; one additional element is the conversion between two forms of an ellipsoid namely the affine transformation of a unit ball centered at the origin $\Emc = \left\{\xbf:  \left. \xbf = \cbf + \Abf \vbf ~ \right| ~ \| \vbf \| < 1  \right\}$ and the (quadratic) form given in \eqref{eq:ellipsoidDef}. For notational convenience, we define the following functions, assuming $(c, a_{1}, a_{2})$ is given.
\begin{definition}
\label{def:gAndfEoAlternateProof}
\begin{IEEEeqnarray}{rCl}
g(\pbf) & \equiv & \frac{(n-1)^{2}}{n}a_{1}^{2} + a_{2}^{2}\sum_{i=2}^{n}\left( \frac{1}{\sqrt{n}} - p_{i} -\frac{1}{1+\sqrt{n}} (1-p_{1}) \right)^{2}, \nonumber \\
f(\pbf)  & \equiv & \sqrt{g(\pbf)} + \pi (\pbf); \nonumber \\
\tilde{g}(\pbf) & \equiv & \frac{(n-1)^{2}}{n}a_{1}^{2} + \left(-\frac{1}{n} + \sum_{i=1}^{n} p_{i}^{2} \right)a_{2}^{2}, \nonumber \\
\tilde{f}(\pbf)  & \equiv & \sqrt{\tilde{g}(\pbf)} + \pi (\pbf).  
\end{IEEEeqnarray}
\end{definition}
It can be verified that
\begin{equation} \label{eq:g2gtilde}
g(\pbf) =  \tilde{g}(\pbf), ~\text{ if } ~ \sum_{i=1}^{n} p_{i} =1.
\end{equation}

\begin{proposition}
\label{prop:necessaryAndsufficientTestingCdnsUsingTransformations}
Fix an ellipsoid $\Emc(c,a_{1},a_{2})$. Define $\Lambda_{\rm e} \equiv \Smc \setminus \Emc$. Assume $c > \frac{1}{n}(1-\frac{1}{n})^{n-1}$. We have
\begin{equation} 
f (\pbf) \leq c(n-1) \text{ for \textit{all} } \pbf \in \partial \Smc ~ \Rightarrow ~  \Lambda \subseteq \Lambda_{\rm e}.
\end{equation}
Conversely,
\begin{IEEEeqnarray}{rCl}
\IEEEeqnarraymulticol{3}{l}
{\text{If } \Emc \subseteq \Rbb_{+}^{n}, \text{ then }
}\nonumber \\* \quad
& &  \Lambda \subseteq \Lambda_{\rm e}  \Rightarrow f (\pbf) \leq c(n-1) \text{ for \textit{all} } \pbf \in \partial \Smc.\IEEEeqnarraynumspace
\end{IEEEeqnarray}
\end{proposition}

Prop.\ \ref{prop:necessaryAndsufficientTestingCdnsUsingTransformations} asks us to check the inequality $f(\pbf) \leq c (n-1)$ holds for all $\pbf \in \partial \Smc$. As before, we formulate this as a constrained optimization problem. Yet Schur-convexity now makes our job a lot easier. For this we need to state an equivalent condition for verifying Schur-convexity.
\begin{proposition}[\cite{MarOlk2011}, Ch.\ 3, Thm.\ A.4] \label{prop:SchurCondition}
Let $\Dmc \subseteq \Rbb$ be an open interval and let a symmetric function $f: \Dmc^{n} \to \Rbb$ be continuously differentiable. Then $f$ is Schur-convex on $\Dmc^{n}$ if and only if for all distinct indices $k, l \in [n]$:
\begin{equation} \label{eq:SchurCondition}
(x_{k} - x_{l})\left( \frac{\partial f}{\partial x_{k}} - \frac{\partial f}{\partial x_{l}} \right) \geq 0, ~\forall x \in \Dmc^{n}.
\end{equation}
\end{proposition}

We now apply Prop.\ \ref{prop:necessaryAndsufficientTestingCdnsUsingTransformations}, meaning we need to solve
\begin{equation} \label{eq:OptimizationFormulationTildef}
\max_{\pbf \in \partial \Smc} ~ \tilde{f}(\pbf) = \sqrt{\tilde{g}(\pbf)} + \pi (\pbf),
\end{equation}
where $g(\pbf)$ ($f(\pbf)$) is replaced by $\tilde{g}(\pbf)$ ($\tilde{f}(\pbf)$) due to \eqref{eq:g2gtilde}. Note $\tilde{f}(\pbf)$ is symmetric.

Algebra gives:
\begin{IEEEeqnarray}{rCl}
\IEEEeqnarraymulticol{3}{l}
{(p_{k} - p_{l}) \left( \frac{\partial \tilde{f}}{\partial p_{k}} - \frac{\partial \tilde{f}}{\partial p_{l}}\right)
}\nonumber \\* \quad
& = &  (p_{k} - p_{l})^{2} \left( \frac{a_{2}^{2}}{\sqrt{\tilde{g}}}  - \frac{\pi(\pbf)}{(1-p_{k})(1-p_{l})} \right).\IEEEeqnarraynumspace
\label{eqn:SchurConditionSpecializedToEo}
\end{IEEEeqnarray}
In order to show the RHS of \eqref{eqn:SchurConditionSpecializedToEo} is non-negative, it suffices to show $\frac{a_{2}^{2}}{\sqrt{\tilde{g}}} \geq 1$, for which we specialize with the expressions for $\Emc_{\rm out}$ given in \eqref{eqn:Eo} and get:
\begin{IEEEeqnarray}{rCl}
\left( \frac{a_{2,\mathrm{out}}^{2}}{\sqrt{\tilde{g}}} \right)^{2}  & = & \frac{(n-1)c}{nc - 1 - c + \sum_{i=1}^{n} p_{i}^{2}} \nonumber \\
& \geq & \frac{(n-1)c}{nc - 1 - c + \sum_{i=1}^{n} p_{i}} = 1.
\label{eqn:EOSchurConvexity}
\end{IEEEeqnarray}
Therefore, Prop.\ \ref{prop:SchurCondition} tells $\tilde{f}$ in \eqref{eq:OptimizationFormulationTildef} specialized to our $\Emc_{\rm out}$ is Schur-convex. It then follows from the definition of Schur-convexity and the fact $\ebf_{i}$ majorizes every other point in the feasible set $\partial \Smc$ that the \textit{global} maximum of $\tilde{f}$ is attained at each $\ebf_{i}$. Finally, evaluating $\tilde{f}$ at $\ebf_{i}$ gives $c(n-1)$, the desired global maximum (according to Prop.\ \ref{prop:necessaryAndsufficientTestingCdnsUsingTransformations}). This then completes our alternative proof that our proposed ellipsoid outer bound is valid.

\begin{remark}
The Schur-convexity approach would not be directly applicable to the proposed ellipsoid inner bound in Prop.\ \ref{prop:Ei}. Because even if we could have a parallel result to Prop.\ \ref{prop:necessaryAndsufficientTestingCdnsUsingTransformations} for inner bounding ellipsoids, the fact that our inner bounding ellipsoid in \eqref{eqn:Ei} is tight at $\mbf$ as well as at $\ebf_{i}$ precludes Schur-convexity, since $\mbf$ is strictly majorized by every other point and $\ebf_{i}$ majorizes every other point.
\end{remark}

\section{Generalized convexity properties} 
\label{sec:genconv}

Given an arrival rate vector $\xbf = (x_{1}, \ldots, x_{n})$, define the set of stabilizing controls $\Pmc(\xbf)$ assuming the worst-case service rate:
\begin{equation}
\label{eq:t}
\Pmc(\xbf) \equiv \left\{ \pbf \in [0,1]^n : x_{i} \leq p_i \prod_{j \neq i} (1-p_j), ~ i \in [n] \right\}.
\end{equation}
This set is related to $\Lambda$ in that $\Pmc(\xbf) = \emptyset$ if and only if $\xbf \not \in \Lambda$.  Whereas $\Lambda$ is an important inner bound on the Aloha stability region $\Lambda_{\rm A}$, the set $\Pmc(\xbf)$ can be viewed as the set of control options given a desired arrival rate vector $\xbf$, which is important since knowing the system is stablizable is useful only if one also knows how it may be stabilized. To proceed, we now view $\xbf$ as parameters (instead of $\pbf$, as was done previously) and define an ``excess rate'' (since the term $p_i \prod_{j \neq i} (1-p_j)$ is the worst-case service rate for user $i$) function for each $i \in [n]$:
\begin{equation} 
\label{eq:excessRateFunction}
f_i(\pbf) \equiv x_i - p_i \prod_{j \neq i} (1-p_j), ~ i \in [n].
\end{equation}
The excess rate functions are closely related to the set of stabilizing controls.  To see this, define the $\alpha$-sublevel sets of each excess rate function 
\begin{equation}
S_{i,\alpha} \equiv \{ \pbf  \in [0,1]^n : f_i(\pbf) \leq \alpha \}, ~ \alpha \in \Rbb.
\end{equation}
Denote by $S_{\alpha} = \cap_{i=1}^n S_{i,\alpha}$, then $\Pmc(\xbf) = S_0 = \cap_{i=1}^n S_{i,0}$. In words, the set of stabilizing controls associated with $\Lambda$ is the intersection of $0$-sublevel sets of these $n$ excess rate functions.

One motivation for studying the generalized convexity properties of the excess rate functions is that, given $\xbf$, it is natural to seek the stabilizing control to maximize some norm of the difference between $\xbf$ and the worst-case service rate achieved by $\pbf$, over all controls $\pbf \in \Pmc(\xbf)$.  By analogy, the expected delay in an $M/M/1$ queue with arrival rate $\lambda$ and service rate $\mu$ is $1/(\mu-\lambda)$, and thus delay is minimized for a given $\lambda$ by maximizing $\mu$.  The convexity properties of the excess rate function may be applied to this optimization problem.

For simplicity, throughout this section we will work with the open convex domain $(0,1)^{n}$. The following propositions show these excess rate functions are not convex (Prop.\ \ref{prop:excRateFuncNotConv}), but are quasiconvex (Prop.\  \ref{prop:excRateFuncQuasiConv}), pseudoconvex (Prop.\ \ref{prop:excRateFuncPseudoConv}), and invex (Prop.\ \ref{prop:excRateFuncInvex}). 

There are important implications to be observed. For example, since the set $\Pmc(\xbf)$ is convex (Cor.\ \ref{cor:convexityofStabilizingControls}), we then know $i)$ testing membership in $\Lambda$ of a given rate vector $\xbf$ can be posed as a convex program, $ii)$ any convex combination of (timesharing between) two stabilizing controls is itself a stabilizing control, and $iii)$ the problem of finding a ``nearest'' control under a change in rate vector is a convex program.  That is, suppose the arrival rate $\xbf$ is stabilized by $\pbf \in \Pmc(\xbf)$, and a change in arrival rate to $\xbf'$ necessitates a change in control to some $\pbf'$.  To minimize the magnitude of the change in the control we may seek to minimize $\|\pbf - \pbf'\|$ over $\pbf' \in \Pmc(\xbf')$; the convexity of $\Pmc(\xbf)$ ensures this problem is the convex program of projecting a point $\pbf$ onto a set $\Pmc(\xbf')$ (see e.g., \cite[\S 8.1]{BoyVan2004}).   

\begin{proposition} 
\label{prop:excRateFuncNotConv}
For all $i \in [n]$, the excess rate function $f_i(\pbf)$ in \eqref{eq:excessRateFunction} is not convex on $(0,1)^{n}$.
\end{proposition}

\begin{IEEEproof}
In fact we can work with an additional constraint $\sum_{i=1}^{n} p_{i} = 1$, i.e., with $\partial \Smc$ being the domain. W.l.o.g.\ let us show this for $f_{1}(\pbf) = x_1- p_1 \prod_{j \neq 1} (1-p_j)$. Let $\epsilon \in (0,1)$ be determined later. Denote $\pbf_{1,\epsilon} = (1-\epsilon, \epsilon/(n-1), \ldots, \epsilon/(n-1))$, $\pbf_{2} = (1/n) \mathbf{1}$. Form the convex combination $\pbf_{\theta, \epsilon} = \theta \pbf_{1, \epsilon} + (1-\theta)\pbf_{2}$ for $\theta \in (0,1)$. We shall show for all $n$ the existence of $\epsilon$ and $\theta$ in order for the following inequality to hold:
\begin{equation}
f_{1}(\pbf_{\theta, \epsilon}) > \theta f_{1}(\pbf_{1, \epsilon}) + (1 - \theta) f_{1}(\pbf_{2}).
\end{equation}
After substituting the definitions and setting $\theta = 1/2$, the above inequality becomes:
\begin{IEEEeqnarray}{rCl}
\IEEEeqnarraymulticol{3}{l}
{(1 - \epsilon + \frac{1}{n}) \left( 1 - \frac{n\left(\epsilon + 1 \right) - 1}{2 n \left(n - 1\right)}\right)^{n-1}
}\nonumber \\* \quad
& < & \left( 1 - \epsilon \right) \left( 1 - \frac{\epsilon}{n-1} \right)^{n-1} + \frac{1}{n} \left( 1 - \frac{1}{n} \right)^{n-1}. \IEEEeqnarraynumspace
\end{IEEEeqnarray}
Manipulation of the above equation gives the equivalent form
\begin{equation}
\sqrt[n-1]{\frac{1 - \epsilon + \frac{1}{n}}{\left( 1 - \epsilon \right) \left( 1 - \frac{\epsilon}{n-1} \right)^{n-1} + \frac{1}{n} \left( 1 - \frac{1}{n} \right)^{n-1}}} < \frac{1}{1 - \frac{n\left(\epsilon + 1 \right) - 1}{2 n \left(n - 1\right)}}.
\end{equation}
Applying the AM-GM inequality to the above LHS:
\begin{IEEEeqnarray}{rCl}
\IEEEeqnarraymulticol{3}{l}
{\sqrt[n-1]{\frac{1 - \epsilon + \frac{1}{n}}{\left( 1 - \epsilon \right) \left( 1 - \frac{\epsilon}{n-1} \right)^{n-1} + \frac{1}{n} \left( 1 - \frac{1}{n} \right)^{n-1}} \cdot \underbrace{1 \cdots 1}_{\#: n-2}}
}\nonumber \\* \quad \quad \quad \quad  \quad
& \leq & \frac{\frac{1 - \epsilon + \frac{1}{n}}{\left( 1 - \epsilon \right) \left( 1 - \frac{\epsilon}{n-1} \right)^{n-1} + \frac{1}{n} \left( 1 - \frac{1}{n} \right)^{n-1}}  + n-2}{n-1}.\IEEEeqnarraynumspace
\end{IEEEeqnarray}
It can be easily verified that the sequence $\left( 1 - \frac{1}{n} \right)^{n-1}$ monotonically decreases to $1/\erm$, and that $\left( 1 - \frac{\epsilon}{n-1} \right)^{n-2}$ monotonically decreases to $1/\erm^{\epsilon}$.  Because of this, it suffices to show
\begin{equation}
\frac{\frac{1 - \epsilon + \frac{1}{n}}{\left( 1 - \epsilon \right) \left( 1 - \frac{\epsilon}{n-1} \right)\frac{1}{\erm^{\epsilon}} + \frac{1}{n} \frac{1}{\erm}}  + n-2}{n-1} < \frac{1}{1 - \frac{n\left(\epsilon + 1 \right) - 1}{2 n \left(n - 1\right)}},
\end{equation}
which after rearrangement becomes
\begin{equation} \label{eq:notConvCondition}
\frac{1 - \epsilon + \frac{1}{n}}{\left( 1 - \epsilon \right) \left( 1 - \frac{\epsilon}{n-1} \right)\frac{1}{\erm^{\epsilon}} + \frac{1}{n} \frac{1}{\erm}}   - \frac{3 n - 2}{2 n -1} < 0.
\end{equation}
Denote $h_{1}(n, \epsilon) = \frac{1 - \epsilon + \frac{1}{n}}{\left( 1 - \epsilon \right) \left( 1 - \frac{\epsilon}{n-1} \right)\frac{1}{\erm^{\epsilon}} + \frac{1}{n} \frac{1}{\erm}}$ and $h_{2}(n) = \frac{3 n - 2}{2 n -1}$. One can verify that, given $\epsilon \in (0,1)$, $h_{1}(n, \epsilon)$ and $h_{2}(n)$ are monotone decreasing and increasing in $n$ ($n \geq 2$), respectively. Thus it suffices to show \eqref{eq:notConvCondition} holds when $n = 2$. Observe that given $n$, the LHS of \eqref{eq:notConvCondition} i.e., $h_{1}(n,\epsilon) - h_{2}(n)$ is a continuous function of $\epsilon$ for $\epsilon \in (0,1)$, since $h_{1}(2, 0) - h_{2}(2) < 0$, $h_{1}(2,1) - h_{2}(2) > 0$, there exist various choices of $\epsilon$ so that \eqref{eq:notConvCondition} holds for $n=2$, and hence for all $n \geq 2$ (with the same choice of $\epsilon$).
\end{IEEEproof}

Recall a function is called quasiconvex (or unimodal) if its domain and all its sublevel sets are convex \cite{BoyVan2004}.

\begin{proposition} \label{prop:excRateFuncQuasiConv}
The excess rate function is quasiconvex on $(0,1)^{n}$ for all $i \in [n]$.
\end{proposition}

\begin{IEEEproof}
Our approach is to show the convexity of the sublevel sets $S_{i, \alpha}$ for which we discuss two cases: $i)$ $\alpha \geq x_{i}$, and $ii)$ $\alpha < x_{i}$.  Consider case $i)$. If $\alpha \geq x_{i}$ then $f_{i}(\pbf) \leq \alpha$ always holds, and therefore $S_{i, \alpha} = \mathrm{dom} f = (0,1)^{n}$, which is convex.  It remains to consider case $ii)$, with $\alpha < x_{i}$.  Construct $g_{i}(\pbf_{\setminus i}) \equiv \frac{x_{i} - \alpha}{\prod_{j \neq i}(1-p_{j})}$, where $\pbf_{\setminus i}$ is formed from the vector $\pbf$ by dropping component $i$. Observe
\begin{IEEEeqnarray}{rCl}
S_{i,\alpha} & = & \{ \pbf  \in (0,1)^n : f_i(\pbf) \leq \alpha \}  \nonumber \\
& = &  \{ (\pbf_{\setminus i}, p_{i})  \in (0,1)^n : g_i(\pbf_{\setminus i}) \leq p_{i} \}.
\end{IEEEeqnarray}
Thus $S_{i,\alpha}$ can be interpreted as the \textit{epigraph} of the function $g_{i}(\pbf_{\setminus i})$. Since a function is convex iff its epigraph is a convex set, we then need to show this function $g_{i}(\pbf_{\setminus i})$ is convex. Toward this, we take the logarithm and write:
\begin{IEEEeqnarray}{rCl}
\log g_{i}(\pbf_{\setminus i})  & = & \log (x_{i} - \alpha) - \sum_{j \neq i} \log (1-p_{j}) \nonumber \\
& = & \log (x_{i} - \alpha) + \sum_{j \neq i} -\log \left(1- \left(\ebf_{j}\right)_{\setminus i}^{\Tsf} \pbf_{\setminus i}\right),\IEEEeqnarraynumspace
\end{IEEEeqnarray}
where $\left(\ebf_{j}\right)_{\setminus i}$ is the $(n-1)$-vector by peeling off the $i^{\rm th}$ component of $\ebf_{j}$.
This shows the RHS of the above equation is convex by recognizing the convexity of $- \log(\cdot)$ and certain function compositions that preserve convexity. Finally since the function $g_{i}(\pbf_{\setminus i})$ is log-convex, this means $g_{i}(\pbf_{\setminus i})$ is itself convex, which means its epigraph, or equivalently the sublevel set $S_{i, \alpha}$, is convex.
\end{IEEEproof}

\begin{corollary}\label{cor:convexityofStabilizingControls}
Recall $\Pmc(\xbf) = S_{0} = \cap_{i=1}^n S_{i,0}$. It follows that the set of stabilizing controls $\Pmc(\xbf)$ associated with $\xbf \in \Lambda$ is convex, as convexity is preserved under set intersection.
\end{corollary}

\begin{proposition} \label{prop:excRateFuncPseudoConv}
The excess rate function is pseudoconvex on $(0,1)^{n}$ for all $i \in [n]$.
\end{proposition}
\begin{IEEEproof}
We appeal to Thm.\ 3.2.6 of Cambini and Martein \cite{CamMar2009}, which essentially says for a differentiable function defined on an open convex set, a quasiconvex function is pseudoconvex when there are no critical points. We can compute the gradient of $f_{i}(\pbf)$:
\begin{IEEEeqnarray}{rCl}
\frac{\partial}{\partial p_k} f_i(p) & = & \left\{ \begin{array}{ll}
\frac{f_i(p) - x_{i}}{p_i}, \; & k = i \\
-\frac{f_i(p) - x_{i}}{1-p_k}, \; & k \neq i 
\end{array} \right., i \in [n]   \nonumber \\
& = & \left\{ \begin{array}{ll}
-\prod_{j \neq i}(1-p_{j}), \; & k = i \\
p_{i} \prod_{j \neq i,k}(1-p_{j}), \; & k \neq i 
\end{array} \right.,  i \in [n].\IEEEeqnarraynumspace
\label{eq:gradient-of-f}
\end{IEEEeqnarray}
It is clear that there does not exist any critical point in the open convex domain $(0,1)^{n}$. The pseudoconvexity of the excess rate function then follows from Prop.\ \ref{prop:excRateFuncQuasiConv}.
\end{IEEEproof}
\begin{remark}
In \cite{Xie2014} the proof of pseudoconvexity was done by verifying a second-order condition.
\end{remark}

\begin{proposition} \label{prop:excRateFuncInvex}
The excess rate function is invex on $(0,1)^{n}$ for all $i \in [n]$.
\end{proposition}
\begin{IEEEproof}
For differentiable $f$ with open convex domain, $f$ is invex if and only if every stationary point is a global minimizer (see e.g., Thm.\ 4.9.1 of Cambini and Martein \cite{CamMar2009}). Next, Thm.\ 2.27 of Mishra and Giorgi \cite{MisGio2008} says if $f$ is differentiable and quasiconvex with open convex domain, then $f$ is pseudoconvex if and only if every stationary point is a global minimizer. These two results mean that under the assumption of quasiconvexity, invexity and pseudoconvexity coincide. Thus the invexity of our excess rate functions follows from Props.\ \ref{prop:excRateFuncQuasiConv} and \ref{prop:excRateFuncPseudoConv}. 
\end{IEEEproof}

\section{Conclusion} 
\label{sec:conclusion} 
The stability region of the slotted Aloha medium access control protocol remains unknown for $n > 2$ users. In this paper we have made some progress towards understanding an important inner bound on this unknown region, namely, the set $\Lambda$. Specifically, since a naive way of testing membership in this inner bound calls for identifying an auxiliary variable whose support is uncountably infinite, we use both algebraic (\S\ref{sec:rootTesting}) and geometric (\S\ref{sec:piAndpo}, \S\ref{sec:siAndso}, and \S\ref{sec:eiAndeo}) approaches to perform exact or approximate membership testing that overcomes the aforementioned disadvantage. Collectively, our non-parametric region bounds shed light on the geometric structure of $\Lambda$ (and hence on the stability region itself). We have also established relationships between an arrival rate vector and its stabilizing control(s); in particular, given a stabilizable arrival rate vector, we have characterized all the ``critical'' stabilizing control(s) via our root testing, and established the convexity of the set of stabilizing controls via the generalized convexity properties of the ``excess rate'' functions. Several of our results are constructive and are applicable for improved implementation and operation of the slotted Aloha protocol.

A natural extension of our work is to identify different classes of $\Emc(c,a_1,a_2)$ ellipsoids that induce tighter bounds on $\Lambda$. A second extension is to seek a combinatorial expression for the volume of $\Lambda$ that is simpler and easier to compute than the one we have derived.

\section{Acknowledgement}
\label{sec:acknowledgement}
We would like to thank $i)$ Prof.\ Hugo J.\ Woerdeman (Drexel University Department of Mathematics) for showing us Lem.\ \ref{lem:distinctEigenvalues} and for helpful discussions on optimization, $ii)$ Meisam Razaviyayn (University of Minnesota) for suggesting a helpful reference on integrating polynomials over the simplex, and $iii)$ the two anonymous reviewers for their helpful comments. We express particular thanks to the anonymous reviewer who provided detailed comments on the material in \S\ref{sec:rootTesting}.

\bibliographystyle{IEEEtran}
\bibliography{TranIT2014-AlohaBib}

\begin{thebibliography}{10}
\providecommand{\url}[1]{#1}
\def\UrlFont{\rmfamily}
\providecommand{\newblock}{\relax}
\providecommand{\bibinfo}[2]{#2}
\providecommand\BIBentrySTDinterwordspacing{\spaceskip=0pt\relax}
\providecommand\BIBentryALTinterwordstretchfactor{4}
\providecommand\BIBentryALTinterwordspacing{\spaceskip=\fontdimen2\font plus
\BIBentryALTinterwordstretchfactor\fontdimen3\font minus
  \fontdimen4\font\relax}
\providecommand\BIBforeignlanguage[2]{{%
\expandafter\ifx\csname l@#1\endcsname\relax
\typeout{** WARNING: IEEEtran.bst: No hyphenation pattern has been}%
\typeout{** loaded for the language `#1'. Using the pattern for}%
\typeout{** the default language instead.}%
\else
\language=\csname l@#1\endcsname
\fi
#2}}

\bibitem{XieWeb2010}
N.~Xie and S.~Weber, ``Geometric approximations of some aloha-like stability
  regions,'' in \emph{IEEE International Symposium on Information Theory
  ({ISIT})}, Austin, TX, USA, June 2010.

\bibitem{Abr1970}
N.~Abramson, ``The {ALOHA} system: another alternative for computer
  communications,'' in \emph{Proceedings of the fall joint computer
  conference}.\hskip 1em plus 0.5em minus 0.4em\relax ACM, 1970, pp. 281--285.

\bibitem{RaoEph1988}
R.~Rao and A.~Ephremides, ``On the stability of interacting queues in a
  multiple--access system,'' \emph{{IEEE} Transactions on Information Theory},
  vol.~34, no.~5, pp. 918--930, September 1988.

\bibitem{LuoEph1999}
W.~Luo and A.~Ephremides, ``Stability of $n$ interacting queues in
  random--access systems,'' \emph{{IEEE} Transactions on Information Theory},
  vol.~45, no.~5, pp. 1579--1587, July 1999.

\bibitem{KomMaz2013}
S.~Kompalli and R.~Mazumdar, ``On the stability of finite queue slotted-{A}loha
  protocol,'' \emph{{IEEE} Transactions on Information Theory}, vol.~59,
  no.~10, pp. 6357--6366, October 2013.

\bibitem{Abr2009}
N.~Abramson, ``The {ALOHANET} - surfing for wireless data,'' \emph{IEEE
  Communications Magazine}, vol.~47, no.~12, pp. 21--25, 2009.

\bibitem{AbrSac2012}
N.~Abramson, C.~Sacchi, B.~Bellalta, and A.~Vinel, ``Multiple access
  communications in future-generation wireless networks,'' \emph{EURASIP
  Journal on Wireless Communications and Networking}, vol. 2012, no.~1, pp.
  1--4, 2012.

\bibitem{AkyPom2004}
I.~F. Akyildiz, D.~Pompili, and T.~Melodia, ``Challenges for efficient
  communication in underwater acoustic sensor networks,'' \emph{ACM Sigbed
  Review}, vol.~1, no.~2, pp. 3--8, 2004.

\bibitem{SozSto2000}
E.~M. Sozer, M.~Stojanovic, and J.~G. Proakis, ``Underwater acoustic
  networks,'' \emph{IEEE Journal of Oceanic Engineering}, vol.~25, no.~1, pp.
  72--83, 2000.

\bibitem{HeiSto2012}
J.~Heidemann, M.~Stojanovic, and M.~Zorzi, ``Underwater sensor networks:
  applications, advances and challenges,'' \emph{Philosophical Trans. of the
  Royal Society of London A: Mathematical, Physical and Engineering Sciences},
  vol. 370, no. 1958, pp. 158--175, 2012.

\bibitem{Rob1975}
L.~G. Roberts, ``{ALOHA} packet system with and without slots and capture,''
  \emph{{ACM SIGCOMM} Computer Communication Review}, vol.~5, no.~2, pp.
  28--42, April 1975.

\bibitem{Abr1977}
N.~Abramson, ``The throughput of packet broadcasting channels,'' \emph{{IEEE}
  Transactions on Communications}, vol.~25, no.~1, pp. 117--128, January 1977.

\bibitem{TsyMik1979}
B.~Tsybakov and V.~Mikhailov, ``Ergodicity of the slotted {A}loha system,''
  \emph{Problemy Peredachi Informatsii}, vol.~15, no.~4, pp. 73--87, 1979.

\bibitem{Szp1994}
W.~Szpankowski, ``Stability conditions for some distributed systems: buffered
  random access systems,'' \emph{Advances in Applied Probability}, vol.~26,
  no.~2, pp. 498--515, June 1994.

\bibitem{Ana1991}
V.~Anantharam, ``The stability region of the finite-user slotted {A}loha
  protocol,'' \emph{{IEEE} Transactions on Information Theory}, vol.~37, no.~3,
  pp. 535--540, May 1991.

\bibitem{BorMcD2008}
C.~Bordenave, D.~McDonald, and A.~Proutiere, ``Performance of random medium
  access control, an asymptotic approach,'' in \emph{Proceedings of the 2008
  ACM SIGMETRICS International Conference on Measurement and Modeling of
  Computer Systems}.\hskip 1em plus 0.5em minus 0.4em\relax New York, NY, USA:
  ACM, 2008, pp. 1--12.

\bibitem{LuoEph2006}
J.~Luo and A.~Ephremides, ``On the throughput, capacity and stability regions
  of random multiple access,'' \emph{{IEEE} Transactions on Information
  Theory}, vol.~52, no.~6, pp. 2593--2607, June 2006.

\bibitem{SubLei2013}
V.~G. Subramanian and D.~J. Leith, ``On the rate region of {CSMA/CA} {WLAN}s,''
  \emph{{IEEE} Transactions on Information Theory}, vol.~59, no.~6, pp.
  3932--3938, June 2013.

\bibitem{LeiSub2010}
D.~J. Leith, V.~G. Subramanian, and K.~R. Duffy, ``Log-convexity of rate region
  in 802.11 e {WLAN}s,'' \emph{{IEEE} Communications Letters}, vol.~14, no.~1,
  pp. 57--59, Jan. 2010.

\bibitem{GupSto2012}
P.~Gupta and A.~L. Stolyar, ``Throughput region of random-access networks of
  general topology,'' \emph{{IEEE} Transactions on Information Theory},
  vol.~58, no.~5, pp. 3016--3022, May 2012.

\bibitem{MasMat1985}
J.~Massey and P.~Mathys, ``The collision channel without feedback,''
  \emph{{IEEE} Transactions on Information Theory}, vol.~31, no.~2, pp.
  192--204, March 1985.

\bibitem{Pos1985}
K.~Post, ``Convexity of the nonachievable rate region for the collision channel
  without feedback,'' \emph{{IEEE} Transactions on Information Theory},
  vol.~31, no.~2, pp. 205--206, March 1985.

\bibitem{OuyTen2015}
Y.~Ouyang and D.~Teneketzis, ``A common information-based multiple access
  protocol achieving full throughput,'' in \emph{IEEE International Symposium
  on Information Theory ({ISIT})}, Hong Kong, June 2015.

\bibitem{HasLei1996}
J.~H{\aa}stad, T.~Leighton, and B.~Rogoff, ``Analysis of backoff protocols for
  multiple access channels,'' \emph{{SIAM} Journal on Computing}, vol.~25,
  no.~4, pp. 740--774, Aug. 1996.

\bibitem{JiaWal2011}
L.~Jiang and J.~Walrand, ``Approaching throughput-optimality in distributed
  {CSMA} scheduling algorithms with collisions,'' \emph{{IEEE/ACM} Transactions
  on Networking}, vol.~19, no.~3, pp. 816--829, June 2011.

\bibitem{JiaSha2010}
L.~Jiang, D.~Shah, J.~Shin, and J.~Walrand, ``Distributed random access
  algorithm: scheduling and congestion control,'' \emph{{IEEE} Transactions on
  Information Theory}, vol.~56, no.~12, pp. 6182--6207, Dec. 2010.

\bibitem{Avi2000}
D.~Avis, ``lrs: A revised implementation of the reverse search vertex
  enumeration algorithm,'' in \emph{Polytopes -- Combinatorics and
  Computation}, G.~Kalai and G.~Ziegler, Eds.\hskip 1em plus 0.5em minus
  0.4em\relax Birkhauser--Verlag, 2000, pp. 177--198,
  http://cgm.cs.mcgill.ca/\~{}avis/doc/avis/Av98a.ps.

\bibitem{CouVem2016}
B.~Cousins and S.~Vempala, ``A practical volume algorithm,'' \emph{Mathematical
  Programming Computation}, vol.~8, no.~2, pp. 133--160, June 2016.

\bibitem{BerTsi2008}
D.~Bertsekas and J.~Tsitsiklis, \emph{Introduction to Probability},
  2nd~ed.\hskip 1em plus 0.5em minus 0.4em\relax Athena Scientific Press, 2008.

\bibitem{GruMol1978}
A.~Grundmann and H.~M{\"o}ller, ``Invariant integration formulas for the
  $n$-simplex by combinatorial methods,'' \emph{{SIAM} Journal on Numerical
  Analysis}, vol.~15, no.~2, pp. 282--290, April 1978.

\bibitem{Wil1994}
H.~Wilf, \emph{Generatingfunctionology}, 2nd~ed.\hskip 1em plus 0.5em minus
  0.4em\relax San Diego, CA: Academic Press, 1994.

\bibitem{Ell1976}
R.~S. Ellis, ``Volume of an $n$-simplex by multiple integration,''
  \emph{Elemente Der Mathematik}, vol.~31, no.~3, pp. 57--59, 1976.

\bibitem{Bat2008}
N.~Batir, ``Sharp inequalities for factorial $n$,'' \emph{Proyecciones Journal
  of Mathematics}, vol.~27, no.~1, pp. 97--102, May 2008.

\bibitem{Xie2014}
N.~Xie, ``Multiple access stability and broadcast delay in wireless networks,''
  Ph.D. dissertation, Drexel University, August 2014.

\bibitem{BoyVan2004}
S.~Boyd and L.~Vandenberghe, \emph{Convex Optimization}.\hskip 1em plus 0.5em
  minus 0.4em\relax Cambridge University Press, 2004.

\bibitem{GulGur2007}
O.~G{\"u}ler and F.~G{\"u}rtuna, ``The extremal volume ellipsoids of convex
  bodies, their symmetry properties, and their determination in some special
  cases,'' \emph{arXiv:0709.0707v1}, September 2007.

\bibitem{Mor2001}
D.~Mortari, ``On the rigid rotation concept in $n$-dimensional spaces,''
  \emph{Journal of the Astronautical Sciences}, vol.~49, no.~3, July-September
  2001.

\bibitem{MarOlk2011}
A.~W. Marshall, I.~Olkin, and B.~C. Arnold, \emph{Inequalities: Theory of
  Majorization and Its Applications}, 2nd~ed.\hskip 1em plus 0.5em minus
  0.4em\relax Springer, 2011.

\bibitem{CamMar2009}
A.~Cambini and L.~Martein, \emph{Generalized Convexity and Optimization: Theory
  and Applications}, ser. Lecture notes in economics and mathematical
  systems.\hskip 1em plus 0.5em minus 0.4em\relax Springer, 2009, vol. 616.

\bibitem{MisGio2008}
S.~K. Mishra and G.~Giorgi, \emph{Invexity and Optimization}, ser. Nonconvex
  optimization and its applications.\hskip 1em plus 0.5em minus 0.4em\relax
  Springer, 2008, vol.~88.

\end{thebibliography}

\section{Appendix}
\label{sec:appendix}

\subsection{Proof of Prop.\ \ref{prop:augumentedRootTesting}}
\label{ssec:augmentedRootTestingProof}
We first show a supporting lemma and its corollary.  We define a univariate function, parameterized by $\pbf$, $g(t_{\pi}) = g(t_{\pi}, \pbf)$, as
\begin{equation} 
\label{eq:gtpi}
g(t_{\pi}, \pbf)  \equiv  \prod_{i=1}^{n} (1 + p_{i} t_{\pi}) - (1+t_{\pi}), ~ \pbf \in [0,1]^{n}.
\end{equation}

\begin{lemma} \label{lem:gtpi}
For all $n \geq 2$, $g(t_{\pi}, \pbf)$ can only have one or two real roots on $(-1, \infty)$. More specifically:
\begin{itemize}
\item $t_{\pi} = 0$ is always a root, and is the unique root if and only if $\sum_{i=1}^{n} p_{i} = 1$, i.e., $\pbf$ is a probability vector;
\item besides $t_{\pi} = 0$, the other root is on $(0, \infty)$ if and only if $\sum_{i=1}^{n} p_{i} < 1$;
\item besides $t_{\pi} = 0$, the other root is on $(-1, 0)$ if and only if $\sum_{i=1}^{n} p_{i} > 1$.
\end{itemize}
\end{lemma}

\begin{IEEEproof}
Applying the chain rule of differentiation, we have:
\begin{IEEEeqnarray}{rCl}
g' (t_{\pi}) & = & \sum_{i=1}^{n} p_{i} \prod_{j \neq i} (1 + p_{j} t_{\pi}) - 1, \nonumber \\
g'' (t_{\pi}) & = & \sum_{i=1}^{n} p_{i}  \sum_{j \neq i} p_{j} \prod_{ k \neq j, i} (1 + p_{k} t_{\pi}).\IEEEeqnarraynumspace
\end{IEEEeqnarray}
Here are some simple yet important observations: $g(0) = 0$, $g(-1) = \pi (\pbf) > 0$, $g'(0) = \sum_{i=1}^{n} p_{i} - 1$, and $g'(-1) = \sum_{i=1}^{n} p_{i} \prod_{j \neq i} (1-p_{j}) - 1 \leq 0$.  The last inequality is justified since we can construct a vector $\xbf (\pbf) $ according to \eqref{eq:xOfp} in Def.\ \ref{def:fOfdeltaxAndpOfdeltaxAndxOfp}, and then apply the fact $\Lambda_{\rm eq} = \Lambda \subseteq \Smc$.  Furthermore, $g''(0) \geq 0$ and in fact $g''(t_{\pi}) \geq 0$ for all $t_{\pi} > -1$ which means $g'(t_{\pi})$ is monotone increasing on $(-1, \infty)$. In the following we first show the forward part, i.e., how to go from the condition on $\sum_{i=1}^{n} p_{i}$ to the properties of the roots.
\begin{itemize}
\item case 1: $\sum_{i=1}^{n} p_{i} = 1$. In this case, since $g'(0)=0$ and $g''(0) \geq 0$, the stationary point 0 is a local minimizer.  As $g(-1) = \pi (\pbf)$, $g'(-1) \leq 0$, $g'(t_{\pi})$ is monotone increasing on $(-1, \infty)$, so what happens on $(-1, \infty)$ is: $g(t_{\pi})$ is monotone decreasing from $\pi (\pbf)$ ($t_{\pi} = -1$) to 0 ($t_{\pi} = 0$), and then monotone increasing from 0 to $\infty$ (as $t_{\pi} \to \infty$).  Thus the only root on $(-1, \infty)$ is 0.
\item case 2: $\sum_{i=1}^{n} p_{i} < 1$. In this case, $g'(0) < 0$. Thus $g(t_{\pi})$ is decreasing from $\pi (\pbf)$ ($t_{\pi} = -1$) to 0 ($t_{\pi} = 0$), then keeps decreasing until some stationary point $t_{\pi, 2}^{*} > 0$ such that $g'(t_{\pi,2}^{*}) = 0$, after which $g(t_{\pi})$ keeps increasing as $t_{\pi} \to \infty$. Note $g(t_{\pi,2}^{*}) < 0$, so the only other root $\tilde{t}_{\pi, 2}$ is on $(0, \infty)$.
\item case 3: $\sum_{i=1}^{n} p_{i} > 1$. In this case, $g'(0) > 0$. Thus $g(t_{\pi})$ is decreasing from $\pi (\pbf)$ ($t_{\pi} = -1$) to 0 (at some $\tilde{t}_{\pi,3}$), and keeps decreasing until at some stationary point $t_{\pi,3}^{*}$ such that $g'(t_{\pi, 3}^{*}) = 0$, after which $g(t_{\pi})$ keeps increasing as $t_{\pi} \to \infty$. Since $g'(t_{\pi}) > 0$ for all $t_{\pi} > t_{\pi, 3}^{*}$ and recall $g'(t_{\pi})$ itself is monotone increasing, so $t_{\pi, 3}^{*} \in (-1, 0)$ which further implies the other root i.e., $\tilde{t}_{\pi,3}$ has to be on $(-1,0)$.
\end{itemize}
The converses are then clear (proof by contradiction) and are omitted.
\end{IEEEproof}

\begin{corollary} \label{cor:gtpiAndfdelta}
Fix $n \geq 2$ and $\xbf \in \Rbb_{+}^{n}$.  $i)$ The total number of positive roots of $f(\delta, \xbf)$ is at most two. Furthermore, $ii)$ \textit{if} there exists a compatible stabilizing control $\pbf$ with $\xbf$ in the sense of $\Lambda_{\rm eq}$, namely the $(\xbf, \pbf)$ pair satisfies \eqref{eq:xOfp} in Def.\ \ref{def:fOfdeltaxAndpOfdeltaxAndxOfp}, \textit{then} $f(\delta, \xbf)$ has a positive root $\delta_{t}$ if and only if $g(t_{\pi}, \pbf)$ has a root $t_{\pi_{\pbf}} \in (-1, \infty)$, and the roots of $f, g$ can be related by $\delta_{t} = \frac{1}{\pi (\pbf)} (1+t_{\pi_{\pbf}})$. More specifically $f(\delta, \xbf)$ has either one or two positive roots:
\begin{itemize}
\item $\delta_{t} = \frac{1}{\pi (\pbf)}$ is always a positive root, and is the unique root if and only if $\sum_{i=1}^{n} p_{i} = 1$;
\item besides $\delta_{t} = \frac{1}{\pi (\pbf)}$, $f(\delta, \xbf)$ also has a larger positive root if and only if $\sum_{i=1}^{n} p_{i} < 1$;
\item besides $\delta_{t} = \frac{1}{\pi (\pbf)}$, $f(\delta, \xbf)$ also has a smaller positive root if and only if $\sum_{i=1}^{n} p_{i} > 1$.
\end{itemize}
\end{corollary}

\begin{IEEEproof}
We first show $ii)$. Substituting $\delta_{t}$, $\pi (\pbf)$ and $x_{i}(\pbf) =  \frac{p_{i}}{1-p_{i}} \pi (\pbf)$ into the definition of $f(\delta, \xbf)$, we have
\begin{IEEEeqnarray}{rCl}
\IEEEeqnarraymulticol{3}{l}
{f(\delta_{t}, \xbf) 
}\nonumber \\* ~~
& = & \prod_{i=1}^{n} \left(1 + \frac{p_{i}}{1-p_{i}} \pi (\pbf) \frac{1}{\pi (\pbf)} (1+t_{\pi_{\pbf}}) \right) - \frac{1}{\pi (\pbf)} (1+t_{\pi_{\pbf}}) \nonumber \\
& = & \frac{1}{\pi (\pbf)} \left( \prod_{i=1}^{n} (1 + p_{i} t_{\pi_{\pbf}}) - (1+t_{\pi_{\pbf}}) \right) = \frac{g (t_{\pi_{\pbf}}, \pbf)}{\pi (\pbf)} .\IEEEeqnarraynumspace
\end{IEEEeqnarray}
Therefore, given $\xbf$ and its compatible $\pbf$ in the sense of $\Lambda_{\rm eq}$, $\delta_{t}$ is a root of $f(\delta, \xbf)$ if and only if $t_{\pi_{\pbf}}$ is a root of $g(t_{\pi}, \pbf)$. Since $\delta_{t}$ is positive if and only if $t_{\pi_{\pbf}} \in (-1, \infty)$, the statement for the three cases then follows from Lem.\ \ref{lem:gtpi}. 

We next show $i)$. Recall from the proof of Prop.\ \ref{prop:Lambdarootconditions} that if $f(\delta, \xbf)$ ever has a positive root $\tilde{\delta}$, then we can construct $\pbf = \pbf(\tilde{\delta}, \xbf)$ according to \eqref{eq:pOfdeltax} in Def.\ \ref{def:fOfdeltaxAndpOfdeltaxAndxOfp}. Since this $\pbf$ is compatible with $\xbf$ in the sense of $\Lambda_{\rm eq}$ meaning $(\xbf, \pbf(\tilde{\delta}, \xbf))$ satisfies \eqref{eq:xOfp} in Def.\ \ref{def:fOfdeltaxAndpOfdeltaxAndxOfp}, the assertion that $f(\delta, \xbf)$ can have at most two positive roots follows from $ii)$ just proved.
\end{IEEEproof}

We now provide the proof of Prop.\ \ref{prop:augumentedRootTesting}.  

\underline{Proof of part $1)$}

``$\Rightarrow$'':
Given $\xbf \in \partial \Lambda$, from the (preliminary) root testing Prop.\ \ref{prop:Lambdarootconditions} we know there exists a positive root $\delta_{u}$ for $f(\delta, \xbf)$. Form a vector of probabilities $\pbf_{u} = \pbf_{u} (\delta_{u}, \xbf)$ according to \eqref{eq:pOfdeltax} in Def.\ \ref{def:fOfdeltaxAndpOfdeltaxAndxOfp}; it can be verified $(\xbf, \pbf_{u})$ satisfies \eqref{eq:xOfp} in Def.\ \ref{def:fOfdeltaxAndpOfdeltaxAndxOfp} namely $\pbf_{u}$ is compatible with $\xbf$ in the sense of $\Lambda_{\rm eq}$. Massey and Mathys \cite{MasMat1985} established the one-to-one correspondence (i.e., bijection) between $\partial \Smc$ and $\partial \Lambda$, so $\pbf_{u} \in \partial \Smc$ and is the only (critical) stabilizing control for $\xbf$. It can also be verified that $\frac{1}{\pi (\pbf_{u})}$ is a positive root for $f(\delta, \xbf)$. That the root is unique follows from case 1 of $ii)$ in Cor.\ \ref{cor:gtpiAndfdelta}.

``$\Leftarrow$'':
If $f(\delta, \xbf)$ has a unique positive root denoted $\delta_{u}$, again form a vector of probabilities $\pbf_{u} =  \pbf_{u} (\delta_{u}, \xbf)$ according to \eqref{eq:pOfdeltax}, which is compatible with $\xbf$ in the sense of $\Lambda_{\rm eq}$. It follows from Cor.\ \ref{cor:gtpiAndfdelta} that $g(t_{\pi}, \pbf_{u})$ has a unique root on $(-1, \infty)$ denoted $t_{\pi_{\pbf_{u}}} = 0$, and is related to $\delta_{u}$ by $\delta_{u} = \frac{1}{\pi (\pbf_{u})} (1+t_{\pi_{\pbf_{u}}}) = \frac{1}{\pi (\pbf_{u})}$.  Furthermore $\pbf_{u}$ is a probability vector (Lem.\ \ref{lem:gtpi}), which implies $\xbf \in \partial \Lambda$ \cite{MasMat1985} .

\underline{Proof of part $2)$}

Given $\xbf \in \Lambda \setminus \partial \Lambda$, there must exist (Prop.\ \ref{prop:Lambdarootconditions}) a positive root $\delta$ for $f(\delta_{t}, \xbf)$. We use this $\delta$ to construct a vector of probabilities $\pbf = \pbf(\delta, \xbf)$ as in \eqref{eq:pOfdeltax}. Recall then the root can be expressed as $\delta = \frac{1}{\pi (\pbf)}$, and furthermore $(\xbf, \pbf)$ satisfies \eqref{eq:xOfp} in Def.\ \ref{def:fOfdeltaxAndpOfdeltaxAndxOfp}. Now we use $\pbf$ to form $g(t_{\pi}, \pbf)$ and solve for roots on $(-1, \infty)$. From part $1)$ of this proposition and Cor.\ \ref{cor:gtpiAndfdelta} we know there must exist a root (denoted $t'_{\pi}$) other than $0$  on $(-1, \infty)$. Define $\delta' \equiv \frac{1}{\pi (\pbf)} \left( 1 + t'_{\pi}\right)$, with which we further define $\pbf' = \pbf'(\delta', \xbf)$, $\xbf' = \xbf' (\pbf')$ as in Def.\ \ref{def:fOfdeltaxAndpOfdeltaxAndxOfp}. First observe $\delta'$ is a positive root of $f(\delta_{t}, \xbf)$ due to Cor.\ \ref{cor:gtpiAndfdelta} (since $t'_{\pi}$ solves $g(t_{\pi}, \pbf) = 0$). Second we claim $\xbf' = \xbf$, because
\begin{IEEEeqnarray}{rCl}
\pi (\pbf') & = &  \prod_{i=1}^{n} \left( 1 - p'_{i} \right)  \nonumber \\
& = & \prod_{i=1}^{n} \left( 1 - \frac{\delta' x_{i}}{1 + \delta' x_{i}} \right) = \frac{1}{\prod_{i=1}^{n} \left( 1 + \delta' x_{i} \right)} = \frac{1}{\delta'},\IEEEeqnarraynumspace
\label{eq:toShowxprimeEqx}
\end{IEEEeqnarray}
where the last equality follows again from Cor.\ \ref{cor:gtpiAndfdelta}. It can then be verified that $x_{i}' = \frac{p_{i}'}{1 - p_{i}'} \pi (\pbf') = x_{i}$ for all $i$.

Note $\delta' \neq \delta$ (as $t'_{\pi} \neq 0$), and as such we can repeat the above procedure starting from this different positive root $\delta'$. More specifically, define $\pbf'_{r} = \pbf'_{r} (\delta', \xbf')$ as in \eqref{eq:pOfdeltax}. Then the root satisfies $\delta' = \frac{1}{\pi (\pbf'_{r})}$ and furthermore $(\xbf', \pbf'_{r})$ satisfies \eqref{eq:xOfp} in Def.\ \ref{def:fOfdeltaxAndpOfdeltaxAndxOfp}. We now use $\pbf'_{r}$ to form $g(t_{\pi}, \pbf'_{r})$ and solve for the root other than $0$ on $(-1, \infty)$. We claim this root has to be such that it allows us to reconstruct $\pbf$. To see this, denote this root as $t''_{\pi} \in (-1,0) \cup (0, \infty)$. Define $\delta'' \equiv \frac{1}{\pi (\pbf'_{r})} \left( 1 + t''_{\pi}\right)$, with which we further define $\pbf''_{r} = \pbf''_{r} (\delta'', \xbf')$, $\xbf''_{r} = \xbf''_{r} (\pbf''_{r})$ as in Def.\ \ref{def:fOfdeltaxAndpOfdeltaxAndxOfp}. According to Cor.\ \ref{cor:gtpiAndfdelta}, $f(\delta'', \xbf') = \frac{1}{\pi (\pbf'_{r})} g(t''_{\pi}, \pbf'_{r})$, then $\delta''$ is a positive root for $f(\delta_{t}, \xbf') = 0$. Now since $\xbf' = \xbf$ (so $f(\delta_{t}, \xbf')$ and $f(\delta_{t}, \xbf)$ are the same) and furthermore $f$ has exactly two positive roots (since Cor.\ \ref{cor:gtpiAndfdelta} says $f$ can have at most two), it then has to hold that $\delta'' = \delta$ (as $t''_{\pi} \neq 0$ implies $\delta'' \neq \delta'$), which further gives $\pbf''_{r} = \pbf$, $\pi (\pbf''_{r}) = \pi (\pbf)$ and $\xbf''_{r} = \xbf$. This proves the claim. Effectively this says the above procedure (as illustrated in \eqref{eqn:rootTestingDiagram} below) can be reversed.

Furthermore, if $\pbf$ is such that $\sum_{i=1}^{n} p_{i} < 1$, then $\pbf'$ is such that $\sum_{i=1}^{n} p'_{i} > 1$, and vice versa. To show this, assume w.l.o.g.\ $\pbf$ is such that $\sum_{i=1}^{n} p_{i} < 1$. Now for $\delta = \frac{1}{\pi (\pbf)}$, $\delta' = \frac{1}{\pi (\pbf)} \left( 1 + t'_{\pi} \right)$, Lem.\ \ref{lem:gtpi} says if $\sum_{i=1}^{n} p_{i} < 1$ then $t'_{\pi} > 0$ and hence $\delta' > \delta$. From \eqref{eq:toShowxprimeEqx} $\delta'$ can also be expressed as $\frac{1}{\pi (\pbf')}$ (so $\pi (\pbf') < \pi (\pbf)$). Now let us repeat the above procedure as illustrated in diagram \eqref{eqn:rootTestingDiagram} starting from $\delta'$ (rather than $\delta$), we will then get $\delta'' = \frac{1}{\pi (\pbf')}\left( 1 + t''_{\pi} \right)$ where $t''_{\pi}$ is a root on $(-1,0) \cup (0, \infty)$ solved from $g(t_{\pi}, \pbf') = 0$. Since it has to hold that $\delta'' = \delta = \frac{1}{\pi (\pbf)}$, it then follows $t''_{\pi} < 0$ which in turn implies (Lem.\ \ref{lem:gtpi}) $\sum_{i=1}^{n} p'_{i} > 1$.  This shows, as a result of the above procedure one of the two critical stabilizing controls is in $\Smc \setminus \partial \Smc$, the other in $[0,1]^{n} \setminus \Smc$.

Finally we show there can not be more than two critical stabilizing controls for a given $\xbf$. We prove by contradiction. Assume w.l.o.g.\ there are two critical stabilizing controls $\pbf_{s}$ and $\tilde{\pbf}_{s}$ both in $\Smc$ (for the case of $\Smc^{c}$ the proof is similar), since both $\pbf_{s}$ and $\tilde{\pbf}_{s}$ are critical stabilizing controls for the same $\xbf$ it has to hold that $\pi (\pbf_{s}) \neq \pi (\tilde{\pbf}_{s})$ (as otherwise $\pbf_{s}$ would not be distinct from $\tilde{\pbf}_{s}$). As we have just shown above for both $\pbf_{s}$ and $\tilde{\pbf}_{s}$ there is a corresponding vector of probabilities in $[0,1]^{n} \setminus \Smc$ that critically stabilizes $\xbf$, denoted as $\pbf_{s^{c}}$ and $\tilde{\pbf}_{s^{c}}$ respectively. From the previous parts of the proof we see it has to hold that $\frac{1}{\pi (\pbf_{s})} \neq \frac{1}{\pi (\pbf_{s^{c}})}$ and $\frac{1}{\pi (\tilde{\pbf}_{s})} \neq \frac{1}{\pi (\tilde{\pbf}_{s^{c}})}$. This means $f(\delta_{t}, \xbf)$ already has more than two positive roots, which is impossible according to Cor.\ \ref{cor:gtpiAndfdelta} (recall with the exceptions of $\ebf_{i}$'s any critical stabilizing control $\pbf$ gives a positive root of $f$ as $\frac{1}{\pi (\pbf)}$).
\begin{IEEEeqnarray}{rCl} 
\label{eqn:rootTestingDiagram}
& & \substack{\text{given} \\ \xbf \in \Lambda \setminus \partial \Lambda} \stackrel{\substack{\text{solve} \\ f(\delta_{t}, \xbf) = 0}}{\longrightarrow} \substack{\text{root } \\ \delta > 0} \longrightarrow \substack{\text{set } \pbf = \pbf(\delta, \xbf) \\ \text{ as in Def.}\ \ref{def:fOfdeltaxAndpOfdeltaxAndxOfp} } \longrightarrow \substack{\text{observe } \\ \delta = \frac{1}{\pi (\pbf)}, ~ \xbf = \xbf(\pbf) \\  \text{ as in Def.}\ \ref{def:fOfdeltaxAndpOfdeltaxAndxOfp}} \nonumber \\
& & \stackrel{\substack{\text{solve} \\ g(t_{\pi}, \pbf) = 0}}{\longrightarrow} \substack{\text{root } \\ t'_{\pi} \in (-1, 0) \cup (0, \infty)}\longrightarrow \substack{\text{set} \\ \delta' = \frac{1}{\pi (\pbf)} \left( 1 + t'_{\pi} \right)} \nonumber \\
& &  \longrightarrow \substack{\text{root } \\ \delta' > 0}       \longrightarrow \substack{\text{set } \pbf' = \pbf'(\delta', \xbf), ~\xbf' = \xbf'(\pbf') \\ \text{ as in Def.\ \ref{def:fOfdeltaxAndpOfdeltaxAndxOfp}}} \longrightarrow \substack{ \text{observe } \\ \xbf' = \xbf, ~\delta' = \frac{1}{\pi (\pbf')} } \IEEEeqnarraynumspace
\end{IEEEeqnarray}

We use the following example to demonstrate the above process in both directions. 
\begin{example}
For $n =2$, let $\xbf = \left( 1/4, 1/5 \right)$. Solving $f(\delta_{t}, \xbf) = 0$ gives two positive roots $\delta = (11-\sqrt{41})/2 \approx 2.29844$, $\delta' = (11+\sqrt{41})/2 \approx 8.70156$. These two roots are also shown in Fig.\ \ref{fig:RootTest} (the green curve).
\begin{itemize}
\item If we start from $\delta$, then $\pbf \equiv \pbf(\delta, \xbf) = \left( \frac{1}{40} \left(21-\sqrt{41}\right),  \frac{1}{40} \left(19-\sqrt{41}\right) \right)$ $\approx (0.364922, 0.314922)$. Solving $g(t_{\pi}, \pbf) = 0$ yields the non-zero root as $t_{\pi}' = \frac{1}{40} \left(41+11 \sqrt{41}\right) \approx 2.78586$ with which it can be verified that $\delta_{t}' = \frac{1}{\pi (\pbf)} \left( 1+t_{\pi}'\right)$ equals $\delta'$.  We can also verify $\delta = \frac{1}{\pi (\pbf)}$.
\item Now if we start from $\delta'$, then $\pbf' \equiv \pbf(\delta', \xbf) = \left(\frac{1}{40} \left(21+\sqrt{41}\right), \frac{1}{40} \left(19+\sqrt{41}\right) \right)$ $\approx (0.685078, 0.635078)$. Solving $g(t_{\pi}'', \pbf') = 0$ yields the non-zero root as $t_{\pi}'' = \frac{1}{40} \left(41-11 \sqrt{41}\right)$ $\approx -0.735859$ with which it can be verified that $\delta_{t}'' = \frac{1}{\pi (\pbf')} \left( 1+t_{\pi}'' \right)$ equals $\delta$. We can also verify $\delta' = \frac{1}{\pi (\pbf')}$.
\end{itemize}
\end{example}

\begin{IEEEbiographynophoto}{Nan Xie}
(S'10-M'14) received his B.S. degree in communication engineering and M.S. degree in circuits and systems, both from Wuhan University, China, in 2003 and 2006, respectively. From 2006 to 2007 he was with the Department of Electrical and Computer Engineering at the University of Florida, Gainesville, FL, USA. Since 2008 he has been with the Department of Electrical and Computer Engineering at Drexel University, Philadelphia, PA, USA, where he earned his Ph.D.\ degree in electrical engineering in 2014. His research interests are focused on a better understanding of metrics such as delay, stability, throughput, and fairness, in the context of design and performance optimization of networks.
\end{IEEEbiographynophoto}

\begin{IEEEbiographynophoto}{John MacLaren Walsh} (S'01-M'07) received the B.S. (\emph{magna cum laude}), M.S. and Ph.D. degrees in electrical and computer engineering from Cornell University, Ithaca, NY in 2002, 2004, and 2006, respectively.  In September 2006, he joined the Department of Electrical and Computer Engineering at Drexel University, Philadelphia PA, where he is currently an Associate Professor.  At Drexel, he directs the Adaptive Signal Processing and Information Theory Research Group. 
\end{IEEEbiographynophoto}

\begin{IEEEbiographynophoto}{Steven Weber}
(S'97-M'03-SM'11) received the B.S. degree in 1996 from Marquette University, Milwaukee, WI, USA, in 1996 and the M.S. and Ph.D.\ degrees from The University of Texas at Austin, TX, USA, in 1999 and 2003, respectively.  He joined the Department of Electrical and Computer Engineering at Drexel University, Philadelphia, PA, USA, in 2003, where he is currently a Professor.  His research interests are centered around mathematical modeling of computer and communication networks.
\end{IEEEbiographynophoto}

\end{document}